\documentclass[11pt]{article}

\usepackage{mathtools}
\usepackage{amsthm}
\usepackage{graphicx} 
\usepackage{array} 

\usepackage{amsmath, amssymb, amsfonts, verbatim}
\usepackage{hyphenat,epsfig,multirow}
\usepackage{nicefrac}
\usepackage{paralist}

\usepackage[usenames,dvipsnames]{xcolor}
\usepackage{algorithm, setspace}

\DeclareFontFamily{U}{mathx}{\hyphenchar\font45}
\DeclareFontShape{U}{mathx}{m}{n}{
      <5> <6> <7> <8> <9> <10>
      <10.95> <12> <14.4> <17.28> <20.74> <24.88>
      mathx10
      }{}
\DeclareSymbolFont{mathx}{U}{mathx}{m}{n}
\DeclareMathSymbol{\bigtimes}{1}{mathx}{"91}

\usepackage{tcolorbox}
\tcbuselibrary{skins,breakable}
\tcbset{enhanced jigsaw}

\usepackage[normalem]{ulem}
\usepackage[compact]{titlesec}

\definecolor{DarkRed}{rgb}{0.5,0.1,0.1}
\definecolor{DarkBlue}{rgb}{0.1,0.1,0.5}

\usepackage{nameref}
\definecolor{ForestGreen}{rgb}{0.1333,0.5451,0.1333}
\definecolor{Red}{rgb}{0.9,0,0}
\usepackage[linktocpage=true,
	pagebackref=true,colorlinks,
	linkcolor=DarkRed,citecolor=ForestGreen,
	bookmarks,bookmarksopen,bookmarksnumbered]
	{hyperref}
\usepackage[noabbrev,nameinlink]{cleveref}
\crefname{property}{property}{Property}
\creflabelformat{property}{(#1)#2#3}
\crefname{equation}{eq}{Eq}
\creflabelformat{equation}{(#1)#2#3}

\usepackage{bm}
\usepackage{url}
\usepackage{xspace}
\usepackage[mathscr]{euscript}

\usepackage{tikz}
\usetikzlibrary{arrows}
\usetikzlibrary{arrows.meta}
\usetikzlibrary{shapes}
\usetikzlibrary{backgrounds}
\usetikzlibrary{positioning}
\usetikzlibrary{decorations.markings}
\usetikzlibrary{patterns}
\usetikzlibrary{calc}
\usetikzlibrary{fit}
\usetikzlibrary{snakes}

\usepackage{mdframed}

\usepackage[noend]{algpseudocode}
\makeatletter
\def\BState{\State\hskip-\ALG@thistlm}
\makeatother

\usepackage{cite}
\usepackage{enumitem}

\usepackage{subcaption}

\usepackage[margin=1in]{geometry}

\newtheorem{theorem}{Theorem}
\newtheorem{lemma}{Lemma}[section]
\newtheorem{proposition}[lemma]{Proposition}

\newtheorem{claim}[lemma]{Claim}
\newtheorem{fact}[lemma]{Fact}

\newtheorem{definition}[lemma]{Definition}

\newtheorem{problem}{Problem}

\newtheorem*{claim*}{Claim}
\newtheorem*{proposition*}{Proposition}
\newtheorem*{lemma*}{Lemma}
\newtheorem*{problem*}{Problem}

\crefname{lemma}{Lemma}{Lemmas}
\crefname{claim}{Claim}{Claims}

\newtheorem{mdresult}{Result}
\newenvironment{result}{\begin{mdframed}[backgroundcolor=lightgray!40,topline=false,rightline=false,leftline=false,bottomline=false,innertopmargin=2pt]\begin{mdresult}}{\end{mdresult}\end{mdframed}}

\newtheorem{remark}[lemma]{Remark}

\newtheoremstyle{restate}{}{}{\itshape}{}{\bfseries}{~(restated).}{.5em}{\thmnote{#3}}
\theoremstyle{restate}
\newtheorem*{restate}{}

\makeatletter
\patchcmd{\ALG@step}{\addtocounter{ALG@line}{1}}{\refstepcounter{ALG@line}}{}{}
\newcommand{\ALG@lineautorefname}{Line}
\makeatother

\allowdisplaybreaks

\renewcommand{\qed}{\nobreak \ifvmode \relax \else
      \ifdim\lastskip<1.5em \hskip-\lastskip
      \hskip1.5em plus0em minus0.5em \fi \nobreak
      \vrule height0.75em width0.5em depth0.25em\fi}

\newcommand{\tvd}[2]{\ensuremath{\norm{#1 - #2}_{\mathrm{tvd}}}}
\newcommand{\Ot}{\ensuremath{\widetilde{O}}}
\newcommand{\eps}{\ensuremath{\varepsilon}}

\newcommand{\Bracket}[1]{\Big[#1\Big]}
\newcommand{\bracket}[1]{\left[#1\right]}
\newcommand{\paren}[1]{\ensuremath{\left(#1\right)}\xspace}
\newcommand{\card}[1]{\left\lvert{#1}\right\rvert}

\newcommand{\norm}[1]{\ensuremath{\|#1\|}}

\newcommand{\expect}[1]{\Exp\bracket{#1}}
\newcommand{\var}[1]{\textnormal{Var}\bracket{#1}}
\newcommand{\cov}[1]{\textnormal{Cov}\bracket{#1}}

\newcommand{\set}[1]{\ensuremath{\left\{ #1 \right\}}}
\newcommand{\poly}{\mbox{\rm poly}}
\newcommand{\polylog}{\mbox{\rm  polylog}}

\newcommand{\alg}{\ensuremath{\mathcal{A}}\xspace}

\DeclareMathOperator*{\Exp}{\ensuremath{{\mathbb{E}}}}
\DeclareMathOperator*{\Prob}{\ensuremath{\textnormal{Pr}}}
\renewcommand{\Pr}{\Prob}

\newcommand{\Ex}{\Exp}

\newcommand{\etal}{{\it et al.\,}}

\newenvironment{tbox}{\begin{tcolorbox}[
		enlarge top by=5pt,
		enlarge bottom by=5pt,
		 breakable,
		 boxsep=0pt,
                  left=4pt,
                  right=4pt,
                  top=10pt,
                  arc=0pt,
                  boxrule=1pt,toprule=1pt,
                  colback=white
                  ]
	}
{\end{tcolorbox}}

\newcommand{\event}{\ensuremath{\mathcal{E}}}

\newcommand{\rv}[1]{\ensuremath{{\mathsf{#1}}}\xspace}
\newcommand{\rA}{\rv{A}}
\newcommand{\rB}{\rv{B}}
\newcommand{\rC}{\rv{C}}
\newcommand{\rD}{\rv{D}}

\newcommand{\supp}[1]{\ensuremath{\textnormal{\text{supp}}(#1)}}
\newcommand{\distribution}[1]{\ensuremath{\textnormal{dist}(#1)}\xspace}
\newcommand{\distributions}[2]{\ensuremath{\textnormal{dist}_{#2}(#1)}\xspace}

\newcommand{\kl}[2]{\ensuremath{\mathbb{D}(#1~||~#2)}}
\newcommand{\II}{\ensuremath{\mathbb{I}}}
\newcommand{\HH}{\ensuremath{\mathbb{H}}}
\newcommand{\mi}[2]{\ensuremath{\def\mione{#1}\def\mitwo{#2}\mireal}}
\newcommand{\mireal}[1][]{
  \ifx\relax#1\relax%
    \II(\mione \,; \mitwo)%
  \else%
    \II(\mione \,; \mitwo\mid #1)%
  \fi
}
\newcommand{\en}[1]{\ensuremath{\HH(#1)}}

\newcommand{\itfacts}[1]{\Cref{fact:it-facts}-(\ref{part:#1})\xspace}

\newcommand{\bv}{\mathbf{v}}
\newcommand{\bu}{\mathbf{u}}
\newcommand{\bOne}{\mathbf{1}}
\newcommand{\odd}{B}
\newcommand{\even}{A}

\newcommand{\E}{\mathop{\mathbb{E}}}

\newcommand{\prot}{\pi}
\newcommand{\Prot}{\Pi}

\newcommand{\FMT}{\ensuremath{\textnormal{\textbf{FMT}}}\xspace}

\def\plug{\mathsf{plug}}
\def\nest{\mathsf{nest}}
\def\n{\ell}
\def\Y{\mathsf{Y}}
\def\N{\mathsf{N}}
\def\distnest{\mathcal{D}}

\def\alice{\mathsf{A}}
\def\bob{\mathsf{B}}

\def\ev{\text{\normalfont even}}
\def\od{\text{\normalfont odd}}
\def\pub{\mathsf{pub}}
\def\priv{\mathsf{priv}}
\def\apriv{\alice\text{-}\priv}
\def\bpriv{\bob\text{-}\priv}

\def\layer{\mathcal{L}}
\def\trans{\mathsf{tr}}
\def\out{\mathsf{out}}
\def\expand{\mathsf{Exp}}
\def\Z{\mathsf{Z}}

\title{Multi-Pass Graph Streaming Lower Bounds for Cycle Counting, MAX-CUT, Matching Size, and Other Problems}
\author{Sepehr Assadi\footnote{Department of Computer Science, Rutgers University. Email: \texttt{sepehr.assadi@rutgers.edu}.} \and 
Gillat Kol\footnote{Department of Computer Science, Princeton University. Email: \texttt{gillat.kol@gmail.com}.} \and 
Raghuvansh R. Saxena\footnote{Department of Computer Science, Princeton University. Email: \texttt{rrsaxena@princeton.edu}.} \and 
Huacheng Yu\footnote{Department of Computer Science, Princeton University. Email: \texttt{yuhch123@gmail.com}.}}

\date{}

\begin{document}
\maketitle

\pagenumbering{roman}


\begin{abstract}

Consider the following \emph{gap cycle counting} problem in the streaming model: The edges of a $2$-regular $n$-vertex graph $G$ are arriving one-by-one in a stream and we are promised that $G$
is a disjoint union of either $k$-cycles or $2k$-cycles for some small $k$;  the goal is to distinguish between these two cases using a limited memory. Verbin and Yu [SODA 2011] introduced this problem and showed that any 
single-pass streaming algorithm solving it requires $n^{1-\Omega(\nicefrac{1}{k})}$ space. This result and the proof technique
behind it---the \emph{Boolean Hidden Hypermatching} communication problem---has since been  used extensively for proving streaming lower bounds for various problems, including  
approximating MAX-CUT, matching size, property testing, matrix rank and Schatten norms, streaming unique games and CSPs, and many others. 

\medskip

Despite its significance and broad range of applications, the lower bound technique of Verbin and Yu comes with a key weakness that is also inherited by {all} subsequent results: the Boolean Hidden Hypermatching problem 
is hard only if there is exactly one round of communication and, in fact, can be solved with logarithmic communication in two rounds. Therefore, all streaming lower bounds derived from this problem only hold for \emph{single-pass}  algorithms. 
Our goal in this paper is to remedy this state-of-affairs. 

\medskip

We prove the first \emph{multi-pass} lower bound for the gap cycle counting problem: Any $p$-pass streaming algorithm that can distinguish between disjoint union of $k$-cycles vs $2k$-cycles---or \emph{even} $k$-cycles vs one Hamiltonian cycle---requires $n^{1-\nicefrac{1}{k^{\Omega(1/p)}}}$ space. This  makes progress on multiple open questions in this line of research dating back to the work of Verbin and Yu.

\medskip

As a corollary of this result and by simple (or even no) modification of prior 
reductions, we can extend many of previous lower bounds to multi-pass algorithms. For instance, we can now prove that any streaming algorithm that $(1+\eps)$-approximates the value of MAX-CUT, maximum matching size, or rank of 
an $n$-by-$n$ matrix, requires either $n^{\Omega(1)}$ space or $\Omega(\log{(\nicefrac{1}{\eps})})$ passes. For all these problems, prior work  left open the possibility of even an $O(\log{n})$ space algorithm  in only two passes.

\end{abstract}

\clearpage

\setcounter{tocdepth}{3}
\tableofcontents

\clearpage

\pagenumbering{arabic}
\setcounter{page}{1}


\newcommand{\cycle}{\ensuremath{\textnormal{\textbf{One-or-Many Cycles}}}\xspace}
\newcommand{\OMC}{\ensuremath{\textnormal{\textbf{OMC}}}\xspace}
\newcommand{\omc}{\ensuremath{\textnormal{{OMC}}}\xspace}

\newcommand{\BHH}{\ensuremath{\textnormal{\textbf{BHH}}}\xspace}

\newcommand{\Yes}{\ensuremath{\textnormal{\emph{Yes}}}\xspace}
\newcommand{\No}{\ensuremath{\textnormal{\emph{No}}}\xspace}

\section{Introduction}\label{sec:intro}

Graph streaming algorithms process graphs presented as a sequence of edges under the usual constraints of the streaming model, i.e., by making one or a few passes over the input and using a limited memory. 
There are two main areas of research on graph streams: $(i)$ the \emph{semi-streaming} algorithms that use $O(n \cdot \polylog{(n)})$ space for $n$-vertex graphs and target problems on ({dense}) graphs
such as \emph{finding} MST~\cite{FeigenbaumKMSZ05,McGregor14}, large matchings~\cite{McGregor05,GoelKK12,AssadiKLY16,AhnG15,Kapralov13,GamlathKMS19,AssadiBBMS19}, spanners and shortest 
paths~\cite{FeigenbaumKMSZ08,ElkinZ04,Baswana08,Elkin11,GuruswamiO13,BeckerKKL17,ChakrabartiG0V20}, sparsifiers and minimum cuts~\cite{AhnG09,AhnGM13,KelnerL11,KapralovLMMS14,MukhopadhyayN19}, maximal independent 
sets~\cite{AssadiCK19,CormodeDK19,GhaffariGKMR18}, graph coloring~\cite{AssadiCK19,BeraCG19}, and the like;  and $(ii)$ the \emph{$o(n)$-space streaming} algorithms that use $\polylog{(n)}$ space and aim to \emph{estimate} 
properties of (sparse) graphs such as max-cut value~\cite{KoganK15,KapralovKS15,KapralovKSV17,BhaskaraDV18,KapralovK19}, maximum matching size~\cite{AssadiKL17,EsfandiariHLMO15,ChitnisCEHMMV16,McGregorV16,McGregorV18,CormodeJMM17,KapralovKS14}, number of connected 
components~\cite{HuangP16}, subgraph counting~\cite{Bar-YossefKS02,BravermanOV13,CormodeJ17,McGregorVV16,BeraC17,BulteauFKP16,KallaugherMPV19}, property testing~\cite{HuangP16,MonemizadehMPS17,PengS18,CzumajFPS19}, and others (this is by no means 
a comprehensive list). We will solely focus on \emph{latter} algorithms in this paper\footnote{We shall remark that the challenges for these two classes of algorithms are somewhat different as the former ones
aim to compress the ``edge-space'' of the graph, while the latter ones focus on the ``vertex-space'' (see, e.g.~\cite{GuruswamiO13}).}. 

We study the following problem (or rather family of problems) in the graph streaming model: Given a $2$-regular graph $G=(V,E)$, decide whether $G$ consists of ``many short'' cycles or
all cycles of $G$ are ``rather long''. This problem can be seen as a robust version of 
cycle counting problems (similar-in-spirit to property testing, see, e.g.~\cite{BenderR02,ChenRSS20}). More importantly, this problem turns out to be an excellent {intermediate} problem for studying the {limitations} of streaming algorithms. 

\subsection{Background and Motivation Behind Our Work}\label{sec:background}

The cycle counting problem we study in this paper was first identified by Verbin and Yu~\cite{VerbinY11} for  proving streaming lower bounds for string  problems. In the \emph{gap cycle counting} problem of~\cite{VerbinY11}, 
we are given a graph $G=(V,E)$ and a parameter $k$ and are asked to determine if $G$ is a disjoint union of $k$-cycles or $2k$-cycles. Verbin and Yu proved that any single-pass streaming algorithm for this problem requires 
$n^{1-O(\nicefrac{1}{k})}$ space and used this to establish lower bounds for several other  problems. 

As most other streaming lower bounds, the proof of~\cite{VerbinY11} is through {communication complexity}. The authors first introduced the following \emph{Boolean Hidden Hypermatching} (\BHH) problem:  
In $\BHH_{n,t}$, Alice gets a vector $x \in \set{0,1}^n$ and Bob gets a perfect $t$-hypermatching $M$ on the $n$ coordinates of $x$. We are promised that the $(n/t)$-dimensional vector of {parity} of $x$ on 
hyperedges of $M$, i.e., $Mx = (\oplus_{i=1}^{t} x_{M_{1,i}},\oplus_{i=1}^{t} x_{M_{2,i}},\ldots,\oplus_{i=1}^{t} x_{M_{n/t,i}})$ is either $0^{n/t}$ or $1^{n/t}$; the goal is to distinguish between these two cases using {limited} communication. Building on the Fourier analytic approach of~\cite{GavinskyKKRW07}, Verbin and Yu gave an $\Omega(n^{1-1/t})$ lower bound on the communication complexity of this problem \emph{when only Alice can communicate}. 
The lower bound for cycle counting  then follows from a rather straightforward reduction from $\BHH$, which in turn implies the other lower bounds in~\cite{VerbinY11}.

The $\BHH$ problem has since been extensively used for proving \emph{space} lower bounds for streaming algorithms either through direct reductions~\cite{EsfandiariHLMO15,BuryS15,KoganK15,LiW16,HuangP16,GuruswamiVV17,BravermanCKLWY18}, as a building block for other variants~\cite{KapralovKS15,AssadiKL17,GuruswamiT19}, or as a source of inspirations and ideas~\cite{KapralovKSV17,KapralovK19,KallaugherKP18}.  For instance, 
\begin{enumerate}[label=$(\roman*)$]
\item An $n^{1-O(\eps)}$ space lower bound for $(1+\eps)$-approximation algorithms 
of MAX-CUT by Kogan and Krauthgamer~\cite{KoganK15} and Kapralov, Khanna, and Sudan~\cite{KapralovKS15}, which culminated in the {optimal} lower bound of $\Omega(n)$ for better-than-$2$ approximation by Kapralov and Krachun~\cite{KapralovK19} (see also~\cite{KapralovKSV17} who proved the first $\Omega(n)$ space lower bound for $(1+\Omega(1))$-approximation); 
\item An $n^{1-O(\eps)}$ space lower bound for $(1+\eps)$-approximation 
of maximum matching size by Esfandiari~\etal~\cite{EsfandiariHLMO15} and Bury and Schwiegelshohn~\cite{BuryS15} which was  extended to Schatten norms of $n$-by-$n$ matrices by Li and Woodruff~\cite{LiW16} and Braverman~\etal~\cite{BravermanCKLWY18,BravermmanKKS19}; 
\item An $n^{1-O(\eps)}$ space lower bounds for several property testing problems such as connectivity, cycle-freeness, and bipartiteness by Huang and Peng~\cite{HuangP16}.
\end{enumerate} 
We discuss further background on the $\BHH$ problem and list several of its other implications in~\Cref{app:background-bhh}. Indeed, owing to all these implications, 
$\BHH$ has found its way among the few {canonical} communication problems---alongside Index~\cite{Ablayev93,KremerNR95}, Set Disjointness~\cite{KalyanasundaramS92,Razborov92,Bar-YossefJKS02}, and Gap Hamming
Distance~\cite{IndykW03,ChakrabartiR11}---for proving streaming lower bounds. 

Yet, despite its significance and wide range of applications,  $\BHH$  comes with a {major weakness}: $\BHH$ is a highly {asymmetric} problem and thus its lower bound is \emph{inherently one-way}; Bob can simply send
any of his hyperedges in an $O(t \cdot \log{n})$-bit message which allows Alice to solve the problem. Consequently, \emph{all} aforementioned lower bounds obtained from $\BHH$ in this line of work (and its many variants and generalizations) only hold for \emph{single-pass} algorithms. As a result, we effectively have no knowledge of limitations of \emph{multi-pass} streaming algorithms for these problems, despite the significant attention given
to multi-pass algorithms lately (see, e.g.~\cite{EsfandiariHLMO15,BhaskaraDV18,CormodeJ17,McGregorVV16,BeraC17,BravermanCKLWY18,BravermmanKKS19}). This raises the following fundamental question: 
\begin{quote}
	\emph{Can we prove a \textbf{multi-round} analogue of the Boolean Hidden Hypermatching lower bound that allows for obtaining \textbf{multi-pass} graph streaming lower bounds?} 
\end{quote}
Indeed, this question and its closely related variants have already been raised several time in the literature~\cite{VerbinY11,CutOP,KoganK15,FischerGO17,GuruswamiT19} starting from the work of Verbin and Yu. 

\subsection{Our Contributions}\label{sec:results}

Our main contribution is a multi-round lower bound for the gap cycle counting problem, in fact, in an ``algorithmically simpler''  form, which we call the \cycle (\OMC) problem. 
We then show that by using this problem and simple (or even no) modification of prior 
reductions, we can extend many of previous lower bounds to multi-pass algorithms. 

\subsection*{The \cycle Communication Problem}


\begin{problem}[\cycle ($\OMC$)]
	Let $n,k \geq 1$ be  \underline{even} integers where $n$ divides $k$. In $\OMC_{n,k}$, we have a $2$-regular bipartite graph $G=(L,R,E)$ on $n$ vertices. 
	Edges of $G$ consist of two \emph{disjoint perfect matchings} $M_A$ and $M_B$, given to Alice and Bob, respectively. 
	We are promised that either $(i)$ $G$ consists of a single Hamiltonian cycle (\Yes case) or $(ii)$ $G$ is a collection of $(n/k)$ vertex-disjoint cycles of length $k$ (\No case).  
	The goal is to distinguish between these two cases.  
\end{problem}

$\OMC$  can be seen as the most {extreme} variant of cycle counting problems: in the No case, the graph consists of many short cycles, while in the Yes case, 
the entire graph is one long Hamiltonian cycle. This, at least intuitively, makes this problem ``easiest'' algorithmically (most suitable for reductions) and ``hardest'' for proving lower bounds (our lower bounds extend immediately to many other cycle counting problems including the $k$-vs-$2k$-cycle problem; see~\Cref{rem:arb_mat}).

The following is our main result in this paper. 

\begin{result}\label{res:main}
	For any even integer $k > 0$ and integer $r = o(\log{k})$, any communication protocol (deterministic or randomized) for $\OMC_{n,k}$ in which Alice and Bob send
	up to $r$ messages, i.e., an $r$-round protocol, requires $n^{1-\mathcal{O}(k^{-1/r})}$ communication. 
\end{result}

Our lower bound in~\Cref{res:main} demonstrates a tradeoff between the communication complexity of the $\OMC$ problem and the allowed number of rounds. In particular, 
an immediate corollary of~\Cref{res:main} is that either $\Omega(\log{k})$ rounds or $n^{\Omega(1)}$ communication is needed for solving $\OMC_{n,k}$. 


Let us now briefly compare our~\Cref{res:main} with~\cite{VerbinY11}. By setting $r=1$, we already recover the result of~\cite{VerbinY11} on $n^{1-O(\nicefrac{1}{k})}$ communication lower bound for 
cycle counting problem (up to the hidden-constant in the $O$-notation in the exponent), but this time for the algorithmically easier problem of distinguishing $k$-cycles from a  Hamiltonian cycle. Prior to our work, no lower bounds were known for this problem even for one-round protocols and in fact this was left as an open problem in~\cite[Conjecture 5.1]{VerbinY11}. But much more importantly,~\Cref{res:main} now gives a \emph{multi-round} lower bound 
for $\OMC$ (and other problems such as $k$-vs-$2k$-cycle), making progress on another open problem of~\cite[Conjecture 5.4]{VerbinY11}. We note that our tradeoff does not match  
the conjecture in~\cite{VerbinY11} and it remains a fascinating open question to determine the ``right'' tradeoff for this problem.

\subsection*{Streaming Lower Bounds from OMC} 

The $\OMC$ problem is able to capture the essence of many of previous streaming lower bounds proven via reductions from the $\BHH$ problem. In fact, as we shall see, these reductions often become even easier now that we are working with the $\OMC$ 
problem considering it has a more natural graph-theoretic definition compared to $\BHH$. But more importantly, we can now use the lower bound for $\OMC$ in~\Cref{res:main} to give \emph{multi-pass} streaming lower bounds.

Before listing our results, let us give a concrete example of a lower bound that we can now prove using the $\OMC$ problem to emphasize the simplicity of the reductions. 

\paragraph{Example: Property testing of connectivity.} Huang and Peng~\cite{HuangP16} gave a reduction from $\BHH$ to prove that any single-pass streaming algorithm for property testing of connectivity, namely, deciding whether 
a graph is connected or it requires at least $\eps \cdot n$ more edges to become connected, needs $n^{1-O(\eps)}$ space. We use a reduction from $\OMC$ to extend this lower bound to multiple passes. 

Let $k := \frac{1}{2\eps}$ and $G$ be a graph in the $\OMC_{n,k}$ problem. In the \Yes case, $G$ is a Hamiltonian cycle and is thus already connected. On the other hand,  in the \No case, $G$ consists of $n/k$ disjoint
cycles and thus to be connected requires $n/k-1 > \eps n$ edges. We can run any streaming algorithm on input graphs of $\OMC_{n,k}$ by Alice and Bob creating the stream $M_A$ appended by $M_B$ and 
passing along the memory content of the algorithm to each other. As such, using an algorithm for connectivity, the players can also solve $\OMC_{n,k}$. By~\Cref{res:main}, this implies that any $p$-pass streaming algorithm
for connectivity requires $n^{1-\eps^{\Theta(1/p)}}$ space. Interestingly, not only 
this reduction gives us a multi-pass lower bound, but also it is arguably simpler that the reduction from $\BHH$.

\medskip
 We prove the following lower bounds by reductions from~\OMC and our lower bound in~\Cref{res:main}.  

\begin{mdresult}\label{res:implications}
Let $g(\eps,p) := O(\eps^{c/p})$ for some large enough absolute constant $c > 0$. Then, 
any $p$-pass streaming algorithm for any of the following problems on $n$-vertex graphs requires $n^{1-g(\eps,p)}$ space (the references below list the previous single-pass lower bounds for the corresponding problem): 

\begin{enumerate}[label=$(\roman*)$]
	\item $(1 + \eps)$-approximation of MAX-CUT in sparse graphs (cf. \cite{KoganK15,KapralovKS15,KapralovKSV17,KapralovK19});
	\item $(1+\eps)$-approximation of maximum matching size in planar graphs (cf. \cite{EsfandiariHLMO15,BuryS15});
	\item $(1+\eps)$-approximation of the maximum acyclic subgraph in directed graphs (cf. \cite{GuruswamiVV17});
	\item $(1+\eps)$-approximation of the weight of a minimum spanning tree (cf. \cite{FeigenbaumKMSZ05,HuangP16});
	\item property testing of connectivity, bipartiteness, and cycle-freeness for  parameter $\eps$ (cf. \cite{HuangP16}).  
\end{enumerate}
Our lower bounds can be extended beyond graph streams. For instance, we can also prove lower bounds of $n^{1-g(\eps,p)}$ space for 
 $(1+\eps)$-approximation of rank and other Schatten norms of $n$-by-$n$ matrices (cf. \cite{BuryS15,LiW16,BravermanCKLWY18}), or 
 sorting-by-block-interchange on $n$-length strings (cf. \cite{VerbinY11}). 
\end{mdresult}

A simple corollary of these lower bounds is that for any of these problems, any streaming algorithm requires either $\Omega(\log{(\nicefrac{1}{\eps})})$ passes 
or $n^{\Omega(1)}$ space. Prior to our work, even an $O(\log{n})$ space algorithm in two passes was not ruled out for any of these problems. These results settle or make progress on multiple open questions in the literature 
regarding the multi-pass streaming complexity of gap cycle counting~\cite{VerbinY11}, MAX-CUT~\cite{KoganK15,CutOP}, and streaming CSPs~\cite{GuruswamiT19}. 
We postpone the exact details of our results and further backgrounds in this part to~\Cref{sec:implications}. 

\medskip
To conclude, we believe that~\Cref{res:main} and~\Cref{res:implications} identify $\OMC$ as a natural multi-round analogue of the $\BHH$ problem (as was asked in prior work~\cite{VerbinY11,KoganK15}), 
answering our motivating question. This is indeed the main {conceptual} contribution of our work. 

\subsection{Our Techniques}\label{sec:tech}

We briefly mention the techniques used in our paper here and postpone further details to the technical overview of our proof in~\Cref{sec:technical-overview}. 

The lower bound for $\BHH$ in~\cite{VerbinY11} and other variants in this line of research~\cite{KapralovKS15,GuruswamiT19,KapralovKSV17,KapralovK19,KallaugherKP18} 
all relied heavily on techniques from Fourier analysis on the Boolean hypercube. In contrast, our proofs in this paper solely relies on tools from 
information theory. 


The first main technical ingredient of our work is a novel \emph{round-elimination} argument for $\OMC$. Typical round-elimination arguments for similar problems such as pointer chasing on graphs~\cite{Yehudayoff16,DurisGS84,PapadimitriouS84,NisanW91,ChakrabartiCM08,PonzioRV99} ``track'' the information revealed about a particular \emph{path} inside the graph, ensuring that the player who is speaking next is unaware of which pointer to chase. On the other hand, we crucially need to track the information revealed about \emph{multiple} vertices at the same time (to account for the strong promise in the input instance). 
As such, our proof takes a different approach. We first show that after the first message of the protocol, there is a ``large'' \emph{minor} of the graph---obtained by contracting ``short'' paths---that still ``looks random'' to players. Eliminating a communication round at this point then boils down to \emph{embedding} a hard instance for $r$-round protocols \emph{inside this minor} of the hard instance of $(r+1)$-round protocols; this in particular involves embedding a ``lower dimensional'' instance on smaller number of vertices 
inside a ``higher dimensional'' one (as a graph minor and {not} a subgraph). 

Our second main technical ingredient is the proof of existence of this ``random looking'' graph minor after the first message of the protocol. At the heart of this part is an argument showing that after a single message, 
the \emph{joint} distribution of {all} the vertices reachable from a fixed set of $\approx (n/k)$ vertices by taking a {constant number of edges} remains almost identical to its original distribution. 
This is done by first proving a one-round ``low-advantage'' version of standard pointer chasing lower bounds: For each starting vertex $v$ in the graph, after one message, the distribution of the unique vertex which
is at distance $c$ of $v$ is $\approx n^{-\Omega(c)}$-close to its original (uniform) distribution. The proof is based on bounding the $\ell_2$ norm of the distribution of this unique vertex, and applying a \emph{direct product} type argument for $\Theta(c)$ different sub-problems, each corresponding to going ``one more edge away'' from the starting vertex. The final proof for the joint distribution of the targets of $\approx (n/k)$ vertices is done through a series of reductions from this single-vertex variant.

\subsection{Erratum}
The previous version of this paper had a mistake in one of the technical lemmas: all the main results of the paper are correct as stated, but Lemma 4.2 of Section 5 of the previous version has a subtle flaw (which is in fact not in the main part of the proof, but rather a false observation made in the final part to generalize the proof). Lemma 4.2 was not a significant result in and of itself, but was used as a building
block for our main result. In this version, we present a modified proof of our results. Most of the building blocks and the overall structure of the proofs are  the
same as before, except we no longer rely on the false observation in Lemma 4.2 and instead prove a somewhat weaker variant of this lemma; the remaining changes 
are there to allow for working with this weaker statement in the rest of the proof. See~\Cref{rem:error-FMT} for some more context. 






\section{Technical Overview}\label{sec:technical-overview}

We present a streamlined overview of our technical approach for proving~\Cref{res:main} in this section (we leave the details of our reductions in~\Cref{res:implications} to~\Cref{sec:implications}). 
We emphasize that this section oversimplifies many details and the discussions will be informal for the sake of intuition. 


Before getting to the discussion of our lower bound, it helps to  consider what are some natural ways for Alice and Bob to solve $\OMC_{n,k}$. At one extreme, there is a ``sampling'' approach: Alice can randomly sample $O(n^{1-1/k})$ edges from her input and send them to Bob; 
the (strong) promise of the problem ensures that in the \No-case, Bob will, with constant probability, see an entire $k$-cycle in the graph and thus can distinguish this from the \Yes-case. At the other extreme, 
there is a ``pointer chasing'' approach: the players can start from any vertex of the graph and simply ``chase'' a single (potential) $k$-cycle one edge per round and in (roughly) $k$ rounds solve the problem. And then there are different
interpolations between these two, for instance by chasing $O(n^{1-1/\sqrt{k}})$ random vertices in $\sqrt{k}$-rounds. Our lower bound has to address all these approaches \emph{simultaneously}. 

\subsection{The Pointer Chasing Aspect of \OMC} 
Let us start with the ``pointer chasing'' aspect of our lower bound. Suppose we put the following additional structure on the input to Alice and Bob: 
\begin{enumerate}[label=$(\roman*)$,leftmargin=16pt]
\item The input graph $G$ consists of $k$ layers $V_1,\ldots,V_k$
of size $m = n/k$ with $k$ perfect matchings $M_1,\ldots,M_k$ between them, where $M_i$ is between $V_{i}$ and $V_{(i+1\!\!\mod k)}$. 
\item For any $v \in V_1$, define $P(v) \in V_k$ as the \emph{unique} vertex reachable from $v$ in $G$ \emph{using only} $M_1,\ldots,M_{k-1}$. 
We promise that either: $\textbf{(a)}$ for all $v_i \in V_1$, $P(v_i)$ connects back to $v_i$ in the matching $M_k$, i.e., $M_k \circ P = I$ where $I$ is the identity matching; or $\textbf{(b)}$ for all $v_i \in V_1$, $P(v_i)$ connects to $v_{(i+1\mod m)}$ in $M_k$, 
i.e., $M_k \circ P = S_{+1}$, where $S_{+1}$ is the cyclic-shift-by-one matching.  
\item The input to the players are then alternating matchings from this graph, namely, Alice
receives even-indexed matchings $M_2,M_4,M_6,\ldots$, and Bob receives odd-indexed ones $M_1,M_3,M_5,\ldots$ 
\end{enumerate}
It is easy to verify that for even values of $k$, these graphs form valid inputs to $\OMC_{n,k}$ problem (with case $\textbf{(a)}$ corresponding to the \No-case and $\textbf{(b)}$ corresponding to the Hamiltonian cycle case). 

Under this setting, we can interpret $\OMC_{n,k}$ as some pointer chasing problem: for \emph{some} vertex $v \in V_1$, the players need to ``chase the pointers'' 
\[
M_1(v), \quad M_2 \circ M_1(v), \quad \ldots \quad , M_{k-1} \circ M_{k-2} \circ \cdots \circ M_1(v)
\]
to reach $P(v) \in V_k$; then check whether $P(v)$ connects back to $v$ in $M_k$ or not. 

There are however several differences between a typical pointer chasing problem (see, e.g.~\cite{Yehudayoff16,NisanW91,ChakrabartiCM08,GuruswamiO13,AssadiCK19b,GolowichS20} for many different variants) and our problem. 
Most important among these is that the strong promise in the input effectively means  there is \emph{no single particular} pointer that the players need to chase--all they need to do is to figure out $P(v)$ for 
some $v \in V$ \emph{after} communicating the messages (this is on top of apparent issues such as players being able to chase pointers from ``both ends'' and the like). We elaborate more on this below. 

For intuition, let us consider the following \emph{specialized} protocol $\prot$: the players first completely \emph{ignore} the matching $M_k$ and instead aim to ``learn'' the mapping $P =  M_{k-1} \circ M_{k-2}\circ \cdots\circ M_1$; only then, they will 
look at $M_k$ and check whether $P \circ M_k$ is $I$ or $S_{+}$, corresponding to cases $\textbf{(a)}$ or $\textbf{(b)}$ above. Under this view, we can think of the following two-phase problem: 
\begin{enumerate}
	\item Given $M_1,\ldots,M_{k-1}$ sampled independently and uniformly at random, the players run $\prot$ with transcript $\Prot$ which induces a distribution $\distribution{P \mid \Prot}$ for $P =  M_{k-1} \circ \cdots \circ M_1$;
	\item \emph{Only then}, given a matching $M_k$ sampled uniformly from just \emph{two} choices $\set{I \circ P^{-1}, S_{+1} \circ P^{-1}}$ implied by cases $\textbf{(a)}$ or $\textbf{(b)}$, they need to determine which case $M_k$ belongs to. 
\end{enumerate}
For such a protocol to fail to solve the problem better than random guessing, we should have that:
\begin{align*}
	\Ex_{\Prot}{\tvd{\distribution{I \circ P^{-1} \mid \Prot}}{\distribution{S_{+1} \circ P^{-1} \mid \Prot}}} = o(1), \tag{here and throughout, $\Prot$ is the transcript of the protocol}
\end{align*}
as the players are getting one sample, namely, $M_k$, from one of the two distributions and thus should not be able to distinguish them with one sample\footnote{Note that technically we should also condition on the input of the player outputting the final answer but for simplicity, we will ignore that in this discussion for now.}. As such, the task of proving a lower bound 
essentially boils down to showing that for a ``small'' round and communication protocol $\prot$, 
\begin{align}
	\Ex_{\Prot}\tvd{\distribution{P \mid \Prot}}{\distribution{P}} = o(1) \label{eq:over1},
\end{align}
namely, that the protocol cannot change the distribution of the \emph{entire} mapping $P$ by much; this should be contrasted with typical lower bounds for pointer chasing
that require that distribution of a \emph{single} pointer $P(v)$ for $v \in V_1$ does not differ considerably after the communication. 

This outline oversimplifies many details. Most importantly, it is not at all the case that the only protocols that solve the problem adhere to the special two-phase approach mentioned above. Indeed, the input of players are highly correlated in the problem
and this can reveal information to the players. Consequently, in the actual lower bound, we need to handle these correlation throughout the entire proof. In particular, we need stronger variant of \eqref{eq:over1} that shows the value of 
distribution $P \mid \Prot$ on \emph{two particular} points (rather than two marginally random points) in the support are close. 
We postpone these details to the actual proof and for now focus on the lower bound for~\eqref{eq:over1} which captures the crux of the argument. 


\subsection{The One Round Lower Bound}

We first prove a \emph{stronger} variant of~\eqref{eq:over1} for one round protocols. Suppose $M_1,\ldots,M_{k-1}$ are sampled uniformly at random  and $M_{\odd}$ denotes 
the input of Bob among these matchings. Then, we show that if Alice sends a \emph{single}  message $\Prot$ of size $C = o(m)$ (recall that $m = n/k$), we will have: 
\begin{align}
	\mi{\rv P}{\rv M_{\odd},\rv \Pi} \leq m^{\Theta(1)} \cdot \paren{\frac{C}{m}}^{\Theta(k)}. \label{eq:over2}
\end{align}
Let us interpret this bound: the mutual information between $P$ and $(M_{\odd},\Pi)$ is a measure of how much the distribution of $P$ is affected by the extra conditioning on $M_{\odd},\Pi$ (see~\Cref{prelim-fact:kl-info});
in particular, if we set $C = m^{1-\Theta(1/k)}$ in the above bound, we get that the RHS is $o(1)$ for large enough $k$ which in turn implies that distribution of $P$ conditioned on \emph{both} $M_{\odd},\Pi$, 
is very close to its original distribution, implying~\eqref{eq:over1} (in a stronger form). 
The proof of~\eqref{eq:over2} consists of two parts. 

\paragraph{Part One:} The main part of the proof is to consider a ``low advantage'' variant of pointer chasing: 
\begin{itemize}[label=$-$]
	\item \textbf{Low-advantage pointer chasing:} Given a \emph{fixed} vertex $v \in V_1$ and assuming we are only allowed one round of communication, what is the probability of guessing $P(v) \in V_k$ correctly? 
\end{itemize}
\noindent
Standard pointer chasing lower bounds such as~\cite{Yehudayoff16,NisanW91} imply that the answer for a $(<k)$-round protocol with communication cost $C$ is roughly $\nicefrac{1}{m}+\nicefrac{C}{km}$. We on the other hand prove a much stronger
bound but only for \emph{one}-round protocols which is roughly $\nicefrac{1}{m}+\paren{\nicefrac{C}{m}}^{\Theta(k)}$. 

The proof of this part is one of the main technical contributions of our work. For our purpose it would be easier to bound the $\ell_2$-norm of the vector $P(v)$, i.e., show that: 
\begin{align}
	\Ex_{M_{\odd},\Prot} \norm{\rv P(v) \mid M_{\odd},\Prot}^2_2 \leq \frac{1}{m} + \paren{\frac{C}{m}}^{\Theta(k)}, \label{eq:over3}
\end{align}
which can be used  to bound the advantage of the protocol over random guessing. 

It turns out the key to bounding the LHS above is to understand the ``power'' of message $\Prot$ in changing the distribution of collections of edges chosen from \emph{disjoint} matchings 
in the input of Alice. Let $S \subseteq \set{2,4,\ldots,k-2}$ be a set of indices of Alice's matchings and $\bv_S = (\hat{e}_{i_1},\ldots,\hat{e}_{i_S})$ denote a collection of ``potential'' edges in $M_{i_1},M_{i_2},\ldots,M_{i_S}$ (i.e., pairs of vertices which may or may not be an edge in each matching). In particular, we can bound the LHS of~\eqref{eq:over3} by
\begin{align}
	\frac{1}{m} + \frac{1}{m \cdot (m-1)^{k/2-1}} \cdot \Ex_{\Prot}\bracket{\sum_{S \subseteq \set{2,4,\ldots,k-2}}\sum_{\bv_S} \Pr\paren{ \hat e_{i_1} \in M_{i_1} \wedge \cdots \wedge \hat e_{i_{|S|}} \in M_{i_S} \mid \Prot}^2}, \label{eq:over4}
\end{align}
and then prove that for any $S \subseteq \set{2,4,\ldots,k-2}$: 
\begin{align}
	\Ex_{\Prot}\bracket{\sum_{\bv_S} \Pr\paren{ \hat e_{i_1} \in M_{i_1} \wedge \cdots \wedge \hat e_{i_{|S|}} \in M_{i_S} \mid \Prot}^2} \leq \paren{{Cm}}^{\card{S}/2}. \label{eq:over5}
\end{align}
Plugging in~\eqref{eq:over5} inside~\eqref{eq:over4} then prove the ``low advantage'' lower bound we want in~\eqref{eq:over3}. 

The statement (and the proof of)~\eqref{eq:over5} can be seen as some \emph{direct product} type result: When $\card{S} = 1$, we are simply bounding the (square of the) probability that a potential edge belongs to the 
graph of Alice conditioned on an ``average'' message; considering each matching $M_i$ is a random permutation of size $[m]$ and the message reveals only $C = o(m)$ bits about it, we  expect only a small number of edges to have a ``high'' probability 
of appearing in $M_i$ conditioned on $\Prot$ (see~\Cref{lem_info_perm}). Our bounds in~\eqref{eq:over5} then show that repeating this task for $\card{S}$ times, namely, increasing the probability of an entire $\card{S}$-tuple $\bv_S$ of edges, becomes 
roughly $\card{S}$ times less likely. 

\paragraph{Part Two:} Our lower bound in~\Cref{eq:over3} can also be interpreted as bounding: 
\begin{align}
	\mi{\rv P(v)}{\rv M_{\odd},\rv \Pi} \leq \paren{\frac{C}{m}}^{\Theta(k)}, \label{eq:over6}
\end{align}
for a fixed vertex $v \in V_1$, namely, a \emph{single-vertex} version of the bound we want in~\eqref{eq:over2}. We obtain the final bound using a series of reductions from this. 
In particular, by chain rule: 
\[
	\mi{\rv P}{\rv M_{\odd},\rv \Pi} = \sum_{v \in V_1} \mi{\rv P(v_i)}{\rv M_{\odd},\rv \Pi \mid \rv P(v_1),\ldots,\rv P(v_{i-1})}.
\]
For the \emph{first} $m/2$ terms in the sum above, we can show that the bounds in~\eqref{eq:over6} continue to hold \emph{even conditioned} on the new $\rv P$-values; this is simply because even conditioned on these values, 
at least half of each matching is ``untouched'' and thus we can apply the previous lower bound to the underlying subgraph with $(\geq m/2)$-size layers instead. 
This argument however cannot be extended to all values in the sum simply because the size of 
layers are becoming smaller and smaller through this conditioning. To handle this in our actual proof, we modify our graph by shrinking the first and the last layers by a constant factor and generalize all our arguments to this modified graph. We postpone the details of this part to the actual proof but note here that this essentially amounts to ``getting rid'' of the last $m/2$ terms.

\subsection{The Round Elimination Argument} 

The next key ingredient of our proof is a round elimination argument for ``shaving off'' the rounds in any $r$-round protocol one by one, until we end up with a $0$-round protocol that can still 
solve a non-trivial problem; a contradiction. 

A typical round elimination argument for pointer chasing shows that after the first message of the protocol, the distribution of \emph{next immediate} pointer to chase (namely, $M_2 \circ M_1(v)$ when chasing $M_{k} \circ \cdots \circ M_1(v)$) 
is still almost uniform as before. Thus, the players now need to solve the same problem with \emph{one less round} and \emph{one less matching}. Unfortunately, such an approach does not suffice for our purpose in proving~\eqref{eq:over1} 
in which chasing \emph{any} pointer solves the problem. 

Our round elimination argument thus takes a different route. We show that after the first message of the protocol, the {joint} distribution of \emph{all long paths} in the \emph{entire} graph is still almost uniform. 
Let us formalize this as follows for proving a lower bound for $r$-round protocols. Consider the recurrence $k_r = c \cdot k_{r-1}$ and $k_0 = 1$ for some sufficiently large constant $c > 1$, and a $k_{r}$-layered
graph $G_r$ as before. For any $i \in [k_{r-1}]$, define $P_i := M_{i \cdot c} \circ \ldots \circ M_{(i-1) \cdot c +1}$, namely, the composition of the matchings in \emph{blocks} of length $c$ each. We will use our bounds in~\eqref{eq:over2} 
to show that after the {first} message $\Prot_1$ of any $r$-round protocol $\prot$ with communication cost $C$, 
\begin{align}
	\Ex_{\Prot_1,M_{\odd}}{\tvd{\distribution{P_1,\ldots,P_{k_{r-1}} \mid \Prot_1,M_{\odd}}}{\distribution{P_1,\ldots,P_{k_{r-1}}}}} \leq  m^{\Theta(1)} \cdot \paren{\frac{C}{m}}^{\Theta(c)} \label{eq:over7}. 
\end{align}
 In words, after the first round, the {joint} distribution of compositions of matchings still look almost uniform \emph{to Bob}. Notice that the main difference in~\eqref{eq:over7}
compared to~\eqref{eq:over2} that the bounds are now applied to \emph{each} block of length $c$ in the graph, not the entire $k$ layers. 

Now let us see how we can use this to eliminate one round of the protocol. This is done through an {embedding argument} in which we embed an instance $G_{r-1}$ of the problem on $k_{r-1}$ layers 
inside a graph $G_r$ of $k_{r}$ layers {conditioned on the first message}, and run the protocol $\prot$ for $G_r$ {from the second round onwards} to obtain an $(r-1)$-round protocol $\theta$ for $G_{r-1}$. The embedding is as follows. 

\paragraph{Embedding $G_{r-1}$ inside $G_r \mid \Prot_1$.} Let $M_1,\ldots,M_{k_{r-1}}$ be the inputs to Alice and Bob in $G_{r-1}$. In the protocol $\theta$, the players 
first use {public randomness} to sample a message $\Prot_1$ from the distribution induced by $\prot$ on $G_r$. We would now like to \emph{sample} a graph $G_{r} \mid \Prot_1$ such that:
\begin{enumerate}[label=$(\roman*)$]
\item for every $i \in [k_{r-1}]$, we have $P_i = M_i$, i.e., the composition of the $i$-th block of $c$ consecutive matchings in $G_r$ looks the same as the matching $M_i$; 
\item Alice in $\theta$ gets the input of Bob in $\prot$ in $G_r$ and Bob in $\theta$ gets the input of Alice in $\prot$.
\end{enumerate}

{Assuming} we can do this, Alice and Bob in $\theta$ can continue running $\prot$ on $G_r \mid \Prot_1$ as they both know the first message of $\prot$ and by property $(ii)$ above have the proper inputs; moreover, by property $(i)$, 
$M_{k_{r-1}} \circ \ldots \circ M_{1} = P_{k_{r-1}} \circ \ldots \circ P_{1} = P$, i.e., the distribution of pointers they would like to chase in both $G_r$ and $G_{r-1}$ is the same. Thus if $\prot$ was able to change the distribution of $P$ in $G_r$ in $r$ 
rounds, then $\theta$ can also change the distribution of $P$ in $G_{r-1}$ in $(r-1)$ rounds. 

Of course, we cannot hope to straightaway perform the sampling above without any communication between the players as $(i)$ $\Prot_1$ correlates the distribution of $P_1,\ldots,P_{k_{r-1}}$ in $G_r \mid \Prot_1$ (even though they were independent originally), 
and $(ii)$ Alice and Bob in $\theta$ know only half of $P_1,\ldots,P_{k_{r-1}}$ each (dictated by their input in $G_{r-1}$). This is where we use~\eqref{eq:over7}. Intuitively, since the distribution of $P_1,\ldots,P_{k_{r-1}} \mid \Prot_1,M_{\odd}$ has not changed 
dramatically (and \emph{not at all} if we condition on $M_{\even}$ instead of $M_{\odd}$ since Alice is the sender of $\Prot_1$), we can design a sampling process based on a combination of public and private randomness that 
``simulates'' sampling
\[
G_r \sim G_r \mid \Prot_1, P_1 = M_1, \ldots, P_{k_{r-1}} = M_{k_{r-1}}
\]
by instead sampling from the product distribution 
\[
	G_r \sim \bigtimes_{i \in [k_{r-1}]} \text{matchings $M^{i}_1,\ldots, M^{i}_c$ in $i$-th block of $G_r$} \mid \Prot_1, P_{i} = M_{i},
\]
while obtaining the same answer as $\prot$ up to a negligible factor of $1/\poly(m)$ error. 

To conclude, if the communication cost of the protocol is only $C = {m^{1-\Theta(1/c)}}$, we can shave off all the $r$-rounds of the protocol, while shrinking the number of layers in input graphs by a factor of $c$ each time; by the choice of 
$k_r = c^r$, we will eventually end up with a $0$-round protocol on a non-empty graph which cannot change the distribution of corresponding $P$ at all. Tracing back the argument above then implies that the 
original $r$-round protocol should not be able to change the distribution of its own mapping $P$ by more than $O(r/\poly(m))$ as desired. Rewriting $c$ as $k^{1/r}$, we obtain an $m^{1-\Theta(1/k^{1/r})}$ lower bound on the communication cost of 
protocols for~\eqref{eq:over1}.

\section{Preliminaries and Problem Definition}\label{sec:prelim} 

\subsection{Notation}
\paragraph{Sequences.} Unless stated otherwise, we use the letters $S, T$ to denote sequences. The notation $S \Vert T$ will denote the concatenation of the sequences $S$ and $T$ and $\card{S}$ denotes the length of $S$. For a sequence $S$ such that $\card{S} > 0$, we use $S[:-1]$ to denote the sequence $S$ without its last element. For $i \in [\card{S}]$, we use $S_{\leq i}$ to denote the sequence with the first $i$ elements of $S$. Define $S_{\geq i}$ analogously. Finally, for some $s$ and $c > 0$, the notation $s^{(c)}$ denotes the $c$-length sequence with all elements equal to $s$. 

\paragraph{Random variables.} When there is room for confusion, we use sans-serif letters for random variables (e.g.~$\rA$) and normal letters for their realizations (e.g. $A$). For a random variable $\rA$, we use $\supp{\rA}$ to denote the support of $\rA$ and $\distribution{\rA}$ to denote the distribution of $\rA$. When it is clear from the context, we may abuse the notation and use $\rA$ and $\distribution{\rA}$ interchangeably. By norm $\norm{\rA}$ of a random variable with distribution $\mu$ and support $a_1,\ldots,a_m$, we mean the norm of the vector $(\mu(a_1),\ldots,\mu(a_m))$.

\paragraph{Information theory.} For random variables $\rA,\rB$, we use $\en{\rA}$ to denote the \emph{Shannon entropy} of $\rA$, $\en{\rA \mid \rB}$ to denote the \emph{conditional entropy}, and 
$\mi{\rA}{\rB}$ to denote the \emph{mutual information}. Similarly, for two distributions $\mu$ and $\nu$ on the same support, $\tvd{\mu}{\nu}$ denotes their \emph{total variation distance} and $\kl{\mu}{\nu}$ is 
their \emph{KL divergence}.~\Cref{app:info} contains the definitions and standard properties of these notions as well as some auxiliary lemmas that we prove in this paper. 

\subsection{Layered Graphs and Nested Block Graphs}

We start by defining a layered graph.

\begin{definition}
\label{def:layered}
Let $k > 0$ and $T = \{t_i\}_{i = 0}^k$ be a sequence of positive integers. A $T$-layered graph $G = (V, E)$ is such that $\card{V} = \sum_{i=0}^k t_i$ and the following holds:

If $V = V_0 \sqcup V_1 \sqcup \cdots \sqcup V_k$ is the lexicographically smallest partition of $V$ into sets of sizes $t_0, t_1, \cdots t_k$ respectively, then the set $E$ can be partitioned as $E = E_1 \sqcup \cdots \sqcup E_k$ where for all $i \in [k]$, $E_i$ is a matching of size $\min(t_{i-1}, t_i)$ between $V_{i-1}$ and $V_i$.
\end{definition}

Observe that the partition $E = E_1 \sqcup \cdots \sqcup E_k$ is unique. When we wish to emphasize the partitions, we write $G = ( \bigsqcup_{i=0}^k V_i, \bigsqcup_{i=1}^k E_i )$. We call the sets $V_i$ the `layers' of $G$, with $V_0$ being the first and $V_k$ being the last layer. Often when we talk about layered graphs, we will have a fixed set of layers in mind and consider different sets of edges between these layers. In particular, when we talk about a random layered graph, we mean a graph chosen uniformly at random amongst all layered graphs with those layers. We use $\layer_T$ to denote the uniform distribution over all $T$-layered graphs and omit the subscript when it is clear from context.

As the edges in a layered graph form a matching between adjacent layers, any vertex in any layer is connected to at most one vertex in any other layer. For $i < i' \in \{0\} \cup [k]$, we define the function $G_{i \to i'} : V_i \to V_{i'}$ to be such that for all $v \in V_i$,  $G_{i \to i'}(v)$ is either the unique vertex in layer $V_{i'}$ that $v$ connects to, or is $\bot$, if no such vertex exists. We simply say $G_{i'}$ if $i' = i+1$. We say that a layered graph is \emph{nice} if for all $i < i' \in \{0\} \cup [k]$, the number of vertices in $V_i$ such that $G_{i \to i'}(v) \neq \bot$ is the largest possible. Formally, 
\begin{definition}
\label{def:nice}
Let $k > 0$ and $T = \{t_i\}_{i = 0}^k$ be a sequence of positive integers. A $T$-layered graph $G = ( \bigsqcup_{i=0}^k V_i, \bigsqcup_{i=1}^k E_i )$ is said to be nice if, for all $i < i' \in \{0\} \cup [k]$, we have
\[
\card{ \{ v \in V_i \mid G_{i \to i'}(v) \neq \bot \} } = \min_{i \leq i'' \leq i'} t_{i''} .
\]
\end{definition}

\paragraph{Nested block graphs.} Next we define a subset of layered graphs that we use throughout this work. Let $c, s, k > 0$ and $T = \{ t_i \}_{i = 0}^k$ be a sequence of positive integers. Define the sequence $\plug(T, c, s)$ as:
\[
\plug(T, c, s) = t_0 \Vert s^{(c)} \Vert t_1 \Vert \cdots \Vert s^{(c)} \Vert t_k.
\]
Namely, $\plug(T, c, s)$ is the sequence obtained by plugging $c$ occurrences of $s$ between any two elements of $T$. Let $S = \{ s_i \}_{i=1}^k$ be an increasing\footnote{We adopt the convention that the sequence $S$ will be indexed from $1$ and increasing, while $T$ will be indexed from $0$ and may or may not be increasing.} sequence of positive integers. Define:
\[
\nest(c, S) = \begin{cases}
s_1^{(2)}, &\text{~if~} \card{S} = 1\\
\plug(\nest(c, S[:-1]), 2c - 1, s_{\card{S}}), &\text{~if~} \card{S} > 1
\end{cases} .
\]
We define $\n_{c, S} = \card{\nest(c, S)} - 1$ and note that $\n_{c, S} = (2c)^{\card{S} - 1}$. We drop the subscripts $c, S$ when they are clear from context. 

\begin{definition}[Nested block graphs]
\label{def:nbc}
Let $c > 0$ and $S$ be an increasing sequence of positive integers. A $(c, S)$-nested block graph is a $\nest(c, S)$-layered graph that is nice. 
\end{definition}

Analogous to layered graphs, we shall consider random nested block graphs over the same set of vertices. We reserve $\distnest_{c, S}$ to denote the uniform distribution over $(c, S)$-nested block graphs and omit the subscripts when they are clear from context. For a permutation $\Z$ on $[s_1]$ we use $\distnest_{c, S}^{\Z}$ to denote the distribution $\distnest_{c, S}$ over nested block graphs $G$ conditioned on $G_{0 \to \n} = \Z$. 

\paragraph{Almost nested block graphs.} 
Let $c, k > 0$ and $S = \{s_i\}_{i=1}^k$ be an increasing sequence of positive integers. Let $0 < s \leq s_1$ be given. We define the sequence $\nest_l(s, c, S)$ to equal to $\nest(c, S)$ except that the first element is $s$ instead of $s_1$. Similarly, the sequence $\nest_r(s, c, S)$ to equal to $\nest(c, S)$ except that the last element is $s$.  We omit the subscript $l, r$ when it is clear from context. We define an $(s, c, S)$-almost nested block graph to be a $\nest(s, c, S)$-layered graph that is nice. In other words, an $(s, c, S)$-almost nested block graph is the same as a $(c, S)$-nested block graph except that either the first or the last layer is of size $s$. We define the notation $\distnest_{s, c, S}$ and $\distnest_{s, c, S}^{\Z}$ for almost nested block graphs analogous to our notation for nested block graphs. We omit left and right in this notation as it shall be clear from context.

\subsection{The Final-Matching Testing Problem}
We work in the standard two-party communication model \cite{Yao79} (see~\cite{KushilevitzN97} for an overview of the standard definitions). For a sequence $T$ and $r > 0$, an $r$-round $T$-protocol $\prot$ is defined over a $T$-layered graph $G = ( \bigsqcup_{i=0}^{\card{T} - 1} V_i, \bigsqcup_{i=1}^{\card{T} - 1} E_i )$. Alice's input is the edges $E_i$ for even $i$ and Bob's input is the edges $E_i$ for odd $i$. Both players know the partition $\bigsqcup_{i=0}^{\card{T}} V_i$.  Using these inputs, Alice and Bob communicate for $r$ rounds where Alice sends a message to Bob in odd rounds and Bob sends a message to Alice in even rounds. If $r$ is even (resp. odd), then, after $r$ rounds, Alice (resp. Bob) computes an output based on its input and the transcript. We use $\norm{\prot}$ to denote the communication cost of the protocol $\prot$ defined as the worst-case number of bits communicated (not including the output) between Alice and Bob in $\prot$ over any input. We also use $\prot(G)$ or $\prot(\bigsqcup_{i=1}^{\card{T}} E_i)$ to denote the output of the protocol $\prot$ when the inputs are $\bigsqcup_{i=1}^{\card{T}} E_i$.

If $T = \nest(c, S)$ for some sequence $S$ and $c > 0$, we may say that $\prot$ is an $r$-round $(c,S)$-protocol. Analogously, define an $r$-round $(s, c, S)$-protocol with almost nested block graphs instead of nested block graphs.

Based on these, we define the following communication problem.

\begin{problem}[\textbf{(Distributional) Final-Matching Testing}] \label{prob:FMT}
	The $\FMT$ problem is a two-player communication problem parameterized by integers $m,c,r \geq 1$ and two fixed permutations $\Y \neq \N$ on $[m]$. Let $S$ be the sequence $S = m, m + \left\lfloor{\frac{m}{r}}\right\rfloor, m + \left\lfloor{\frac{2m}{r}}\right\rfloor, \cdots, 2m$. 
	
	We sample a $(c,S)$-nested block graph $G$ from the distribution $\frac12 \cdot \distnest_{c, S}^{\Y} + \frac12 \cdot \distnest_{c, S}^{\N}$, and the goal is to determine which is the case, i.e., whether the {\em final matching} $G_{0 \to \n}$ is $\Y$ or $\N$. 
\end{problem}


\subsection{Main Result}

Our main technical result is the following lower bound for the FMT problem. \Cref{res:main} then follows easily from this theorem as we show in~\Cref{sec:OMC}.  
\begin{theorem}
\label{thm:FMT}
Let $r, m, c > 0$ be such that $m$ is odd and sufficiently large. Define the sequence $S = m, m + \left\lfloor{\frac{m}{r}}\right\rfloor, m + \left\lfloor{\frac{2m}{r}}\right\rfloor, \cdots, 2m$.  For all permutations $\Y, \N$ on $[m]$ and all randomized $r$-round $(c,S)$-protocols $\Pi$ such that $\norm{\Pi} \leq 2^{-\mathcal{O}(r)} \cdot m^{1 - \frac{200}{c}}$, it holds that:
\[
\tvd{ \distributions{ \rv{\Pi}(G) }{ G \sim \distnest_{c, S}^{\Y} } }{ \distributions{ \rv{\Pi}(G) }{ G \sim \distnest_{c, S}^{\N} } } \leq \frac{1}{m^3} .
\]
\end{theorem}

The proof of \Cref{thm:FMT} is by induction on $r$. The base case ($r = 1$) is captured in \Cref{lemma:oneround} while the induction step is captured in \Cref{lemma:multiround}. Observe that \Cref{lemma:oneround} and \Cref{lemma:multiround} talk about almost nested block graphs while \Cref{thm:FMT} talks about nested block graphs. This generalization is needed to make the induction work.

\newcommand{\lemmaoneround}{
Let $S = \{s_i\}_{i = 1}^2$ be an increasing sequence and $c, C > 0$ be integers. For all $s \leq s_1$, injections $\Y, \N: [s] \to [s_1]$, and all randomized $1$-round $(s, c, S)$-protocols $\Pi$ such that $\norm{\Pi} \leq C$, we have
\[
\tvd{ \distributions{ \rv{\Pi}(G) }{ G \sim \distnest_{s, c, S}^{\Y} } }{ \distributions{ \rv{\Pi}(G) }{ G \sim \distnest_{s, c, S}^{\N} } } \leq 2s_2^2\cdot \left( \frac{40C}{s_2 - s} \right)^{c/6} := \Phi(s, c, s_2, C).
\]
}

\begin{lemma}\label{lemma:oneround}
	\lemmaoneround
\end{lemma}

\newcommand{\lemmamultiround}{
Fix an increasing sequence $S$ and $c, C, \epsilon_0 > 0$. Assume that $\card{S} \geq 2$ and for all $s \leq s_2$, injections $\Y', \N' : [s] \to [s_2]$, and all randomized $(\card{S} - 2)$-round $(s, c, S_{\geq 2})$-protocols $\Pi'$ such that $\norm{\Pi'} \leq C$, we have:
\[
\tvd{ \distributions{ \rv{\Pi'}(G') }{ G' \sim \distnest_{s, c, S_{\geq 2}}^{\Y'} } }{ \distributions{ \rv{\Pi'}(G') }{ G' \sim \distnest_{s, c, S_{\geq 2}}^{\N'} } } \leq \epsilon_0 .
\]
Then, for all $s \leq s_1$, injections $\Y, \N : [s] \to [s_1]$, and all randomized $(\card{S} - 1)$-round $(s, c, S)$-protocols $\Pi$ such that $\norm{\Pi} \leq C$, we have:
\[
\tvd{ \distributions{ \rv{\Pi}(G) }{ G \sim \distnest_{s, c, S}^{\Y} } }{ \distributions{ \rv{\Pi}(G) }{ G \sim \distnest_{s, c, S}^{\N} } } \leq 2c\epsilon_0 + \Phi(s, c, s_2, C) ,
\]
where $\Phi(\cdot)$ is as defined in \Cref{lemma:oneround}.  
}
\begin{lemma}\label{lemma:multiround}
	\lemmamultiround
\end{lemma}

We now prove~\Cref{thm:FMT}.

\begin{proof}[Proof of~\Cref{thm:FMT} (assuming~\Cref{lemma:oneround,lemma:multiround})] 
Let $C = \norm{\Pi}$ and observe that $\card{S} = r+1$. Using \Cref{lemma:oneround} once and \Cref{lemma:multiround} $r-1$ times, we get that:
\[
\tvd{ \distributions{ \rv{\Pi}(G) }{ G \sim \distnest_{c, S}^{\Y} } }{ \distributions{ \rv{\Pi}(G) }{ G \sim \distnest_{c, S}^{\N} } } \leq (2c)^{r - 1} \cdot \sum_{i = 2}^{r} \Phi(s_{i-1}, c, s_i, C) .
\]
Plugging in the definition of $\Phi(\cdot)$ from \Cref{lemma:oneround}, we get:
\[
\tvd{ \distributions{ \rv{\Pi}(G) }{ G \sim \distnest_{c, S}^{\Y} } }{ \distributions{ \rv{\Pi}(G) }{ G \sim \distnest_{c, S}^{\N} } } \leq (2c)^{r - 1} \cdot \sum_{i = 2}^{r} 2s_i^2\cdot \left( \frac{40C}{s_i - s_{i-1}} \right)^{c/6} .
\]
We simplify using the definition of $S$:
\[
\tvd{ \distributions{ \rv{\Pi}(G) }{ G \sim \distnest_{c, S}^{\Y} } }{ \distributions{ \rv{\Pi}(G) }{ G \sim \distnest_{c, S}^{\N} } } \leq (2c)^{r - 1} \cdot 8rm^2 \cdot  \left( \frac{80 r C}{m} \right)^{c/6} .
\]
Using $C \leq 2^{-\mathcal{O}(r)} \cdot m^{1 - \frac{200}{c}}$, we get:
\[
\tvd{ \distributions{ \rv{\Pi}(G) }{ G \sim \distnest_{c, S}^{\Y} } }{ \distributions{ \rv{\Pi}(G) }{ G \sim \distnest_{c, S}^{\N} } } \leq (2c)^{r - 1} \cdot 8rm^2 \cdot   2^{-\mathcal{O}(cr)} \cdot m^{- 30} \leq \frac{1}{m^3} .
\]
\end{proof}

\begin{remark}\label{rem:error-FMT}
	In an earlier version of this paper, we defined a similar FMT problem with the difference that the size of every layer was equal and the final matching corresponded to a \emph{perfect matching} between the first and last layers. And we mistakenly claimed an analogue of~\Cref{thm:FMT} for this problem. 
	
	However, one can verify that in that case, if the \emph{parity} of permutations $\Y$ and $\N$ are different, then there is a simple $O(1)$ communication protocol that can solve the problem in one-round by Alice sending the parity of her combined matchings. As such,~\Cref{thm:FMT} in such generality cannot hold 
	for that version of FMT problem, hence necessitating the approach taken in the current version. 
\end{remark}

\subsection{The Proof of \Cref{res:main}}
\label{sec:OMC}

In this section, we formalize and prove \Cref{res:main} assuming \Cref{thm:FMT}. 

\begin{theorem}
\label{thm:mainformal}
For even $k > 0$ and $n,r > 0$ such that $n$ is an odd multiple of $k$, any $r$-round communication protocol $\prot$ that solves $\OMC_{n,k}$ with probability at least $(1/2+(k/n)^2)$ has communication cost $\norm{\prot} = 2^{-\mathcal{O}(r)} \cdot (n/k)^{1 - 10^4 \cdot k^{-1/r}}$.
\end{theorem}
\begin{proof}
Let us first assume that $k$ is a multiple of $4$ and define $m = n/k$. As $n$ is an odd multiple of $k$, we have that $m$ is odd. Let $c$ be the largest such that $(2c)^r < k/4$, $\Y$ be the permutation on $[m]$ such that $\Y(i) = (i + 1) \bmod m$ for all $i \in [m]$ and $\N$ be the identity permutation, {\em i.e.}, $\N(i) = i$ for all $i \in [m]$. Finally, let $S = m, m + \left\lfloor{\frac{m}{r}}\right\rfloor, m + \left\lfloor{\frac{2m}{r}}\right\rfloor, \cdots, 2m$ as in the statement of \Cref{thm:FMT}.

We use $\prot$ to construct an $r$-round $(c,S)$-protocol $\Pi$ as follows: Given as input $\nest(c, S)$-layered graph $G = ( \bigsqcup_{i=0}^{\n} V_i, \bigsqcup_{i=1}^{\n} E_i )$, where Alice's input is the edges $E_i$ for even $i$ and Bob's input is the edges $E_i$ for odd $i$, the protocol $\Pi$ behaves as follows:

\begin{enumerate}[label=$(\arabic*)$]
\item \label{item:fmt2omc1}Alice and Bob add dummy vertices to all layers so that all of them have size $2m$. For all odd (resp. even) $i \in [\n]$, Alice (resp. Bob) adds dummy edges to $E_i$ so that it is a perfect matching between $V_{i-1}$ and $V_i$. 
\item Alice and Bob add $k/4 - \n - 1$ new layers of vertices of size $2m$ each. For all $\n < k' < k/4$, if $k'$ is even, then Alice adds the identity matching between layer $k' - 1$ and layer $k'$ to her set of edges. Otherwise, Bob adds the identity matching to his set of edges.
\item The players make a copy of the graph, denoted by $\tilde{G}$, with layers $\tilde{V}_0, \tilde{V}_1,\ldots,\tilde{V}_{k/4-1}$.
\item Bob adds the following edge to his input (see Figure~\ref{fig_proof_main_informal}):
	\begin{enumerate}
		\item an identity matching (i.e., $i$-th vertex to $i$-th vertex) between the bottom ({\em i.e.}, not the dummy vertices added in \Cref{item:fmt2omc1}) $m$ vertices in $V_{k/4-1}$ and the bottom $m$ vertices $\tilde V_0$,
		\item an identity matching between the bottom $m$ vertices in $\tilde V_{k/4-1}$ and the top $m$ vertices in $V_{0}$,
		\item an identity matching between the top $m$ vertices in $V_{k/4-1}$ and the top $m$ vertices in $\tilde V_{k/4-1}$,
		\item an identity matching between the top $m$ vertices in $\tilde V_0$ and the bottom $m$ vertices in $V_0$.
	\end{enumerate}
\item Alice and Bob run the protocol $\pi$ on the new graph $G_{\omc}$. 
\end{enumerate}
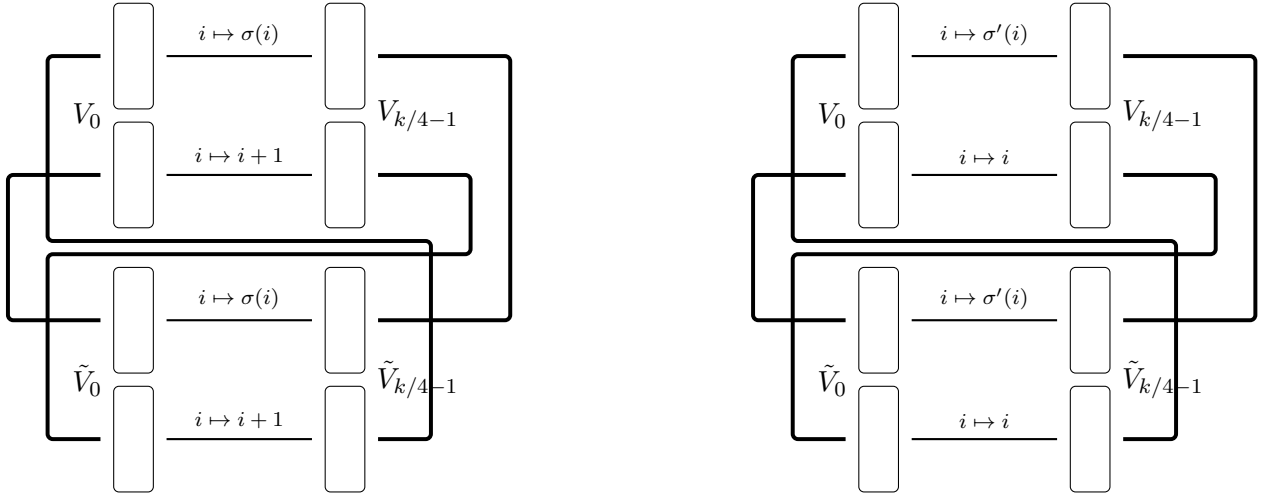
\begin{figure}
	\centering
    \begin{subfigure}{0.4\textwidth}
        \centering
        \begin{tikzpicture}
        	\draw[rectangle, rounded corners=2pt] (0, 0) rectangle (15pt, 40pt);
        	\draw[rectangle, rounded corners=2pt] (0, 45pt) rectangle (15pt, 85pt);
        	\draw[rectangle, rounded corners=2pt] (0, 100pt) rectangle (15pt, 140pt);
        	\draw[rectangle, rounded corners=2pt] (0, 145pt) rectangle (15pt, 185pt);

        	\draw[rectangle, rounded corners=2pt] (80pt, 0) rectangle (95pt, 40pt);
        	\draw[rectangle, rounded corners=2pt] (80pt, 45pt) rectangle (95pt, 85pt);
        	\draw[rectangle, rounded corners=2pt] (80pt, 100pt) rectangle (95pt, 140pt);
        	\draw[rectangle, rounded corners=2pt] (80pt, 145pt) rectangle (95pt, 185pt);

        	\node at (-10pt, 142.5pt) {$V_0$};
        	\node at (-10pt, 42.5pt) {$\tilde V_0$};
        	\node at (115pt, 142.5pt) {$V_{k/4-1}$};
        	\node at (115pt, 42.5pt) {$\tilde V_{k/4-1}$};

        	\draw (20pt, 20pt) edge [thick] node [above] {\scriptsize $i\mapsto i+1$} (75pt, 20pt);
        	\draw (20pt, 65pt) edge [thick] node [above] {\scriptsize $i\mapsto \sigma(i)$} (75pt, 65pt);
        	\draw (20pt, 120pt) edge [thick] node [above] {\scriptsize $i\mapsto i+1$} (75pt, 120pt);
        	\draw (20pt, 165pt) edge [thick] node [above] {\scriptsize $i\mapsto \sigma(i)$} (75pt, 165pt);

        	\draw [line width=1.5pt, rounded corners=2pt] (100pt, 165pt) -- (150pt, 165pt) -- (150pt, 65pt) -- (100pt, 65pt);
        	\draw [line width=1.5pt, rounded corners=2pt] (100pt, 20pt) -- (120pt, 20pt) -- (120pt, 95pt) -- (-25pt, 95pt) -- (-25pt, 165pt) -- (-5pt, 165pt);
        	\draw [line width=1.5pt, rounded corners=2pt] (100pt, 120pt) -- (135pt, 120pt) -- (135pt, 90pt) -- (-25pt, 90pt) -- (-25pt, 20pt) -- (-5pt, 20pt);
        	\draw [line width=1.5pt, rounded corners=2pt] (-5pt, 65pt) -- (-40pt, 65pt) -- (-40pt, 120pt) -- (-5pt, 120pt);
        \end{tikzpicture}
        \caption{\scriptsize Vertex $i$ in top of $V_0$ connects to vertex $i+2$ in top of $V_0$.}\label{fig_sub_fig_Pyes}
    \end{subfigure}
    \hfill
    \begin{subfigure}{0.4\textwidth}
		\centering
       \begin{tikzpicture}
        	\draw[rectangle, rounded corners=2pt] (0, 0) rectangle (15pt, 40pt);
        	\draw[rectangle, rounded corners=2pt] (0, 45pt) rectangle (15pt, 85pt);
        	\draw[rectangle, rounded corners=2pt] (0, 100pt) rectangle (15pt, 140pt);
        	\draw[rectangle, rounded corners=2pt] (0, 145pt) rectangle (15pt, 185pt);

        	\draw[rectangle, rounded corners=2pt] (80pt, 0) rectangle (95pt, 40pt);
        	\draw[rectangle, rounded corners=2pt] (80pt, 45pt) rectangle (95pt, 85pt);
        	\draw[rectangle, rounded corners=2pt] (80pt, 100pt) rectangle (95pt, 140pt);
        	\draw[rectangle, rounded corners=2pt] (80pt, 145pt) rectangle (95pt, 185pt);

        	\node at (-10pt, 142.5pt) {$V_0$};
        	\node at (-10pt, 42.5pt) {$\tilde V_0$};
        	\node at (115pt, 142.5pt) {$V_{k/4-1}$};
        	\node at (115pt, 42.5pt) {$\tilde V_{k/4-1}$};

        	\draw (20pt, 20pt) edge [thick] node [above] {\scriptsize $i\mapsto i$} (75pt, 20pt);
        	\draw (20pt, 65pt) edge [thick] node [above] {\scriptsize $i\mapsto \sigma'(i)$} (75pt, 65pt);
        	\draw (20pt, 120pt) edge [thick] node [above] {\scriptsize $i\mapsto i$} (75pt, 120pt);
        	\draw (20pt, 165pt) edge [thick] node [above] {\scriptsize $i\mapsto \sigma'(i)$} (75pt, 165pt);

        	\draw [line width=1.5pt, rounded corners=2pt] (100pt, 165pt) -- (150pt, 165pt) -- (150pt, 65pt) -- (100pt, 65pt);
        	\draw [line width=1.5pt, rounded corners=2pt] (100pt, 20pt) -- (120pt, 20pt) -- (120pt, 95pt) -- (-25pt, 95pt) -- (-25pt, 165pt) -- (-5pt, 165pt);
        	\draw [line width=1.5pt, rounded corners=2pt] (100pt, 120pt) -- (135pt, 120pt) -- (135pt, 90pt) -- (-25pt, 90pt) -- (-25pt, 20pt) -- (-5pt, 20pt);
        	\draw [line width=1.5pt, rounded corners=2pt] (-5pt, 65pt) -- (-40pt, 65pt) -- (-40pt, 120pt) -- (-5pt, 120pt);
        \end{tikzpicture}
		\caption{\scriptsize Vertex $i$ in top of $V_0$ connects to vertex $i$ in top of $V_0$.}\label{fig_sub_fig_Pno}
     \end{subfigure}
	\caption{Figure~\ref{fig_sub_fig_Pyes} and~\ref{fig_sub_fig_Pno} show the new graph $G_{\omc}$ when the final matching of $G$ is $\Y$ and $\N$ respectively.}\label{fig_proof_main_informal}
\end{figure}

Suppose the matching from top $m$ vertices of $V_0$ to top $m$ vertices of $V_{\n}$ is $\sigma$.
Observe that if $G$ has final matching $\Y$, then the $i$-th vertex in bottom of $V_0$ connects to $(i+1)$-th vertex in bottom of $V_{k/4-1}$, which connects to $(i+1)$-th vertex in bottom of $\tilde V_0$, and $(i+2)$-th vertex in bottom of $\tilde V_{k/4-1}$.
Then it connects to $(i+2)$-th vertex in the top of $V_0$, which connects to the $\sigma(i+2)$-th vertex in the top of $V_{k/4-1}$, and then the $\sigma(i+2)$-th vertex in the top of $\tilde V_{k/4-1}$.
It then connects to the $(i+2)$-th vertex in the top of $\tilde V_0$, and finally, to the $(i+2)$-th vertex in the bottom of $V_0$.
Therefore, since $m$ is odd, all vertices are connected, and $G_{\omc}$ is simply a Hamiltonian cycle.
On the other hand, if $G$ has final matching $\N$, one may verify similarly that $G_{\omc}$ has $n/k$ vertex-disjoint $k$-cycles.
This means that $\Pi$ has probability of success $\geq (1/2+(k/n)^2)$. 
On the other hand, $\norm{\Pi}=\norm{\prot}$, and by \Cref{thm:FMT},  
\[
\norm{\Pi} \geq 2^{-\mathcal{O}(r)} \cdot m^{1 - \frac{200}{c}}.
\]
By the fact that $\n =  (2c)^{r}$ and $k/4 \leq (2(c + 1))^r$, we get $c+1\geq\frac{1}{2}(k/4)^{1/r}  \geq \frac{1}{8} k^{1/r}$. Next, if $k^{1/r} < 10^4$, the result is trivial. Assuming, otherwise, we have $c \geq \frac{1}{16}k^{1/r}$, and thus,  
\begin{align*}
	\norm{\prot} &\geq 2^{-\mathcal{O}(r)} \cdot m^{1 - \frac{200}{c}} \\
	&\geq 2^{-\mathcal{O}(r)} \cdot m^{1 - 10^4 \cdot k^{-1/r}}.
\end{align*}
This proves the theorem when $k$ is a multiple of $4$.

For even $k$ that are not a multiple of $4$,  we make the same reduction for $k'=k - 2$ and $m=n/k$.
Then we split each vertex in the top of $V_0$ into a chain of length $3$.
The new graph has $m\cdot k'+2m=mk=n$ vertices.
Moreover,  if the graph is a Hamiltonian cycle before the split, then it remains a Hamiltonian cycle; if the graph was a disjoint union of $k'$-cycles, then the split adds $2$ to the length of each cycle. 
Hence, this proves the theorem for all even $k$.
\end{proof}

\begin{remark}
\label{rem:mainformal}
Observe that implicit in our proof above is a distribution $\mathcal{D}_{\mathsf{OMC}}$ for the $\OMC_{n, k}$ problem that is ``hard'' for $r$-rounds protocols. It is notable that $\mathcal{D}_{\mathsf{OMC}}$ can even be taken to be the uniform distribution that with probability $\frac{1}{2}$, is the uniform distribution over all Hamiltonian cycles and with probability $\frac{1}{2}$, is the uniform distribution on all graphs with $(n/k)$ vertex-disjoint $k$-cycles provided  that before running the protocol $\pi$, Alice and Bob permute the identity of the vertices randomly (using public randomness).
\end{remark}
\begin{remark}
\label{rem:arb_mat}
	By plugging in different pairs of $\Y$ and $\N$, we can prove lower bounds for distinguishing between different configurations of cycles in a graph.
	For example, by setting $\Y$ to the matching such that \[
		\Y(i)=\begin{cases}
			i-3 & \textrm{if $4\mid i$} \\
			i+1 & \textrm{otherwise}
		\end{cases}
	\]
	and $\N(i)=i$, and assuming $m=n/k$ is a multiple of four, the same argument proves that distinguishing between union of $2k$-cycles and $k$-cycles is hard.
\end{remark}


\newcommand{\rP}{\rv{P}}

\section{Lower Bound for One Round (\Cref{lemma:oneround})}\label{sec:single-block}

We prove~\Cref{lemma:oneround}, restated below, in this section. 

\newcommand{\contlemoneroundmarginalelltwo}{Let $m > 0$ and $G$ be a graph sampled from $\distnest_{1, c, (m,m)}$ with $v_0$ being the only vertex in the first layer. For any $1$-round $(1, c, (m,m))$-protocol $\Pi$ with communication $C$, we must have
	\[
		\Ex_{\Pi,M_{\odd}}\norm{\rv G_{0\to 2c}(v_0) \mid \rv{\Pi}=\Pi,\rv{M}_{\odd}=M_{\odd}}_2^2 \leq \frac{1}{m}+\left(\frac{40C}{m}\right)^{c/3}.
	\]}
	
\begin{restate}[\Cref{lemma:oneround}]
	\lemmaoneround
\end{restate}
\begin{proof} 

To start, note that we can assume without loss of generality that $s_1 = s_2$. Indeed, if we show the lemma for $s_1 = s_2$, the lemma for $s_1 < s_2$ follows as any protocol for the latter case can be use to construct a protocol for the former case by Alice deleting $s_2 - s_1$ vertices from the last layer. We can also assume that $\norm{\Pi} = C$. In order to show the lemma, we show that if $\Y, \N$ differ only in one coordinate, then:
\[
\tvd{ \distributions{ \rv{\Pi}(G) }{ G \sim \distnest_{s, c, S}^{\Y} } }{ \distributions{ \rv{\Pi}(G) }{ G \sim \distnest_{s, c, S}^{\N} } } \leq s_2 \cdot \left(\frac{40C}{s_2-s}\right)^{c/6} .
\]
The lemma then follows by triangle inequality as $\Y, \N$ differ in at most $s$ coordinates and $s<s_2$. Thus, there exists a sequence of injections $\Z_0,\cdots,\Z_t$ for $t\leq 2s$ such that $\Z_0=\Y,\Z_t=\N$, $\Z_i$ and $\Z_{i+1}$ differ in exactly one coordinate. 

Without loss of generality, assume that $\Y, \N$ differ in the coordinate corresponding to vertex $v_0 \in V_0$. Denote the other $s-1$ vertices by $S=V_0\setminus \{v_0\}$. Throughout this section, we use $M_A$ to denote Alice's input and $M_B$ to denote Bob's input. We abuse notation slightly and use $\Pi$ to denote Alice's message and $\Pi^{\out}$ to denote Bob's output. The following lemma is the core of the proof. It bounds the $\ell_2$-norm of the distribution of the image of $v_0$ after one message.

\begin{lemma}[$\ell_2$-norm of the $\rv{G}_{0\to 2c}(v_0)$-vector]\label{lem_one_round_marginal_ell_2}
	\contlemoneroundmarginalelltwo
\end{lemma}

The lemma is proved in the next subsection. Here, we use it to prove \Cref{lemma:oneround}. As Bob's output is determined by his input and Alice's message, we have:
\begin{align*}
&\tvd{ \distributions{ \rv{\Pi}^{\out}(G) }{ G \sim \distnest_{s, c, S}^{\Y} } }{ \distributions{ \rv{\Pi}{\out}(G) }{ G \sim \distnest_{s, c, S}^{\N} } } \\
&\leq \tvd{ \distributions{ \rv{\Pi}(G), \rv{M}_{\odd} }{ G \sim \distnest_{s, c, S}^{\Y} } }{ \distributions{ \rv{\Pi}(G), \rv{M}_{\odd} }{ G \sim \distnest_{s, c, S}^{\N} } }   \\
&= \frac{1}{2} \cdot \sum_{\Pi, M_{\odd}} \left\lvert{ \Pr_{ \distnest_{s, c, S} } \left( \rv{\Pi}(G) = \Pi, \rv{M}_{\odd} = M_{\odd} \mid G_{0 \to 2c} = \Y \right) - \Pr_{ \distnest_{s, c, S} } \left( \rv{\Pi}(G) = \Pi, \rv{M}_{\odd} = M_{\odd} \mid G_{0 \to 2c} = \N \right) }\right\rvert .
\end{align*}
To continue, define $W$ be the event that $\rv G_{0\to 2c}(v)=\Y(v)(=\N(v))$ for all $v\neq v_0$. Thus, the event $G_{0 \to 2c} = \Y$ is equivalent to $W$ and $G_{0 \to 2c}(v_0) = \Y(v_0)$ and $G_{0 \to 2c} = \N$ is equivalent to $W$ and $G_{0 \to 2c}(v_0) = \N(v_0)$. Also, for any $a,b \in \{0\} \cup [2c]$, and any set $S\subseteq V_{a}$, we denote by $P_{[a, b]|S}$, the \emph{paths} from $S$ to the vertices in $V_b$. By Bayes's theorem, we have:
\begin{align*}
&\tvd{ \distributions{ \rv{\Pi}^{\out}(G) }{ G \sim \distnest_{s, c, S}^{\Y} } }{ \distributions{ \rv{\Pi}{\out}(G) }{ G \sim \distnest_{s, c, S}^{\N} } } \\
&\leq \frac{s_2}{2} \cdot \E_{\rv{\Pi}, \rv{M}_{\odd} \mid W} \left\lvert{ \Pr_{ \distnest_{s, c, S} } \left( \rv{G}_{0 \to 2c}(v_0) = \Y(v_0) \mid W, \Pi, M_{\odd} \right) - \Pr_{ \distnest_{s, c, S} } \left( \rv{G}_{0 \to 2c}(v_0) = \N(v_0) \mid W, \Pi, M_{\odd} \right) }\right\rvert \\
&\leq s_2 \cdot \max_{v_{2c} \in V_{2c}} \E_{\rv{\Pi}, \rv{M}_{\odd} \mid W} \left\lvert{ \Pr_{ \distnest_{s, c, S} } \left( \rv{G}_{0 \to 2c}(v_0) = v_{2c} \mid W, \Pi, M_{\odd} \right) - \frac{1}{s_2-s+1} }\right\rvert \tag{Triangle inequality} \\
&\leq s_2 \cdot \max_{v_{2c} \in V_{2c}} \E_{\rv{\Pi}, \rv{M}_{\odd}, \rv{P}_{[0,2c]\mid S} \mid W} \left\lvert{ \Pr_{ \distnest_{s, c, S} } \left( \rv{G}_{0 \to 2c}(v_0) = v_{2c} \mid W, \Pi, M_{\odd}, P_{[0,2c]\mid S} \right) - \frac{1}{s_2-s+1} }\right\rvert \tag{Triangle inequality} .
\end{align*}
To continue, note that the max of expectation is at most the expectation of the max. We have:
\begin{align*}
&\tvd{ \distributions{ \rv{\Pi}^{\out}(G) }{ G \sim \distnest_{s, c, S}^{\Y} } }{ \distributions{ \rv{\Pi}{\out}(G) }{ G \sim \distnest_{s, c, S}^{\N} } } \\
&\leq s_2 \cdot \E_{\rv{\Pi}, \rv{M}_{\odd}, \rv{P}_{[0,2c]\mid S} \mid W} \Bracket{ \max_{v_{2c} \in V_{2c}} \left\lvert{ \Pr_{ \distnest_{s, c, S} } \left( \rv{G}_{0 \to 2c}(v_0) = v_{2c} \mid W, \Pi, M_{\odd}, P_{[0,2c]\mid S} \right) - \frac{1}{s_2-s+1} }\right\rvert } \\
&\leq s_2 \cdot \sqrt{ \E_{\rv{\Pi}, \rv{M}_{\odd}, \rv{P}_{[0,2c]\mid S} \mid W} \Bracket{ \max_{v_{2c} \in V_{2c}} \left\lvert{ \Pr_{ \distnest_{s, c, S} } \left( \rv{G}_{0 \to 2c}(v_0) = v_{2c} \mid W, \Pi, M_{\odd}, P_{[0,2c]\mid S} \right) - \frac{1}{s_2-s+1} }\right\rvert^2 } } \tag{Jensen's inequality} .
\end{align*}
Next, we note that if $v = (v_1, v_2, \cdots, v_k)$ is a probability vector and $k \geq 2$, then $\sum_i v_i^2 \geq \frac{1}{k} + \max_i  (\frac{1}{k} - v_i)^2$. This gives:
\begin{align*}
&\tvd{ \distributions{ \rv{\Pi}^{\out}(G) }{ G \sim \distnest_{s, c, S}^{\Y} } }{ \distributions{ \rv{\Pi}{\out}(G) }{ G \sim \distnest_{s, c, S}^{\N} } } \\
&\leq s_2 \cdot \sqrt{ \E_{\rv{\Pi}, \rv{M}_{\odd}, \rv{P}_{[0,2c]\mid S} \mid W} \Bracket{ \norm{ \rv{G}_{0 \to 2c}(v_0) \mid W, \Pi, M_{\odd}, P_{[0,2c]\mid S} }^2_2 - \frac{1}{s_2-s+1} } } .
\end{align*}

To finish, let us analyze $\Pi$ on the distribution $\distnest_{s, c, S}$, {\em i.e.}, without the condition on the final matching $\rv G_{0\to 2c}$. Conditioned on $W, P_{[0,2c]\mid S}$, the remaining graph simply becomes a random graph sampled from $\distnest_{1, c, (s_2-s+1, s_2-s+1)}$. By \Cref{lem_one_round_marginal_ell_2}, we have:
\[
\tvd{ \distributions{ \rv{\Pi}^{\out}(G) }{ G \sim \distnest_{s, c, S}^{\Y} } }{ \distributions{ \rv{\Pi}{\out}(G) }{ G \sim \distnest_{s, c, S}^{\N} } } \leq s_2 \cdot \left(\frac{40C}{s_2-s+1}\right)^{c/6} \leq s_2 \cdot \left(\frac{40C}{s_2-s}\right)^{c/6} .
\]

\end{proof}

\subsection{Block Matching of a Single Vertex After the First Round (\Cref{lem_one_round_marginal_ell_2})}

In this subsection, we prove the $\ell_2$ bound on the matching of a single vertex after one message.
\begin{restate}[\Cref{lem_one_round_marginal_ell_2}]
	\contlemoneroundmarginalelltwo
\end{restate}
\noindent
The rest of this subsection and the next one is then dedicated to the proof of~\Cref{lem_one_round_marginal_ell_2}. 

\paragraph{Notation.} Let $T=\{t_1,\ldots,t_l\}$ be a subset of $\{0,\ldots,2c\}$, we use $V_T$ to denote $V_{t_1}\times V_{t_2}\times\cdots\times V_{t_l}$, i.e., the set of all $l$-tuples of vertices with one vertex from each layer $V_{t_i}$.

Fix a message $\Pi$, and consider the distribution of Alice's input matchings conditioned on $\rv \Pi=\Pi$.
Let $\bv=(v_0,v_1,\ldots,v_{2c})\in V_{\{0,\ldots,2c\}}$ be a sequence of vertices (i.e., $v_i\in V_i$), and define: 
\[
	p(\bv\mid \Pi):=\Pr[\forall i\in\{1,\ldots,c\},\rv G_{2i}(v_{2i-1})=v_{2i}\mid \rv \Pi=\Pi],
\]
that is, $p(\bv\mid\Pi)$ denotes the probability that the sequence $\bv$ is consistent with Alice's input conditioned on her sending the message $\Pi$.

Likewise, for any subset of Alice's matchings $S\subseteq \{1,\ldots,c\}$, we define $p_S$ as follows: For any $\bv_S\in V_{(2S-1)\cup 2S}$ (where $2S-1 := \set{2i-1 \mid i \in S}$ and similarly $2S := \set{2i \mid i \in S}$, i.e., the sets of endpoints of 
Alice's matchings in $S$),
\[
	p_S(\bv_S\mid \Pi):=\Pr[\forall i\in S,\rv G_{2i}(v_{2i-1})=v_{2i}\mid\rv\Pi=\Pi].
\]
It is easy to verify the following claim (it simply calculates $p_S(\bv_S\mid \Pi)$ by going over all possible choices for edges in one extra matching of Alice which is not in $S$).
\begin{claim}\label{cl_delta_sum}
	For any $i\notin S$, $v_{2i}\in V_{2i}$ and $\bv_S\in V_{(2S-1)\cup 2S}$, we have
	\[
		p_S(\bv_S\mid \Pi)=\sum_{v_{2i-1}\in V_{2i-1}} p_{S\cup \{i\}}(\bv_S,v_{2i-1},v_{2i}\mid \Pi).
	\]
	Similarly, for any $v_{2i-1}\in V_{2i-1}$,
	\[
		p_S(\bv_S\mid \Pi)=\sum_{v_{2i}\in V_{2i}} p_{S\cup \{i\}}(\bv_S,v_{2i-1},v_{2i}\mid \Pi).
	\]
\end{claim}

In the following lemma, we give a formula for the LHS of~\Cref{lem_one_round_marginal_ell_2} in terms of $p_S$. We then use this formula in the next subsection to bound 
the LHS and conclude the proof.

\begin{lemma}\label{lem_sec_moment}
	For any message $\Pi$, 
	\begin{align*}
		\E_{M_B}\norm{\rv G_{0 \to 2c}(v_0)\mid \rv \Pi=\Pi,\rv M_{\odd}=M_B]}^2_2 =\frac{1}{m} +\frac{1}{m(m-1)^{c-1}}\cdot\sum_{S\subseteq \{1,\ldots,c\}}\sum_{\bv_S\in V_{(2S-1)\cup 2S}} (-1)^{c-|S|}p_{S}(\bv_S\mid \Pi)^2.
	\end{align*}
\end{lemma}
\begin{proof}
	Note that, for all $v_{2c} \in V_{2c}$, we have:
	\[
		\Pr[\rv G_{0 \to 2c}(v_0)=v_{2c}\mid \rv \Pi=\Pi,\rv M_{\odd}=M_{\odd}]=\sum_{\bv\in \{v_0\}\times V_{\{1,\ldots,2c-1\}}\times \{v_{2c}\}}p(\bv\mid \Pi)\cdot \prod_{i=1}^c\bOne_{G_{2i-1}(v_{2i-2})=v_{2i-1}}.
	\]

	Since Alice and Bob's inputs are independent, we have
	\begin{align*}
		&\, \E_{M_{\odd}}\left[\Pr[\rv G_{0 \to 2c}(v_0)=v_{2c}\mid \rv \Pi=\Pi,\rv M_{\odd}=M_{\odd}]^2\right] \\
		=&\, \sum_{\bv,\bu\in\{v_0\}\times V_{\{1,\ldots,2c-1\}}\times \{v_{2c}\}}p(\bv\mid \Pi)p(\bu\mid \Pi)\cdot \E_{M_{\odd}}\left[\prod_{i=1}^c\bOne_{G_{2i-1}(v_{2i-2})=v_{2i-1}\wedge G_{2i-1}(u_{2i-2})=u_{2i-1}}\right] \\
		=&\, \sum_{\bv,\bu\in\{v_0\}\times V_{\{1,\ldots,2c-1\}}\times \{v_{2c}\}}p(\bv\mid \Pi)p(\bu\mid \Pi)\cdot \prod_{i=1}^c\Pr[\rv G_{2i-1}(v_{2i-2})=v_{2i-1}\wedge \rv G_{2i-1}(u_{2i-2})=u_{2i-1}].
	\end{align*}
	Observe that: $(i)$ for $(v_{2i-2},v_{2i-1})=(u_{2i-2},u_{2i-1})$,
	\[
		\Pr[\rv G_{2i-1}(v_{2i-2})=v_{2i-1}\wedge \rv G_{2i-1}(u_{2i-2})=u_{2i-1}]=\frac{1}{m};
	\]
	and $(ii)$ for $v_{2i-2}\neq u_{2i-2}$ and $v_{2i-1}\neq u_{2i-1}$,
	\[
		\Pr[\rv G_{2i-1}(v_{2i-2})=v_{2i-1}\wedge \rv G_{2i-1}(u_{2i-2})=u_{2i-1}]=\frac{1}{m(m-1)};
	\]
	otherwise, it is equal to $0$.
	Since $u_0=v_0$ and $u_{2c}=v_{2c}$, we have
	\begin{multline}\label{eqn_var}
		\sum_{v_{2c}\in V_{2c}}\E_{M_{\odd}}\left[\Pr[\rv G_{0 \to 2c}(v_0)=v_{2c}\mid \rv \Pi=\Pi,\rv M_{\odd}=M_{\odd}]^2\right] \\
		=\sum_{T\subseteq\{1,\ldots,c-1\}}\sum_{\stackrel{\bu,\bv:\forall i\notin T, v_{2i}=u_{2i},v_{2i+1}=u_{2i+1}}{\forall i\in T, v_{2i}\neq u_{2i},v_{2i+1}\neq u_{2i+1}}}\frac{p(\bv\mid \Pi)p(\bu\mid \Pi)}{m^c(m-1)^{|T|}}.
	\end{multline}

	Now, let us fix $T$, $v_{2i}(=u_{2i})$ and $v_{2i+1}(=u_{2i+1})$ for all $i\notin T$, then take the sum over all $v_{2i}\neq u_{2i},v_{2i+1}\neq u_{2i+1}$ for all $i\in T$.
	By inclusion-exclusion, we have
	\begin{align*}
		&\quad \sum_{\forall i\in T, v_{2i}\neq u_{2i},v_{2i+1}\neq u_{2i+1}}p(\bv\mid \Pi)p(\bu\mid \Pi)\\
		&=\sum_{T'\subseteq 2T\cup (2T+1)}\sum_{\forall j\in(2T\cup (2T+1))\setminus T', v_j=u_j}\sum_{\forall j\in T',v_j,u_j}(-1)^{2|T|-|T'|}\cdot p(\bv\mid \Pi)p(\bu\mid \Pi)\\
		&=\sum_{T'\subseteq 2T\cup (2T+1)}(-1)^{|T'|}\cdot\sum_{\forall j\in(2T\cup (2T+1))\setminus T', v_j}\left(\sum_{\forall j\in T',v_j}p(\bv\mid \Pi)\right)^2.
	\end{align*}
	Plug it into the RHS of \eqref{eqn_var}, we get
	\begin{align*}
		&\, \sum_{T\subseteq\{1,\ldots,c-1\}}\sum_{\stackrel{\bu,\bv:\forall i\notin T, v_{2i}=u_{2i},v_{2i+1}=u_{2i+1}}{\forall i\in T, v_{2i}\neq u_{2i},v_{2i+1}\neq u_{2i+1}}}\frac{p(\bv\mid \Pi)p(\bu\mid \Pi)}{m^c(m-1)^{|T|}} \\
		=&\, \sum_{T\subseteq\{1,\ldots,c-1\}}\sum_{\forall i\notin T, v_{2i},v_{2i+1}}\sum_{T'\subseteq2T\cup (2T+1)}\frac{(-1)^{|T'|}}{m^c(m-1)^{|T|}}\cdot \sum_{\forall j\in (2T\cup(2T+1))\setminus T',v_j}\left(\sum_{\forall j\in  T',v_j}p(\bv\mid \Pi)\right)^2  \\
		=&\, \sum_{T\subseteq\{1,\ldots,c-1\}}\sum_{T'\subseteq2T\cup (2T+1)}\frac{(-1)^{|T'|}}{m^c(m-1)^{|T|}}\cdot\sum_{\forall j\in\{1,\ldots,2c\}\setminus T',v_j} \left(\sum_{\forall j\in  T',v_j}p(\bv\mid \Pi)\right)^2  \\
		=&\, \sum_{T'\subseteq\{2,\ldots,2c-1\}}\left(\sum_{T\supset \lfloor T'/2\rfloor}\frac{(-1)^{|T'|}}{m^c(m-1)^{|T|}}\right)\cdot\left(\sum_{\forall j\in\{1,\ldots,2c\}\setminus  T',v_j} \left(\sum_{\forall j\in  T',v_j}p(\bv\mid \Pi)\right)^2\right), \\
		\intertext{where $\lfloor T'/2\rfloor$ denotes the set $\{\lfloor i/2\rfloor: i\in T'\}$, and it is}
		=&\, \sum_{T'\subseteq\{2,\ldots,2c-1\}}\frac{(-1)^{|T'|}}{m^{|\lfloor T'/2\rfloor|+1}(m-1)^{c-1}}\cdot\sum_{\forall j\in\{1,\ldots,2c\}\setminus T',v_j} \left(\sum_{\forall j\in  T',v_j}p(\bv\mid \Pi)\right)^2 \\
		\intertext{by Claim~\ref{cl_delta_sum}, it is}
		=&\, \sum_{T'\subseteq\{2,\ldots,2c-1\}}\frac{(-1)^{|T'|}}{m^{|\lfloor T'/2\rfloor|+1}(m-1)^{c-1}}\cdot\sum_{\forall i\in\{1,\ldots,c\}\setminus \lceil T'/2\rceil,v_{2i-1},v_{2i}} p_{\{1,\ldots,c\}\setminus\lceil T'/2\rceil}(\bv\mid \Pi)^2\cdot m^{|T'|} \\
		\intertext{let $S:=\{1,\ldots,c\}\setminus \lceil T'/2\rceil$, it is}
		=&\, \frac{1}{m(m-1)^{c-1}}\cdot\sum_{S\subseteq \{1,\ldots,c\}}\sum_{\bv_S\in V_{(2S-1)\cup 2S}} p_{S}(\bv_S\mid \Pi)^2\cdot \left(\sum_{T'\subseteq \{2,\ldots,2c-1\}:\lceil T'/2\rceil=\{1,\ldots,c\}\setminus S}\frac{(-m)^{|T'|}}{m^{|\lfloor T'/2\rfloor|}}\right).
	\end{align*}
	We claim that
	\begin{equation}\label{eqn_alt_sum}
		\sum_{T'\subseteq \{2,\ldots,2c-1\}:\lceil T'/2\rceil=\{1,\ldots,c\}\setminus S}\frac{(-m)^{|T'|}}{m^{|\lfloor T'/2\rfloor|}}=\begin{cases}
			(m-1)^{c-1}+(-1)^c & \textrm{if $S=\emptyset$,} \\
			(-1)^{c-|S|} & \textrm{otherwise.}
		\end{cases}
	\end{equation}
	Since $p_{\emptyset}=1$, it is easy to verify that plugging Equation~\eqref{eqn_alt_sum} proves the lemma.

	To prove \eqref{eqn_alt_sum}, define
	\[
		\alpha(m, c):=\sum_{T'\subseteq \{2,\ldots,2c-1\}:\lceil T'/2\rceil=\{1,\ldots,c\}}\frac{(-m)^{|T'|}}{m^{|\lfloor T'/2\rfloor|}},
	\]
	\[
		\beta(m, c):=\sum_{T'\subseteq \{2,\ldots,2c\}:\lceil T'/2\rceil=\{1,\ldots,c\}}\frac{(-m)^{|T'|}}{m^{|\lfloor T'/2\rfloor|}},
	\]
	by symmetry, $\beta(c)$ is also equal to
	\[
		\beta(m, c)=\sum_{T'\subseteq \{1,\ldots,2c-1\}:\lceil T'/2\rceil=\{1,\ldots,c\}}\frac{(-m)^{|T'|}}{m^{|\lfloor T'/2\rfloor|}},
	\]
	and
	\[
		\gamma(m, c):=\sum_{T'\subseteq \{1,\ldots,2c\}:\lceil T'/2\rceil=\{1,\ldots,c\}}\frac{(-m)^{|T'|}}{m^{|\lfloor T'/2\rfloor|}}.
	\]
	For $S=\emptyset$, the left-hand-side of \eqref{eqn_alt_sum} is simply $\alpha(m, c)$.
	Otherwise, suppose $S=\{i_1,\ldots,i_s\}$, where $1\leq i_1<i_2<\cdots<i_s\leq c$. 
	Then the sum in \eqref{eqn_alt_sum} only has $T'$ such that $2i_1-1,2i_1,2i_2-1,2i_2,\ldots\notin T'$.
	These elements divide the set $\{2,\ldots,c-1\}$ into $s+1$ chunks.
	Moreover, the constraint $\lceil T'/2\rceil=\{1,\ldots,c\}\setminus S$ and the value $\frac{(-m)^{|T'|}}{m^{|\lfloor T'/2\rfloor|}}$ are independent in all chunks.
	Therefore, the LHS is equal to
	\[
		\beta(m, i_1-1)\cdot \left(\prod_{l=1}^{s-1}\gamma(m, i_{l+1}-i_l-1)\right)\cdot \beta(m, c-i_s).
	\]

	At last, we calculate $\alpha,\beta$ and $\gamma$.
	First, note $\beta(m, c)=-\beta(m, c-1)$. 
	To see this, since $\lceil T'/2\rceil=\{1,\ldots,c\}$, at least one of $2c-1$ and $2c$ must be in $T'$.
	There are three cases: (1) $2c-1,2c\in T'$; (2) $2c-1\in T'$ and $2c\notin T'$; (3) $2c-1\notin T'$ and $2c\in T'$.
	Case (1) and case (2) cancel, since adding $2c$ to the set increases both $|T'|$ and $|\lfloor T'/2\rfloor|$ by one, and the two cases have opposite signs by definition.
	Case (3) is equal to $-\beta(m, c-1)$.
	For the same reason, we have $\gamma(m, c)=-\gamma(m, c-1)$.
	By the fact that $\beta(m, 1)=-1$ and $\gamma(m, 1)=-1$, we have $\beta(m, c)=\gamma(m, c)=(-1)^c$.
	This shows that the LHS when $S\neq \emptyset$ is 
	\[
		(-1)^{i_1-1}\cdot \prod_{l=1}^{s-1}(-1)^{i_{l+1}-i_l-1}\cdot (-1)^{c-i_s}=(-1)^{c-|S|},
	\]
	proving \eqref{eqn_alt_sum} for $S\neq \emptyset$.

	On the other hand, for $\alpha(m, c)$, $2c-1$ must be in $T'$ (otherwise $c\notin \lceil T'/2\rceil$), and at least one of $2c-3$ and $2c-2$ must be in $T'$.
	There are again three cases for $2c-1,2c-2,2c-3$: (1) $2c-3,2c-2,2c-1\in T'$; (2) $2c-3,2c-1\in T'$ and $2c-2\notin T'$; (3) $2c-3\notin T'$ and $2c-2,2c-1\in T'$. 
	Case (1) is equal to $m\cdot \alpha(m, c-1)$; Case (2) is equal to $-\alpha(m, c-1)$; Case (3) is equal to $m\cdot \beta(m, c-2)$.
	Thus, we have the following recurrence:
	\[
		\alpha(m, c)=(m-1)\alpha(m, c-1)+m\cdot (-1)^c,
	\]
	with $\alpha(m, 1)=0$ and $\alpha(m, 2)=m$.
	Expanding the recurrence, we get
	\begin{align*}
		\alpha(m, c)&=m(-1)^c+(m-1)m(-1)^{c-1}+\cdots+(m-1)^{c-2}m(-1)^{2} \\
		&=m(-1)^c\cdot\frac{1-(1-m)^{c-1}}{1-(1-m)} \\
		&=(m-1)^{c-1}+(-1)^c.
	\end{align*}
	This completes the proof of the lemma.
\end{proof}


The following lemma gives an upper bound on the expectation of each summand on the RHS of~\Cref{lem_sec_moment}. 

\begin{lemma}\label{lem_delta_S}
	For any $S\subseteq \{1,\ldots,c\}$, 
	\[
		\E_{\Pi}\left[\sum_{\bv_S\in V_{(2S-1)\cup 2S}}p_S(\bv_S\mid \Pi)^2\right]\leq \paren{4\sqrt{2(C+\log m)m}+6}^{|S|}.
	\]
	where we define $\Gamma(m,C) := 4\sqrt{2(C+\log m)m} < m/4$. 
\end{lemma}

We prove this lemma in the next subsection and for now, show that how this lemma combined with~\Cref{lem_sec_moment} implies~\Cref{lem_one_round_marginal_ell_2}.
	\begin{proof}[Proof of \Cref{lem_one_round_marginal_ell_2}]
		We have, 
		\begin{align*}
			\text{LHS of~\Cref{lem_one_round_marginal_ell_2}} &= \frac{1}{m}+\frac{1}{m(m-1)^{c-1}}\cdot\sum_{S\subseteq \{1,\ldots,c\}}(-1)^{c-|S|}\cdot \E_{\Pi}\left[\sum_{\bv_S\in V_{(2S-1)\cup 2S}} p_{S}(\bv_S\mid \Pi)^2\right] 
			\tag{by the formula of~\Cref{lem_sec_moment}} \\
			&\leq \frac{1}{m}+\frac{1}{m(m-1)^{c-1}}\cdot\sum_{S\subseteq \{1,\ldots,c\}}\left(\Gamma(m,C)+6\right)^{|S|} \tag{by the bound in~\Cref{lem_delta_S} for each summand} \\
			&=\frac{1}{m}+\frac{1}{m(m-1)^{c-1}}\cdot\left(\Gamma(m,C)+7\right)^c \\
			&\leq \frac{1}{m}+\left(\frac{40C}{m}\right)^{c/3},
		\end{align*}
		for sufficiently large $m$ by the choice of $\Gamma(m,C) := 4\sqrt{2(C+\log m)m}$.
	\end{proof}

\subsection{Proof of \Cref{lem_delta_S}}

We are going to prove \Cref{lem_delta_S} inductively on the size of $S$, i.e., for all $S \subseteq [c]$: 
\begin{align}
\E_{\Pi}\left[\sum_{\bv_S\in V_{(2S-1)\cup 2S}}p_S(\bv_S\mid \Pi)^2\right]\leq \paren{\Gamma(m,C)+6}^{|S|}, \label{eq:sep2}
\end{align}
where $\Gamma(m,C) = 4\sqrt{2(C+\log m)m}$ as defined in~\Cref{lem_delta_S}. 

The proof of this lemma builds heavily on our auxiliary information theory lemmas presented in~\Cref{sec:aux-lem} and the reader may want to consult that section during this proof. 

	The base case is when $S=\emptyset$ and by definition, $p_{\emptyset}=1$; hence, both sides are $1$, and the inequality of \eqref{eq:sep2} holds vacuously.

	Let us now focus on the induction step. Fix any element $i\in S$, and let $S' :=S \setminus \{i\}$.
	We have
	\begin{align}
		\E_{\Pi}\left[\sum_{\bv_S\in V_{(2S-1)\cup 2S}}p_S(\bv_S\mid \Pi)^2\right]&=\sum_{\bv_{S'}\in V_{(2S'-1)\cup 2S'}}\E_{\Pi}\left[\sum_{\stackrel{v_{2i-1}\in V_{2i-1}}{v_{2i}\in V_{2i}}}p_S(\bv_{S'},v_{2i-1},v_{2i}\mid \Pi)^2\right] \nonumber\\
		&=\sum_{\bv_{S'}\in V_{(2S'-1)\cup 2S'}}\E_{\Pi}\left[p_{S'}(\bv_{S'}\mid \Pi)^2\cdot\sum_{\stackrel{v_{2i-1}\in V_{2i-1}}{v_{2i}\in V_{2i}}}p_{\{i\}}(v_{2i-1},v_{2i}\mid \bv_{S'}, \Pi)^2\right].\label{eqn_l2_term_1}
	\end{align}

\newcommand{\BB}{\ensuremath{\mathcal{B}}}	
	
	In the following, for any $\bv_{S'}$, we define: 
	\begin{align*}
		\BB(\bv_{S'}) &:= \set{\Pi \in \supp{\rv \Pi} \mid \Pr[\rv \Pi=\Pi\mid \bv_{S'}]\geq m^{-2}2^{-C}}; \\
		\BB &:= \set{\Pi \in \supp{\rv \Pi} \mid \Pr[\rv \Pi=\Pi]\geq m^{-2}2^{-C}}.
	\end{align*}
	
	In words, $\BB(\bv_{S'})$ contains the messages that are ``likely'' conditioned on $\bv_{S'}$ and $\BB$ contains the messages that are ``likely'' a-priori (without any conditioning). 
	
	We bound each summand in~\eqref{eqn_l2_term_1} for likely messages in the following claim. This is the place where we crucially use the fact that message of Alice is ``small'' (note that in the following, $\Psi$ would be small as we are conditioning on a 
	likely message $\Prot$). 
	\begin{claim}\label{clm:perm}
	Fix $\bv_{S'}$. For any $\Prot \in \BB(\bv_{S'})$: 
	\[
		\sum_{\stackrel{v_{2i-1}\in V_{2i-1}}{v_{2i}\in V_{2i}}}p_{\{i\}}(v_{2i-1},v_{2i}\mid \bv_{S'}, \Pi)^2 \leq \Psi(m,\Pi,\rv G_{2i}),	
	\]
	where we define $\Psi(m,\Pi,\rv G_{2i}) := 4\sqrt{(\log m!-\HH(\rv G_{2i}\mid \bv_{S'},\rv \Pi=\Pi))m}+4$. 
	\end{claim}
	\begin{proof}
	By definition, $p_{\{i\}}(v_{2i-1},v_{2i}\mid \bv_{S'}, \Pi)$ measures the probability of existence of any single edge in $\rv G_{2i}$, thus,
	\begin{align}
		\sum_{\stackrel{v_{2i-1}\in V_{2i-1}}{v_{2i}\in V_{2i}}}p_{\{i\}}(v_{2i-1},v_{2i}\mid \bv_{S'}, \Pi)^2 &=\sum_{v_{2i-1}\in V_{2i-1}}\|{\rv G_{2i}(v_{2i-1})\mid \bv_{S'},\rv\Pi=\Pi}\|_2^2 \notag \\
		&\leq \sum_{v_{2i-1}\in V_{2i-1}}\left(\frac{1}{m}+\log m-\HH(\rv G_{2i}(v_{2i-1})\mid \bv_{S'},\rv \Pi=\Pi)\right) \tag{by the connection between entropy and $\ell_2$-norm in~\cref{lem_l2_entropy-1}}\\
		&=1+m\log m-\sum_{v_{2i-1}\in V_{2i-1}}\HH(\rv G_{2i}(v_{2i-1})\mid \bv_{S'},\rv \Pi=\Pi). \label{eq:sep3}
	\end{align}
	Now recall that $\rv G_{2i} \mid \bv_{S'}$ is still a uniformly random matching and hence can be interpreted as a random permutation over $[m]$ (up to a renaming). Moreover, even conditioned on $\rv\Pi=\Pi$, it has a 
	high entropy since: 
	\begin{align}
		\HH(\rv G_{2i}\mid \rv\Pi=\Pi,\bv_{S'})&=\log m!- \kl{(\rv G_{2i}\mid \rv\Pi=\Pi,\bv_{S'})}{(\rv G_{2i}\mid \bv_{S'})}  \tag{by~\Cref{fact:kl-en}} \\
		&\geq \log m!-\log (1/\Pr[\rv \Pi=\Pi\mid \bv_{S'}]) \tag{by~\Cref{fact:kl-event}} \\
		&\geq \log m!-(C+2\log m) \tag{by the choice of $\Pi \in \BB(\bv_{S'})$}\\
		&>\log m!-m/8. \label{eq:added1}
	\end{align}
	As such, $\rv G_{2i} \mid \bv_{S'}, \rv\Pi=\Pi$ is a high entropy random permutation of $[m]$ and we are interested in bounding the entropy of its average coordinate in~\eqref{eq:sep3}. 
	This is precisely the setting of~\Cref{lem_info_perm}, which allows us to obtain, 
	\begin{align*}
		\eqref{eq:sep3} &\leq 4\sqrt{(\log m!-\HH(\rv G_{2i}\mid \bv_{S'},\rv \Pi=\Pi))m}+4 = \Psi(m,\Pi,\rv G_{2i}). \tag{by~\eqref{eq:added1} and~\Cref{lem_info_perm}}
	\end{align*}
	This concludes the proof of~\Cref{clm:perm}.
	\end{proof}
	
	We continue with the proof of~\Cref{lem_delta_S}. To bound \eqref{eqn_l2_term_1}, let us fix $\bv_{S'}$, and consider the expectation over $\Pi$:
	\begin{align}
		&\, \sum_\Pi \Pr(\rv\Pi=\Pi) \cdot p_{S'}(\bv_{S'}\mid \Pi)^2\cdot \sum_{\stackrel{v_{2i-1}\in V_{2i-1}}{v_{2i}\in V_{2i}}}p_{\{i\}}(v_{2i-1},v_{2i}\mid \bv_{S'}, \Pi)^2  \notag \\
		=&\, \sum_{\Pi \in \BB(\bv_{S'})} \Pr(\rv\Pi=\Pi) \cdot p_{S'}(\bv_{S'}\mid \Pi)^2\cdot \sum_{\stackrel{v_{2i-1}\in V_{2i-1}}{v_{2i}\in V_{2i}}}p_{\{i\}}(v_{2i-1},v_{2i}\mid \bv_{S'}, \Pi)^2 \notag \\
		&\quad +\sum_{\Pi \notin \BB(\bv_{S'})} \Pr(\rv\Pi=\Pi) \cdot p_{S'}(\bv_{S'}\mid \Pi)^2\cdot \sum_{\stackrel{v_{2i-1}\in V_{2i-1}}{v_{2i}\in V_{2i}}}p_{\{i\}}(v_{2i-1},v_{2i}\mid \bv_{S'}, \Pi)^2,
		\intertext{which by~\Cref{clm:perm} for the first term and the trivial upper bound for the second, is at most}
		\leq&\, \sum_{\Pi \in \BB(\bv_{S'})} \Pr(\rv\Pi=\Pi) \cdot p_{S'}(\bv_{S'}\mid \Pi)^2\cdot \Psi(m,\Pi,\rv G_{2i}) \notag \\
		&\quad +\sum_{\Pi \notin \BB(\bv_{S'})} m^{-2} \cdot 2^{-C} \cdot p_{S'}(\bv_{S'})\cdot m^2 \tag{where we define $p_{S'}(\bv_{S'}) := p_{S'}(\bv_{S'} \mid \emptyset)$} \notag \\
		\leq &\, \paren{\sum_{\Pi} \Pr(\rv\Pi=\Pi) \cdot p_{S'}(\bv_{S'}\mid \Pi)^2\cdot \Psi(m,\Pi,\rv G_{2i})}+p_{S'}(\bv_{S'}). \label{eq:sep4}
	\end{align}
	Next, we split the first term in~\eqref{eq:sep4} according to $\Pr(\rv\Pi=\Pi)$, and it is equal to
	\begin{align}
		= & \sum_{\Pi \in \BB} \Pr(\rv\Pi=\Pi) \cdot p_{S'}(\bv_{S'}\mid \Pi)^2\cdot \Psi(m,\Pi,\rv G_{2i}) 
		 +\sum_{\Pi \notin \BB} \Pr(\rv\Pi=\Pi) \cdot p_{S'}(\bv_{S'}\mid \Pi)^2\cdot \Psi(m,\Pi,\rv G_{2i})\nonumber\\
		\leq &  \sum_{\Pi \in \BB} \Pr(\rv\Pi=\Pi) \cdot p_{S'}(\bv_{S'}\mid \Pi)^2\cdot \Psi(m,\Pi,\rv G_{2i}) 
		+ \sum_{\Pi \notin \BB} m^{-2}2^{-C} \cdot p_{S'}(\bv_{S'}\mid \Pi)\cdot m\log m  \nonumber\\
		\leq &\sum_{\Pi \in \BB} \Pr(\rv\Pi=\Pi) \cdot p_{S'}(\bv_{S'}\mid \Pi)^2\cdot \Psi(m,\Pi,\rv G_{2i}) + \sum_{\Pi \notin \BB} 2^{-C} \cdot p_{S'}(\bv_{S'}\mid \Pi)  \label{eq:sep5}.
	\end{align}
	We further simplify the first term in~\eqref{eq:sep5} in the following claim. 
	\begin{claim}\label{clm:first-term-simpler}
		For the first term in~\eqref{eq:sep5}, 
		\begin{align*}
		&\sum_{\Pi \in \BB} \Pr(\rv\Pi=\Pi) \cdot p_{S'}(\bv_{S'}\mid \Pi)^2\cdot \Psi(m,\Pi,\rv G_{2i}) 
		 \leq \left(\Gamma(m,C)+4\right) \cdot \sum_{\Pi \in \BB} \Pr(\rv\Pi=\Pi) \cdot p_{S'}(\bv_{S'}\mid \Pi)^2,
		\end{align*} 
		for the parameter $\Gamma(m,C) := 4\sqrt{2(C+\log m)m}$ defined in~\Cref{lem_delta_S}. 
	\end{claim}
	\begin{proof}
	By the definition of $\Psi$ in~\Cref{clm:perm}, the LHS of the claim is equal to
	\begin{align}
		&\, \sum_{\Pi \in \BB} \Pr(\rv\Pi=\Pi) \cdot p_{S'}(\bv_{S'}\mid \Pi)^2\cdot \paren{4\sqrt{(\log m!-\HH(\rv G_{2i}\mid \bv_{S'},\rv \Pi=\Pi))m}+4}  \notag \\
		\intertext{which by Jensen's inequality and the concavity of square root, is at most}
		\leq &\,  \left(\sum_{\Pi \in \BB} \Pr(\rv\Pi=\Pi) \cdot p_{S'}(\bv_{S'}\mid \Pi)^2\right) \left(4\cdot\sqrt{\left(\log m!-\sum_{\Pi \in \BB} q(\Pi)\cdot \HH(\rv G_{2i}\mid \bv_{S'},\rv \Pi=\Pi)\right)m}+4\right),
	\label{eqn_l2_term_3}
	\end{align}
	where we define $q(\Pi)\propto \Pr(\rv\Pi=\Pi) p_{S'}(\bv_{S'}\mid \Pi)^2\propto \Pr(\rv\Pi=\Pi)^{-1}\Pr(\rv \Pi=\Pi\mid \bv_{S'})^2$ and $\sum_{\Pi \in \BB} q(\Pi)=1$.
	
	Notice the term $\sum_{\Pi \in \BB }q(\Pi)\cdot \HH(\rv G_{2i}\mid \bv_{S'},\rv \Pi=\Pi)$ in above which is a weighted sum of entropy terms. 
	In~\Cref{lem_weighted_sum_entropy}, we give a simple auxiliary lemma that allows for lower bounding such weighted sum of entropy terms. 
	
	We apply \Cref{lem_weighted_sum_entropy} for $\rA$ being $(\rv \Pi\mid \bv_{S'})$ and $\rv M$ being $(\rv G_{2i}\mid \bv_{S'})$, and further set
	\[
	t=m!,\quad \ell:=2^C, \quad \alpha(\Pi) :=\Pr(\rv\Pi=\Pi \mid \bv_{S'}),  \quad \beta(\Pi):=\Pr(\rv\Pi=\Pi),\quad \epsilon :=m^{-2}2^{-C}.
	\]
	Consequently, 
	\[
		q(\Pi)=\frac{\beta(\Pi)^{-1}\alpha(\Pi)^2}{\sum_{\Pi' \in \BB}\beta(\Pi')^{-1}\alpha(\Pi')^2}.
	\]
	Then since $\rv G_{2i}$ is a uniformly random matching conditioned on $\bv_{S'}$, we can apply~\Cref{lem_weighted_sum_entropy} and have, 
	\[
		\sum_{\Pi \in \BB}q(\Pi)\cdot \HH(\rv G_{2i}\mid \bv_{S'},\rv \Pi=\Pi)\geq \log m!-2(C+\log m).
	\]
	Therefore,
	\begin{align*}
		\eqref{eqn_l2_term_3}&\leq \left(\sum_{\Pi \in \BB} \Pr(\rv\Pi=\Pi) \cdot p_{S'}(\bv_{S'}\mid \Pi)^2\right) \left(4\sqrt{2m(C+\log m)}+4\right) \\
		&=\left(\sum_{\Pi \in \BB} \Pr(\rv\Pi=\Pi) \cdot p_{S'}(\bv_{S'}\mid \Pi)^2\right) \left(\Gamma(m,C)+4\right),
	\end{align*}
	by the choice of $\Gamma(m,C)$ as desired. 
	\end{proof}
	
	Finally, by~\Cref{clm:first-term-simpler} and~\eqref{eq:sep5},~\eqref{eq:sep4}, we obtain that for any fixed $\bv_{S'}$, 
	\begin{align*}
		&\sum_\Pi \Pr(\rv\Pi=\Pi) \cdot p_{S'}(\bv_{S'}\mid \Pi)^2\cdot \sum_{\stackrel{v_{2i-1}\in V_{2i-1}}{v_{2i}\in V_{2i}}}p_{\{i\}}(v_{2i-1},v_{2i}\mid \bv_{S'}, \Pi)^2 \\
		&\hspace{1cm} \leq p_{S'}(\bv_{S'}) + {\sum_{\Pi} \Pr(\rv\Pi=\Pi) \cdot p_{S'}(\bv_{S'}\mid \Pi)^2\cdot \Psi(m,\Pi,\rv G_{2i})}\tag{by~\eqref{eq:sep4}} \\
		&\hspace{1cm} \leq  p_{S'}(\bv_{S'}) + \sum_{\Pi \notin \BB} 2^{-C} \cdot p_{S'}(\bv_{S'}\mid \Pi) + \sum_{\Pi \in \BB} \Pr(\rv\Pi=\Pi) \cdot p_{S'}(\bv_{S'}\mid \Pi)^2\cdot \Psi(m,\Pi,\rv G_{2i})  \tag{by~\eqref{eq:sep5}} \\
		&\hspace{1cm} \leq  p_{S'}(\bv_{S'}) + \sum_{\Pi \notin \BB} 2^{-C} \cdot p_{S'}(\bv_{S'}\mid \Pi) +  \sum_{\Pi \in \BB} \Pr(\rv\Pi=\Pi) \cdot p_{S'}(\bv_{S'}\mid \Pi)^2 \cdot \left(\Gamma(m,C)+4\right) \tag{by~\Cref{clm:first-term-simpler}} \\
		&\hspace{1cm} \leq  p_{S'}(\bv_{S'}) + \sum_{\Pi} 2^{-C} \cdot p_{S'}(\bv_{S'}\mid \Pi) +  \sum_{\Pi} \Pr(\rv\Pi=\Pi) \cdot p_{S'}(\bv_{S'}\mid \Pi)^2 \cdot \left(\Gamma(m,C)+4\right). 
	\end{align*}
	
	By plugging these bounds in~\eqref{eqn_l2_term_1} and summing over all choices $\bv_{S'}$, we have
	\begin{align*}
		\eqref{eqn_l2_term_1}&\leq \sum_{\bv_{S'}\in V_{(2S'-1)\cup 2S'}}\left( \left(\Gamma(m,C)+4\right)  \cdot \left(\sum_{\Pi} p(\Pi) \cdot p_{S'}(\bv_{S'}\mid \Pi)^2\right)+ \paren{\sum_{\Pi} 2^{-C} \cdot p_{S'}(\bv_{S'}\mid \Pi)} + p(\bv_{S'})\right) \\
		&= \paren{(\Gamma(m,C)+4)\cdot \sum_{\bv_{S'}\in V_{(2S'-1)\cup 2S'}}\E_{\Pi}  \left[p_{S'}(\bv_{S'}\mid \Pi)^2\right]}+2.
	\end{align*}
	That is,
	\begin{align*}
		\E_{\Pi}\left[\sum_{\bv_S\in V_{(2S-1)\cup2S}}p_S(\bv_S\mid \Pi)^2\right] &\leq (\Gamma(m,C)+4)\cdot \E_\Pi\left[\sum_{\bv_{S'}\in V_{(2S'-1)\cup 2S'}}p_{S'}(\bv_{S'}\mid \Pi)^2\right]+2 \\
		&\leq (\Gamma(m,C)+4)\cdot (\Gamma(m,C)+6)^{\card{S'}}+2 \tag{by the induction hypothesis for $S'$ in~\eqref{eq:sep2}} \\
		&\leq (\Gamma(m,C)+6)^{|S|},
	\end{align*}
	concluding the induction step, and the proof of~\Cref{lem_delta_S}. 

\section{Lower Bound For Multiple Rounds}
\label{sec:multiround}

In this section, we inductively extend the one round lower bound from \Cref{sec:single-block} to many rounds. Specifically, we show that:

\begin{restate}[\Cref{lemma:multiround}]
\lemmamultiround
\end{restate}
\begin{proof}
To start, note that we can assume without loss of generality the protocol $\Pi$ is deterministic. Indeed, as a randomized protocol is simply a distribution over deterministic protocols, if we show the lemma for all deterministic protocols, we can use \Cref{lemma:tvdtriangle} to also get the lemma for all randomized protocols.

Now, fix $s \leq s_1$, injections $\Y, \N : [s] \to [s_1]$, and a deterministic $(\card{S} - 1)$-round $(s, c, S)$-protocol $\Pi$ such that $\norm{\Pi} \leq C$. We assume without loss of generality that $\card{S}$ is even implying that Alice sends the last message in $\Pi$. A symmetric argument, with Alice replaced by Bob, holds when $\card{S}$ is odd.

\paragraph{Notation.} Throughout this proof, we shall use $S' = S_{\geq 2}$ and $\hat{S} = S_{\leq 2}$. Define the sequence $T = \nest(s, c, S)$ and $\hat{T} = \nest(s, c, \hat{S})$. We use normal letters for objects associated with $S$, letters with primes for objects associated with $S'$, and letters with a $\hat{}$ on top to denote objects associated with $\hat{S}$. Thus, $\distnest'^{\N'}_s$ will denote $\distnest_{s, c, S_{\geq 2}}^{\N'}$, $\hat{\n}$ will denote $\n_{c, \hat{S}}$, $\hat{\layer}$ to denote $\layer_{\hat{T}}$, {\em etc.}

For a $\hat{T}$-layered graph $\hat{H}$, define the event $\expand_{\hat{H}}$ over the distribution $\layer$ as follows:
\[
\expand_{\hat{H}} := \forall i \in [2c]: \text{ The subgraph $\{ E_{(i-1)\n' + i'} \}_{i' = 1}^{\n'}$ is nice and } G_{(i-1)\n' \to i\n'} = \hat{H}_i .
\]
For a $\hat{T}$-layered graph $\hat{H}$, we define $\hat{H}_{\alice} = \{\hat{H}_i\}_{\ev~i \in [\hat{\n}]}$ and $\hat{H}_{\bob} = \{\hat{H}_i\}_{\od~i \in [\hat{\n}]}$. Analogously (note that $\hat{\n} = 2c$), for a $T$-layered graph $H$, we define $H_{\alice} = \{ H_{(i-1)\n' \to i\n'} \}_{\ev~i \in [2c]}$ and $H_{\bob} = \{ H_{(i-1)\n' \to i\n'} \}_{\od~i \in [2c]}$.

Let $\Pi^{\trans}$ denote the protocol where Alice and Bob first run the first $(\card{S} - 2)$-rounds of $\Pi$ and then Alice outputs the transcript of $\Pi$, {\em i.e.}, the messages exchanged in the first $(\card{S} - 2)$-rounds and the message she would have sent in the last round (recall that Alice sends the last message in $\Pi$). Fix $\hat{H}^{\star}$ to be some canonical $\hat{T}$-layered graph. We define the following notation for the edges of a $T$-layered graph $H = \left( \bigsqcup_{i=0}^{\n} V_i, \bigsqcup_{i=1}^{\n} E_i \right)$:
\begin{itemize}
\item For all $i \in [2c]$, define $L_i = (L_{i, \ev}, L_{i, \od})$ where:
\[
L_{i, \ev} = \{ E_{(i-1)\n' + i'} \}_{\ev~i' = 1}^{\n'} \hspace{1cm}\text{and}\hspace{1cm} L_{i, \od} = \{ E_{(i-1)\n' + i'} \}_{\od~i' = 1}^{\n'} .
\]
\item Define $L_{\pub} = (L_{\pub, \ev}, L_{\pub, \od})$ where:
\[
L_{\pub, \ev} = \{ L_{i, \ev} \}_{\od~i = 1}^{2c} \hspace{1cm}\text{and}\hspace{1cm} L_{\pub, \od} = \{ L_{i, \od} \}_{\ev~i = 1}^{2c}  .
\]
\item Define:
\[
L_{\apriv} = \{ L_{i, \ev} \}_{\ev~i = 1}^{2c} \hspace{1cm}\text{and}\hspace{1cm} L_{\bpriv} = \{ L_{i, \od} \}_{\od~i = 1}^{2c}  .
\]
\end{itemize}
Thus, $H$ can be equivalently seen as the tuple $(L_{\pub}, L_{\apriv}, L_{\bpriv})$ and we use these representations interchangeably. We first claim that:

\begin{lemma}
\label{lemma:hybrid}
For all $\hat{T}$-layered graphs $\hat{H}^1, \hat{H}^2$ such that $\hat{H}^1_{\alice} = \hat{H}^2_{\alice}$, we have that:
\[
\tvd{ \distributions{ \rv{\Pi}^{\trans}(H), \rv{L}_{\pub} }{ H \sim \layer \mid \expand_{\hat{H}^1} } }{ \distributions{ \rv{\Pi}^{\trans}(H), \rv{L}_{\pub} }{ H \sim \layer \mid \expand_{\hat{H}^2} } } \leq c\epsilon_0 .
\]
\end{lemma}
\begin{proof}
To prove this lemma, we show that if the graphs $\hat{H}^1, \hat{H}^2$ differ in at most one layer, then:
\begin{equation}
\label{eq:hybrid1}
\tvd{ \distributions{ \rv{\Pi}^{\trans}(H), \rv{L}_{\pub} }{ H \sim \layer \mid \expand_{\hat{H}^1} } }{ \distributions{ \rv{\Pi}^{\trans}(H), \rv{L}_{\pub} }{ H \sim \layer \mid \expand_{\hat{H}^2} } } \leq \epsilon_0 .
\end{equation}
\Cref{eq:hybrid1} implies the lemma using the triangle inequality as $\hat{H}^1, \hat{H}^2$ differ in at most $c$ layers. As the argument is symmetric for all layers, we assume without loss of generality that $\hat{H}^1$ and $\hat{H}^2$ differ only in the first layer. 

Define $\hat{H}_i = \hat{H}^1_i =  \hat{H}^2_i$ for all $i \neq 1 \in [2c]$ and consider the following $(\card{S} - 2)$-round $(s, c, S')$-protocol $\Pi'$:

\begin{algorithm}
\caption{The protocol $\Pi'$.}
\label{algo:piprime}
\begin{algorithmic}[1]
\renewcommand{\algorithmicrequire}{\textbf{Input:}}
\renewcommand{\algorithmicensure}{\textbf{Output:}}

\Require An $(s, c, S')$-almost nested block graph $G' = \left( \bigsqcup_{i'=0}^{\n'} V'_{i'}, \bigsqcup_{i'=1}^{\n'} E'_{i'} \right)$. Alice gets the edges $E'_{i'}$ for even $i'$ and Bob gets the edges $E'_{i'}$ for odd $i$.

\medskip

\Statex \hspace{-\algorithmicindent} {\bf Sample a $T$-layered graph $H$:}

\medskip

\State Alice sets $L_{1,\ev} = \{E'_{i'}\}_{\ev~i' = 1}^{\n'}$ and Bob sets $L_{1,\od} = \{E'_{i'}\}_{\od~i' = 1}^{\n'}$. \label{line:piprime:first}

\State For all $i \in \{2, \cdots, 2c-1\}$, publicly sample edges $L_i$ from the distribution $\distnest'^{\hat{H}_i}$. \label{line:piprime:i}

\State Publicly sample edges $L_{2c}$ from the distribution $\distnest'^{\hat{H}_i}_{s_1}$. \label{line:piprime:last}

\medskip

\Statex \hspace{-\algorithmicindent} {\bf Simulate the protocol $\Pi^{\trans}$ on $H$:}

\medskip

\State Alice and Bob run the protocol $\Pi^{\trans}$ on the graph $H$ and Alice outputs $\Pi^{\trans}(H)$ and $L_{\pub}$.

\medskip

\end{algorithmic}
\end{algorithm}

Next, note that when the input to $\Pi'$ is from $\distnest'^{\hat{H}^1_1}_{s}$, then the graph $H$ sampled by $\Pi'$ distributes as $\layer \mid \expand_{\hat{H}^1}$. To see why, note that the edges $L_i$ are mutually independent for all $i \in [2c]$ in the distribution $\layer$ and conditioning on $\expand_{\hat{H}^1}$ does not change this fact. Thus, the distribution $\layer \mid \expand_{\hat{H}^1}$ is exactly the distribution sampled in Lines~\ref{line:piprime:first}, \ref{line:piprime:i}, and \ref{line:piprime:last}, as desired. 

Similarly, when the input to $\Pi'$ is from $\distnest'^{\hat{H}^2_1}_{s}$, then the graph $H$ sampled by $\Pi'$ distributes as $\layer \mid \expand_{\hat{H}^2}$. We get:
\begin{align*}
&\tvd{ \distributions{ \rv{\Pi}^{\trans}(H), \rv{L}_{\pub} }{ H \sim \layer \mid \expand_{\hat{H}^1} } }{ \distributions{ \rv{\Pi}^{\trans}(H), \rv{L}_{\pub} }{ H \sim \layer \mid \expand_{\hat{H}^2} } } \\
&\hspace{1cm}\leq \tvd{ \distributions{ \rv{\Pi'}(G') }{ G' \sim \distnest'^{\hat{H}^1_1}_{s} } }{ \distributions{ \rv{\Pi'}(G') }{ G' \sim \distnest'^{\hat{H}^2_1}_{s} } } \\
&\hspace{1cm}\leq \epsilon_0 ,
\end{align*}
by the assumption in \Cref{lemma:multiround} and \Cref{eq:hybrid1} follows. 
\end{proof}

Next, we show the following property of $\Pi$.

\begin{lemma}
\label{lemma:hatpi:bob}
For the distribution $\distnest_{s}$, it holds that:
\[
\mi{ \rv{L}_{\bpriv} }{ \rv{G}_{0 \to \n} \mid \rv{\Pi}^{\trans}(G), \rv{G}_{\bob}, \rv{L}_{\pub} } = 0.
\]
\end{lemma}
\begin{proof}
To start, note that conditioned on $\rv{G}_{\bob}, \rv{L}_{\pub}$, the random variable $\rv{G}_{0 \to \n}$ is determined by $\rv{L}_{\apriv}$. Thus, we have from \Cref{part:data-processing} of \Cref{fact:it-facts} that:
\[
\mi{ \rv{L}_{\bpriv} }{ \rv{G}_{0 \to \n} \mid \rv{\Pi}^{\trans}(G), \rv{G}_{\bob}, \rv{L}_{\pub} } \leq \mi{ \rv{L}_{\bpriv} }{ \rv{L}_{\apriv} \mid \rv{\Pi}^{\trans}(G), \rv{G}_{\bob}, \rv{L}_{\pub} } .
\]
Next, recall that $\Pi$ is an $(\card{S} - 1)$-round communication protocol and let $\Pi_1(G), \cdots, \Pi_{\card{S} - 1}(G)$ be the messages sent in $\Pi$ when run on input $G$. As $\Pi$ is deterministic, we have that for all odd (resp, even) $r \in [\card{S} - 1]$, the value $\Pi_1(G)$ is determined by Alice's (resp. Bob's) input and the values $\Pi_1(G), \cdots, \Pi_{r-1}(G)$. Using this fact repeatedly, we have:
\begin{align*}
\mi{ \rv{L}_{\bpriv} }{ \rv{G}_{0 \to \n} \mid \rv{\Pi}^{\trans}(G), \rv{G}_{\bob}, \rv{L}_{\pub} } &\leq \mi{ \rv{L}_{\bpriv} }{ \rv{L}_{\apriv} \mid \Pi_1(\rv{G}), \cdots, \Pi_{\card{S} - 1}(G), \rv{G}_{\bob}, \rv{L}_{\pub} } \\
&\leq \mi{ \rv{L}_{\bpriv} }{ \rv{L}_{\apriv} \mid \rv{G}_{\bob}, \rv{L}_{\pub} } \tag{\Cref{prelim-prop:info-decrease}} \\
&\leq \mi{ \rv{L}_{\bpriv},  \rv{L}_{\pub,\ev} }{ \rv{L}_{\apriv}, \rv{L}_{\pub,\od} \mid \rv{G}_{\bob} } \tag{\Cref{fact:it-facts}, \Cref{part:chain-rule}} \\
&\leq  \mi{ \{\rv{L}_i\}_{\od~i = 1}^{2c} }{ \{\rv{L}_i\}_{\ev~i = 1}^{2c} \mid \{ \rv{G}_{(i-1)\n' \to i\n'} \}_{\od~i \in [2c]} } .
\end{align*}
To finish, note that, conditioned on $\{ \rv{G}_{(i-1)\n' \to i\n'} \}_{\od~i \in [2c]}$ the edges in $\{\rv{L}_i\}_{\od~i = 1}^{2c}$ and $\{\rv{L}_i\}_{\ev~i = 1}^{2c}$ are independent, and therefore the right hand side is $0$.
\end{proof}

Consider the following $1$-round $(s, c, \hat{S})$-protocol $\hat{\Pi}$. Observe that in Lines~\ref{line:hatpi:apriv1}, \ref{line:hatpi:apriv2} we resample $L_{i,\od}$. This is fine as our conditioning ensures that the two values are the same. Likewise for Lines~\ref{line:hatpi:bpriv1}, \ref{line:hatpi:bpriv2}.

\begin{algorithm}
\caption{The protocol $\hat{\Pi}$.}
\label{algo:hatpi}
\begin{algorithmic}[1]
\renewcommand{\algorithmicrequire}{\textbf{Input:}}
\renewcommand{\algorithmicensure}{\textbf{Output:}}

\Require An $(s, c, \hat{S})$-almost nested block graph $\hat{G} = \left( \bigsqcup_{i=0}^{\hat{\n}} \hat{V}_i, \bigsqcup_{i=1}^{\hat{\n}} \hat{E}_i \right)$. Alice gets the edges $\hat{E}_i$ for even $i$ and Bob gets the edges $\hat{E}_i$ for odd $i$.

\medskip 

\State Construct a $\hat{T}$-layered graph $\hat{H}$ by setting $\hat{H}_{\alice} = \hat{G}_{\alice}$ and $\hat{H}_{\bob} = \hat{H}^{\star}_{\bob}$. 

\medskip

\Statex \hspace{-\algorithmicindent} {\bf Sample a $T$-layered graph $H$:}

\medskip

\State Publicly sample $L_{\pub}$ uniformly at random. \label{line:hatpi:pub}

\For{$i \in [2c]$} \label{line:hatpi:loop}

\If{$i$ is even}

\State If $i = 2c$, Alice samples edges $L_i$ from the distribution $\distnest'^{\hat{H}_i}_{s_1} \mid \rv{L}_{i,\od} = L_{i,\od}$ .  \label{line:hatpi:apriv1}
\State If $i < 2c$,  Alice samples edges $L_i$ from the distribution $\distnest'^{\hat{H}_i} \mid \rv{L}_{i,\od} = L_{i,\od}$ .  \label{line:hatpi:apriv2}

\Else

\State If $i = 1$, Alice samples edges $L_i$ from the distribution $\distnest'^{\hat{H}_i}_{s} \mid \rv{L}_{i,\ev} = L_{i,\ev}$ .  \label{line:hatpi:bpriv1}
\State If $i > 1$,  Alice samples edges $L_i$ from the distribution $\distnest'^{\hat{H}_i} \mid \rv{L}_{i,\ev} = L_{i,\ev}$ .  \label{line:hatpi:bpriv2}

\EndIf

\EndFor

\medskip

\Statex \hspace{-\algorithmicindent} {\bf Simulate the protocol $\Pi^{\trans}$ on $H$:}

\medskip

\State Alice runs the protocol $\Pi^{\trans}$ on $H$. Observe that she can do this without any communication as she knows all of $H$. She sends the output $\Pi^{\trans}(H)$ to Bob. \label{line:hatpi:comm}

\medskip

\State Bob outputs $(\hat{G}_{\bob}, \Pi^{\trans}(H), L_{\pub})$.

\end{algorithmic}
\end{algorithm}

Let $\hat{H}^{\hat{\Pi}}(\hat{G})$ be the graph constructed in $\hat{\Pi}$ when the input is $\hat{G}$. Also, define the random variable $\rv{H}^{\hat{\Pi}}(\hat{G})$ be the graph $H$ sampled in $\hat{\Pi}$ when the input is $\hat{G}$. Using \Cref{lemma:hybrid}, we conclude that:
\begin{lemma}
\label{lemma:hybrid1}
For all injections $\Z : [s] \to [s_1]$, we have that:
\[
\tvd{ \distributions{ \rv{\hat{G}}_{\bob}, \rv{\Pi}^{\trans}(G), \rv{L}_{\pub} }{ \hat{G} \sim \hat{\distnest}^{\Z}_{s}, G \sim \layer \mid \expand_{\hat{G}} } }{ \distributions{ \rv{\hat{G}}_{\bob}, \rv{\Pi}^{\trans}(G), \rv{L}_{\pub} }{ \hat{G} \sim \hat{\distnest}^{\Z}_{s}, G \sim \layer \mid \expand_{\hat{H}^{\hat{\Pi}}(\hat{G})} } } \leq c\epsilon_0 .
\]
\end{lemma}
\begin{proof}
We have:
\begin{align*}
&\tvd{ \distributions{ \rv{\hat{G}}_{\bob}, \rv{\Pi}^{\trans}(G), \rv{L}_{\pub} }{ \hat{G} \sim \hat{\distnest}^{\Z}_{s}, G \sim \layer \mid \expand_{\hat{G}} } }{ \distributions{ \rv{\hat{G}}_{\bob}, \rv{\Pi}^{\trans}(G), \rv{L}_{\pub} }{ \hat{G} \sim \hat{\distnest}^{\Z}_{s}, G \sim \layer \mid \expand_{\hat{H}^{\hat{\Pi}}(\hat{G})} } } \\
&\leq \E_{ \hat{G}_{\bob} \sim \hat{\distnest}^{\Z}_{s} } \Bracket{ \tvd{ \distributions{ \rv{\Pi}^{\trans}(G), \rv{L}_{\pub} }{ \hat{G} \sim \hat{\distnest}^{\Z}_{s} \mid \hat{G}_{\bob}, G \sim \layer \mid \expand_{\hat{G}} } }{ \distributions{ \rv{\Pi}^{\trans}(G), \rv{L}_{\pub} }{ \hat{G} \sim \hat{\distnest}^{\Z}_{s} \mid \hat{G}_{\bob}, G \sim \layer \mid \expand_{\hat{H}^{\hat{\Pi}}(\hat{G})} } } } \tag{\Cref{fact_tvd_chain_rule}} \\
&\leq \E_{ \hat{G}_{\bob} \sim \hat{\distnest}^{\Z}_{s} } \Bracket{ \tvd{ \E_{ \hat{G} \sim \hat{\distnest}^{\Z}_{s} \mid \hat{G}_{\bob} } \Bracket{ \distributions{ \rv{\Pi}^{\trans}(G), \rv{L}_{\pub} }{ G \sim \layer \mid \expand_{\hat{G}} } } }{ \E_{ \hat{G} \sim \hat{\distnest}^{\Z}_{s} \mid \hat{G}_{\bob} } \Bracket{ \distributions{ \rv{\Pi}^{\trans}(G), \rv{L}_{\pub} }{ G \sim \layer \mid \expand_{\hat{H}^{\hat{\Pi}}(\hat{G})} } } } } \\
&\leq \E_{ \hat{G}_{\bob} \sim \hat{\distnest}^{\Z}_{s} } \E_{ \hat{G} \sim \hat{\distnest}^{\Z}_{s} \mid \hat{G}_{\bob} } \Bracket{ \tvd{ \distributions{ \rv{\Pi}^{\trans}(G), \rv{L}_{\pub} }{ G \sim \layer \mid \expand_{\hat{G}} } }{ \distributions{ \rv{\Pi}^{\trans}(G), \rv{L}_{\pub} }{ G \sim \layer \mid \expand_{\hat{H}^{\hat{\Pi}}(\hat{G})} } } } \tag{\Cref{lemma:tvdtriangle}} \\
&\leq c\epsilon_0 \tag{\Cref{lemma:hybrid}} .
\end{align*}
\end{proof}

A key property satisfied by the protocol $\hat{\Pi}$ is that:
\begin{lemma}
\label{lemma:hatpi:identical}
It holds for all $(s, c, \hat{S})$-almost nested block graphs $\hat{G}$ that:
\[
\distribution{ \rv{H}^{\hat{\Pi}}(\hat{G}) } = \layer \mid \expand_{\hat{H}^{\hat{\Pi}}(\hat{G})} .
\]
\end{lemma}
\begin{proof}
The proof is in two steps.
\paragraph{The marginal distribution of $\rv{L}_{\pub}$ is identical.} First note, that the marginal distribution $\rv{L}_{\pub}$ in $\distribution{ \rv{H}^{\hat{\Pi}}(\hat{G}) }$ is just uniform. We now consider the distribution $\layer \mid \expand_{\hat{H}^{\hat{\Pi}}(\hat{G})}$. As the edges $L_i$ are mutually independent for all $i \in [2c]$ in the distribution $\layer$ and conditioning on $\expand_{\hat{H}^{\hat{\Pi}}(\hat{G})}$ does not change this fact, the distribution $\layer \mid \expand_{\hat{H}^{\hat{\Pi}}(\hat{G})}$ can be seen as simply sampling $L_i$ independently from the distribution $\distnest'^{\hat{H}^{\hat{\Pi}}_i(\hat{G})}$ for $i \notin \{1, 2c\}$, from the distribution $\distnest'^{\hat{H}^{\hat{\Pi}}_i(\hat{G})}_{s}$, when $i = 1$, and from the distribution $\distnest'^{\hat{H}^{\hat{\Pi}}_i(\hat{G})}_{s_1}$, when $i = 2c$. In particular, the marginal distribution of $\rv{L}_{\pub}$ is uniform.
\paragraph{Conditioned on $L_{\pub}$, the marginal distribution of $(\rv{L}_{\apriv}, \rv{L}_{\bpriv})$ is identical.} The value $(\rv{L}_{\apriv}, \rv{L}_{\bpriv})$ can be seen as the $2c$-tuple:
\[
(\rv{L}_{\apriv}, \rv{L}_{\bpriv}) = (\{ \rv{L}_{i, \ev} \}_{\ev~i = 1}^{2c}, \{ \rv{L}_{i, \od} \}_{\od~i = 1}^{2c}) .
\]
We first show that the $2c$ coordinates are mutually independent in both distributions. In the distribution $\distribution{ \rv{H}^{\hat{\Pi}}(\hat{G}) } \mid L_{\pub}$, this is by Lines~\ref{line:hatpi:apriv1}, \ref{line:hatpi:apriv2}, \ref{line:hatpi:bpriv1}, and \ref{line:hatpi:bpriv2}, while for the distribution $\layer \mid \expand_{\hat{H}^{\hat{\Pi}}(\hat{G})}, L_{\pub}$, this is because the coordinates $L_i$ are independent in $\layer$ and conditioning on $\expand_{\hat{H}^{\hat{\Pi}}(\hat{G})}, L_{\pub}$ does not change this fact.

To finish the proof, we fix a coordinate $i \in [2c]$ in this tuple and show that the marginal distribution of the $i^{\text{th}}$ coordinate is the same for both the distributions. Suppose $i = 1$ without loss of generality as the argument is the same for other values of $i$. The marginal distribution of $\rv{L}_{i, \od}$ in the distribution $\distribution{ \rv{H}^{\hat{\Pi}}(\hat{G}) } \mid L_{\pub}$ is the corresponding marginal when the edges $\rv{L}_i$ are sampled from $\distnest'^{\hat{H}_i}_{s} \mid \rv{L}_{i,\ev} = L_{i,\ev}$. The latter is the same as the marginal distribution of $\rv{L}_{i,\od}$ in the distribution $\layer \mid \expand_{\hat{H}^{\hat{\Pi}}(\hat{G})}, L_{\pub}$.
\end{proof}

As the output of the protocol $\Pi$ is determined by the transcript and Bob's input, we get:
\begin{align*}
&\tvd{ \distributions{ \rv{\Pi}(G) }{ G \sim \distnest_{s}^{\Y} } }{ \distributions{ \rv{\Pi}(G) }{ G \sim \distnest_{s}^{\N} } } \\
&\hspace{0.5cm} \leq \tvd{ \distributions{ \rv{G}_{\bob}, \rv{\Pi}^{\trans}(G), \rv{L}_{\pub}, \rv{L}_{\bpriv} }{ G \sim \distnest_{s}^{\Y} } }{ \distributions{ \rv{G}_{\bob}, \rv{\Pi}^{\trans}(G), \rv{L}_{\pub}, \rv{L}_{\bpriv} }{ G \sim \distnest_{s}^{\N} } } \\
&\hspace{0.5cm} \leq \tvd{ \distributions{ \rv{G}_{\bob}, \rv{\Pi}^{\trans}(G), \rv{L}_{\pub} }{ G \sim \distnest_{s}^{\Y} } }{ \distributions{ \rv{G}_{\bob}, \rv{\Pi}^{\trans}(G), \rv{L}_{\pub} }{ G \sim \distnest_{s}^{\N} } } \\
&\hspace{1cm} + \E_{ (G_{\bob}, \Pi^{\trans}(G), L_{\pub}) \sim  \distnest_{s}^{\Y} } \Bracket{ \tvd{ \distributions{ \rv{L}_{\bpriv} }{ G \sim \distnest_{s}^{\Y} \mid G_{\bob}, \Pi^{\trans}(G), L_{\pub} } }{ \distributions{ \rv{L}_{\bpriv} }{ G \sim \distnest_{s}^{\N} \mid G_{\bob}, \Pi^{\trans}(G), L_{\pub} } } } \tag{\Cref{fact_tvd_chain_rule}} .
\end{align*}
To continue, note by \Cref{lemma:hatpi:bob} that, conditioned on $\rv{\Pi}^{\trans}(G), \rv{G}_{\bob}, \rv{L}_{\pub}$, the random variable $\rv{L}_{\bpriv}$ is independent of $\rv{G}_{0 \to \n}$. Thus, the second term on the right hand side is $0$. For the first term, observe that, for all injections $\Z : [s] \to [s_1]$, sampling a graph $G \sim \distnest_{s}^{\Z}$ is the same as first sampling a graph $\hat{G} \sim \hat{\distnest}_{s}^{\Z}$ and then sampling $G \sim \layer \mid \expand_{\hat{G}}$. We get:
\begin{align*}
&\tvd{ \distributions{ \rv{\Pi}(G) }{ G \sim \distnest_{s}^{\Y} } }{ \distributions{ \rv{\Pi}(G) }{ G \sim \distnest_{s}^{\N} } } \\
&\hspace{0.5cm} \leq \tvd{ \distributions{ \rv{G}_{\bob}, \rv{\Pi}^{\trans}(G), \rv{L}_{\pub} }{ \hat{G} \sim \hat{\distnest}_{s}^{\Y}, G \sim \layer \mid \expand_{\hat{G}} } }{ \distributions{ \rv{G}_{\bob}, \rv{\Pi}^{\trans}(G), \rv{L}_{\pub} }{ \hat{G} \sim \hat{\distnest}_{s}^{\N}, G \sim \layer \mid \expand_{\hat{G}} } } \\
&\hspace{0.5cm} \leq \tvd{ \distributions{ \rv{\hat{G}}_{\bob}, \rv{\Pi}^{\trans}(G), \rv{L}_{\pub} }{ \hat{G} \sim \hat{\distnest}_{s}^{\Y}, G \sim \layer \mid \expand_{\hat{G}} } }{ \distributions{ \rv{\hat{G}}_{\bob}, \rv{\Pi}^{\trans}(G), \rv{L}_{\pub} }{ \hat{G} \sim \hat{\distnest}_{s}^{\N}, G \sim \layer \mid \expand_{\hat{G}} } } \\
&\hspace{0.5cm} \leq \tvd{ \distributions{ \rv{\hat{G}}_{\bob}, \rv{\Pi}^{\trans}(G), \rv{L}_{\pub} }{ \hat{G} \sim \hat{\distnest}_{s}^{\Y}, G \sim \layer \mid \expand_{\hat{G}} } }{ \distributions{ \rv{\hat{G}}_{\bob}, \rv{\Pi}^{\trans}(G), \rv{L}_{\pub} }{ \hat{G} \sim \hat{\distnest}_{s}^{\Y}, G \sim \layer \mid \expand_{\hat{H}^{\hat{\Pi}}(\hat{G})} } } \\
&\hspace{1cm} + \tvd{ \distributions{ \rv{\hat{G}}_{\bob}, \rv{\Pi}^{\trans}(G), \rv{L}_{\pub} }{ \hat{G} \sim \hat{\distnest}_{s}^{\Y}, G \sim \layer \mid \expand_{\hat{H}^{\hat{\Pi}}(\hat{G})} } }{ \distributions{ \rv{\hat{G}}_{\bob}, \rv{\Pi}^{\trans}(G), \rv{L}_{\pub} }{ \hat{G} \sim \hat{\distnest}_{s}^{\N}, G \sim \layer \mid \expand_{\hat{H}^{\hat{\Pi}}(\hat{G})} } } \\
&\hspace{1cm} + \tvd{ \distributions{ \rv{\hat{G}}_{\bob}, \rv{\Pi}^{\trans}(G), \rv{L}_{\pub} }{ \hat{G} \sim \hat{\distnest}_{s}^{\N}, G \sim \layer \mid \expand_{\hat{H}^{\hat{\Pi}}(\hat{G})} } }{ \distributions{ \rv{\hat{G}}_{\bob}, \rv{\Pi}^{\trans}(G), \rv{L}_{\pub} }{ \hat{G} \sim \hat{\distnest}_{s}^{\N}, G \sim \layer \mid \expand_{\hat{G}} } } \tag{Triangle Inequality} \\
&\hspace{0.5cm} \leq \tvd{ \distributions{ \rv{\hat{G}}_{\bob}, \rv{\Pi}^{\trans}(G), \rv{L}_{\pub} }{ \hat{G} \sim \hat{\distnest}_{s}^{\Y}, G \sim \layer \mid \expand_{\hat{H}^{\hat{\Pi}}(\hat{G})} } }{ \distributions{ \rv{\hat{G}}_{\bob}, \rv{\Pi}^{\trans}(G), \rv{L}_{\pub} }{ \hat{G} \sim \hat{\distnest}_{s}^{\N}, G \sim \layer \mid \expand_{\hat{H}^{\hat{\Pi}}(\hat{G})} } } \\
&\hspace{4cm} + 2c\epsilon_0 \tag{\Cref{lemma:hybrid1}} \\
&\hspace{0.5cm} \leq \tvd{ \distributions{ \rv{\hat{G}}_{\bob}, \rv{\Pi}^{\trans}(G), \rv{L}_{\pub} }{ \hat{G} \sim \hat{\distnest}_{s}^{\Y}, G \sim \distribution{ \rv{H}^{\hat{\Pi}}(\hat{G}) } } }{ \distributions{ \rv{\hat{G}}_{\bob}, \rv{\Pi}^{\trans}(G), \rv{L}_{\pub} }{ \hat{G} \sim \hat{\distnest}_{s}^{\N}, G \sim \distribution{ \rv{H}^{\hat{\Pi}}(\hat{G}) } } } \\
&\hspace{4cm} + 2c\epsilon_0 \tag{\Cref{lemma:hatpi:identical}} \\
&\hspace{0.5cm} \leq \tvd{ \distributions{ \rv{\hat{\Pi}}(\rv{\hat{G}}) }{ \hat{G} \sim \hat{\distnest}_{s}^{\Y} } }{ \distributions{ \rv{\hat{\Pi}}(\rv{\hat{G}}) }{ \hat{G} \sim \hat{\distnest}_{s}^{\N} } } + 2c\epsilon_0 \\
&\hspace{0.5cm} \leq 2c\epsilon_0 + \Phi(s, c, s_2, C) \tag{\Cref{lemma:oneround}} .
\end{align*}
\end{proof}


\section{Graph Streaming Lower Bounds and Beyond}\label{sec:implications} 

We now list some of the implications of our lower bound for $\OMC$ in~\Cref{res:main} in proving streaming lower bounds for graph problems and beyond. \Cref{res:implications} is formalized by these theorems.

We shall note that many of these reductions are simple (or sometimes even simpler) variants of prior reductions from $\BHH$ for these problems in the literature and we claim no particular novelty in their proofs (although in some cases such as maximum matching  they even improve prior lower bounds for single-pass algorithms). 
Rather, we believe that these reductions showcase the role of $\OMC$ as a ``multi-round version'' of $\BHH$ for proving streaming lower bounds. 

Throughout this section, for any integer $p \geq 1$  and $\eps \in (0,1)$, we define: 
\begin{align}
	g(\eps,p) := \beta \cdot \eps^{1/2p-1} \label{eq:g}
\end{align} 
for some absolute constant $\beta > 0$, depending only on the problem considered (taking $\beta = 10^6$ suffices for all the following applications). 

\subsubsection*{Odd Cycles in the Reductions} 

Before we get to prove our streaming lower bounds, we prove a simple auxiliary lemma that is used in many of the following reductions. Recall that in the $\OMC_{n,k}$ problem, both $n$ and $k$ are even
integers and this in particular implies that we are always working with \emph{even}-length cycles. This is necessary in the definition of the problem as an odd-cycle cannot be partitioned into two matchings as input to players. But what if we
are interested in the case when the input graph consists of many \emph{odd} cycles or {none} at all? As it will become evident shortly, such lower bounds can be very helpful for proving streaming lower bounds for MAX-CUT, matching size, 
or bipartiteness testing. The following lemma establishes a simple reduction for this purpose\footnote{We could have alternatively used~\Cref{thm:FMT} to directly prove a lower bound for a 
slight variant of $\OMC_{n,k}$ for odd-values of $k$ in which the input to players are no longer necessarily matchings--however, in the spirit of this section, we do this part using a reduction as well.}. 

\begin{lemma}\label{lem:odd-cycle}
	Let $n,k$ be even integers such that $n/k$ is also even and sufficiently large.
	Let $G$ be an input graph to $\OMC_{n,k}$. Suppose we sample $(n/k)$ edges from $G$ uniformly at random and obtain a graph $H$ from $G$ by ``stretching'' each sampled edge
	into two edges and placing a new vertex between them.
	Then, $(i)$ in the \Yes case of $\OMC(G)$, $H$ has no odd cycle, while $(ii)$ in the \No case of $\OMC(G)$, with probability at least $9/10$, $H$ has $\nicefrac{n}{10k}$ vertex-disjoint odd cycles of length $k+1$. 
\end{lemma}
\begin{proof}
	Let $s = (n/k)$. The first case is trivial since $H$ is now a Hamiltonian cycle of length $n+s$ and both $n$ and $s$ are even. 
	
	We now analyze the second case. Let $C_1,\ldots,C_{n/k}$ be the $(n/k)$ cycles in $G$. For any $C_i$, we define $X_i \in \set{0,1}$ which is one iff we stretch \emph{exactly} one edge from $C_i$ in $H$. 
	Notice that any cycle $C_i$ with $X_i = 1$ turns into an odd cycle in $H$. For any $i \neq j$, we have, 
	\begin{align*}
		\expect{X_i} &= \frac{\binom{n-k}{s-1}}{\binom{n-1}{s-1}} \geq \frac{1}{2e}; \\
		\var{X_i} &\leq \expect{X_i^2} = \expect{X_i};\\ 
		\cov{X_i X_j} &= \frac{\binom{n-k}{s-1}}{\binom{n-1}{s-1}} \cdot \frac{\binom{n-2k}{s-2}}{\binom{n-k-1}{s-2}} - \paren{\frac{\binom{n-k}{s-1}}{\binom{n-1}{s-1}}}^2 \leq 0. 
	\end{align*}
	Define $X := \sum_{i=1}^{n/k} X_i$ which lower bounds the number of odd cycles in $H$. By the above calculation, 
	\begin{align*}
		\expect{X} &\geq \frac{n}{2e \cdot k}; \\
		\var{X} &= \sum_{i} \var{X_i} + \sum_{i \neq j} \cov{X_i,X_j} \leq \sum_{i} \expect{X_i} + \sum_{i \neq j} \cov{X_i X_j} \leq \frac{n}{2e\cdot k}. 
	\end{align*}
	As such, by Chebyshev's inequality, 
	\begin{align*}
		\Pr\paren{X \leq \frac{n}{10 \cdot k}} \leq \Pr\paren{\card{X - \expect{X}} \geq \frac{n}{5e \cdot k}} \leq \frac{\var{X}}{(\nicefrac{n}{5e \cdot k})^2} \leq \frac{50 \cdot k}{n} \leq \frac{1}{10} \tag{for sufficiently large $n/k \geq 500$}, 
	\end{align*}
	which concludes the proof. 
\end{proof}

\subsection{MAX-CUT}\label{sec:max-cut}

In the MAX-CUT problem, we are given an undirected graph $G=(V,E)$ and our goal is to estimate the \emph{value} of a maximum cut in $G$, i.e., a partition of vertices into two parts that maximizes the number of crossing edges. 

One can approximate MAX-CUT to within a factor of $2$ in $O(\log{n})$ space (by counting the number of edges and dividing by two) and to $(1+\eps)$ in 
$\Ot(n/\eps^2)$ space (by maintaining a cut sparsifier and find the MAX-CUT of the sparsifier in exponential-time); see, e.g.~\cite{KapralovKS15}. A better \emph{two-pass} 
algorithm for this problem on dense graphs is also presented in~\cite{BhaskaraDV18}. A direct reduction from $\BHH$ can prove that any $(1+\eps)$-approximation single-pass algorithm requires $n^{1-O(\eps)}$ space~\cite{KoganK15,KapralovKS15}, and a series of work building on this approach~\cite{KapralovKS15,KapralovKSV17,KapralovK19} culminated in the optimal lower bound of $\Omega(n)$ space for less-than-$2$ approximation in~\cite{KapralovK19}.  

We prove the following theorem for MAX-CUT, formalizing part $(i)$ of~\Cref{res:implications}. 

\begin{theorem}\label{thm:max-cut}
Let $p \geq 1$ be an integer and $\eps \in (0,1)$ be a parameter. Any $p$-pass streaming algorithm that outputs a $(1+\eps)$-approximation to the value of a maximum cut in undirected graphs with probability at least $2/3$ 
requires $\eps \cdot 2^{-O(p)} n^{1-g(\eps,p)}$ space where $g$ is the function in~\Cref{eq:g}. 
\end{theorem}
\begin{proof}
	Suppose $\alg$ is a $p$-pass $S$-space streaming algorithm for $(1+\eps)$-approximation of MAX-CUT. Consider any instance $G$ of $\OMC_{n,k}$ for even integers $n,k$ where $k \leq \frac{1}{20\,\eps}$ such that $n/k$ is also even and is sufficiently large. 
	
	To solve $\OMC(G)$, Alice and Bob use public coins to  sample a graph $H$ as in~\Cref{lem:odd-cycle} (Alice has the original and stretched edges in $M_A$ and Bob has the original and stretched edges in $M_B$). The players then 
	run $\alg$ on the stream obtained by appending edges of Bob to Alice's edges. This can be done in $2p$ rounds and $p \cdot S$ communication.  
	
	When $\OMC(G)$ is a Yes-case, by~\Cref{lem:odd-cycle}, $H$ has no odd cycle (is bipartite) and thus there is a cut that contain all its $n+\nicefrac{n}{k}$ edges. On the other hand, in the No-case of $\OMC(G)$, again by~\Cref{lem:odd-cycle}, 
	with probability $9/10$, $H$ has $\nicefrac{n}{10 \cdot k}$ vertex-disjoint odd cycles. As any cut of this graph has to leave out at least one edge from an odd cycle, we obtain that the value of maximum cut in this case is at most
	$n+ \nicefrac{9n}{10k}$. By the choice of $k$, we obtain that the value of MAX-CUT in the Yes-case is larger  than the No-case by a factor of $(1+\eps)$. Hence, Alice and Bob can use the output of $\alg$ to distinguish between the two cases. 
	
	To conclude, we obtain a $2p - 1$-round $(2p - 1)\cdot S$-communication that computes the correct answer to $\OMC_{n,k}$ with probability at least $2/3-1/10 > 1/2+1/m^2$. By the lower bound for $\OMC_{n,k}$ in~\Cref{thm:mainformal}, 
	we obtain that
	\[
		S = 2^{-O(p)} \cdot (n/k)^{1-10^4\cdot k^{-1/(2p-1)}}\geq \eps 2^{-O(p)}\cdot n^{1-g(\eps,p)}.
	\]
	as desired. 
\end{proof}

\subsection{Maximum Matching Size}\label{sec:matching}

In the maximum matching size estimation problem, we are given an undirected graph $G=(V,E)$ and our goal is to estimate the \emph{size} of a maximum matching in $G$, i.e., the largest collection of vertex-disjoint edges of $G$. 

Maximum matching problem is one of the most studied problems in the graph streaming model. There is an active line of work on estimating  matching size, in particular in planar and low arboricity 
	graphs~\cite{AssadiKL17,EsfandiariHLMO15,McGregorV18,CormodeJMM17,KapralovKS14}. 
	For single-pass algorithms, a reduction from $\BHH$ can prove an $\Omega(\sqrt{n})$ space for better-than-$(3/2)$ approximation in planar graphs~\cite{EsfandiariHLMO15} which was improved to $(1+\eps)$-approximation in 
	$n^{1-O(\eps)}$ space for low-arboricity graphs in~\cite{BuryS15}. 
	
We prove the following theorem for matching size estimation, formalizing part $(ii)$ of~\Cref{res:implications}. 

\begin{theorem}\label{thm:max-matching}
Let $p \geq 1$ be an integer and $\eps \in (0,1)$ be a parameter. Any $p$-pass streaming algorithm that outputs a $(1+\eps)$-approximation to the size of a maximum matching in \emph{planar} graphs with probability at least $2/3$ 
requires $\eps \cdot 2^{-O(p)}\cdot n^{1-g(\eps,p)}$ space where $g$ is the function in~\Cref{eq:g}. 
\end{theorem}
\begin{proof}
	The proof is almost identical to the proof of~\Cref{thm:max-cut}. Recall the reduction from $\OMC_{n,k}$ in that theorem. When $\OMC(G)$ is a Yes-case, $H$ has a perfect matching of size $(n+\nicefrac{n}{k})/2$ as each
	Hamiltonian cycle in a bipartite
	graph can be partitioned into two perfect matchings. On the other hand, when $\OMC(G)$ is a No-case, any odd-cycle of $G$ on $k+1$ vertices can only support a matching of size $k/2$. As such, the maximum matching size in this case is 
	at most $(n+\nicefrac{9n}{10\,k})/2$. The lower bound now holds verbatim as in~\Cref{thm:max-cut}. 
\end{proof}

\begin{remark}
	We remark that the previous best lower bound of~\cite{BuryS15}  for single-pass streaming algorithms on sparse graphs only hold for graphs of arboricity $\approx 1/\eps$ and \emph{not} planar graphs. As such, our lower bound in~\Cref{thm:max-matching}
	improves upon previous work \emph{even} for single-pass algorithms. 
\end{remark}

\subsection{Maximum Acyclic Subgraph}

In the maximum acyclic subgraph problem, we are given a \emph{directed} graph $G=(V,E)$ and our goal is to estimate the \emph{size} of the largest acyclic subgraph in $G$, where the size of the subgraph is measured by its number of 
edges. 

One can approximate this problem to within a factor of $2$ in $O(\log{n})$ space (by counting the number of edges and dividing by two). Under the Unique Games Conjecture~\cite{GuruswamiMR08}, this is the best approximation ratio possible for this problem 
for \emph{poly-time} algorithms but this does \emph{not} imply a space lower bound in the streaming setting. This problem was studied in the streaming setting by~\cite{GuruswamiVV17,GuruswamiT19} and it was shown in~\cite{GuruswamiVV17} 
that a reduction from $\BHH$ implies an $\Omega(\sqrt{n})$-space lower bound for better-than-$(7/8)$ approximation of this problem in single-pass streams. We note that some related problems in directed graphs 
have also been studied recently in~\cite{ChakrabartiG0V20}. 

We prove the following theorem for this problem, formalizing part $(iii)$ of~\Cref{res:implications}.  

\begin{theorem}\label{thm:mas}
Let $p \geq 1$ be an integer and $\eps \in (0,1)$ be a parameter. Any $p$-pass streaming algorithm that outputs a $(1+\eps)$-approximation to the number of edges in a largest acyclic subgraph of a directed graph with probability at least $2/3$ 
requires $\eps 2^{-O(p)}\cdot n^{1-g(\eps,p)}$ space where $g$ is the function in~\Cref{eq:g}. 
\end{theorem}

\begin{proof}
	Suppose $\alg$ is a $p$-pass $S$-space streaming algorithm for $(1+\eps)$-approximation of the maximum acyclic subgraph problem. 
	Consider any instance $G$ of $\OMC_{n,k}$ for even integers $n,k$ where $k \leq \frac{1}{4\,\eps}$ such that $n/k$ is also even
	and is sufficiently large. 
	
	Unlike the proofs of~\Cref{thm:max-cut} and~\Cref{thm:max-matching}, we do not need to introduce odd cycles in the graph $G$ (using~\Cref{lem:odd-cycle}). Instead we will \emph{direct} the graph $G$ as follows. 
	Let $L$ and $R$ be the bipartition of the bipartite graph $G$ which are \emph{known} to both players (by the proof of~\Cref{thm:mainformal}). The players obtain the directed graph $H$ by Alice directing all her edges from $L$ to $R$
	and Bob directing his from $R$ to $L$. The players then 
	run $\alg$ on the stream obtained by appending edges of Bob to Alice's edges. This can be done in $2p - 1$ rounds and $(2p -1)\cdot S$ communication.  
	
	When $\OMC(G)$ is a Yes-case, we can pick a subgraph of $H$ consisting of all but one edge and still have no cycle, hence the answer to the problem is $n-1$ in this case. 
	On the other hand, in the No-case of $\OMC(G)$, we should leave out one edge from each cycle of the graph in the subgraph and thus the answer is $n-n/k$. 
	By the choice of $k$, we obtain that the answer in the Yes-case is larger  than the No-case by a factor of $(1+\eps)$. Hence, Alice and Bob can use the output of $\alg$ to distinguish between the two cases. 
	The rest follows as before. 
\end{proof}

\subsection{Minimum Spanning Tree}\label{sec:mst}

In the minimum spanning tree problem, we are given an undirected graph $G=(V,E)$ and with weight function $w: E \rightarrow \set{1,2,\ldots,W}$ and our goal is to estimate the \emph{weight} of the minimum spanning tree (MST) in $G$. 

There is a simple $O(n)$ space algorithm for computing an \emph{exact} MST in a single pass~\cite{FeigenbaumKMSZ05}. More recently,~\cite{HuangP16} gave an $\Ot(W \cdot n^{1-\nicefrac{1.01\eps}{(W-1)}})$ space algorithm for this problem, 
which achieves a polynomial saving in the space for any constant $\eps,W > 0$. A reduction from $\BHH$ shows that single-pass $(1+\eps)$-approximation algorithms for this problem
requires $\Omega(n^{1-\nicefrac{4\eps}{(W-1)}})$ space~\cite{HuangP16}.

We prove the following theorem for this problem, formalizing part $(iv)$ of~\Cref{res:implications}.  

\begin{theorem}\label{thm:mst}
Let $p \geq 1$ be an integer and $\eps \in (0,1)$ be a parameter. Any $p$-pass streaming algorithm that outputs a $(1+\eps)$-approximation to the weight of a minimum spanning tree in a graph with maximum weight $W$ with probability at least $2/3$ 
requires $\eps W^{-1} 2^{-O(p)}\cdot n^{1-g(\eps/W,p)}$ space where $g$ is the function in~\Cref{eq:g} (notice that the first argument of $g$ is $\eps/W$ and not merely $\eps$). 
\end{theorem}

\begin{proof}
	Suppose $\alg$ is a $p$-pass $S$-space streaming algorithm for $(1+\eps)$-approximation of the minimum weight spanning tree problem. 
	Consider any instance $G$ of $\OMC_{n,k}$ for even integers $n,k$ where $k \leq \frac{W}{4\,\eps}$ such that $n/k$ is also even
	and is sufficiently large. 
	
	Given an instance $G$ of $\OMC_{n,k}$, Alice and Bob create a graph $H$  as follows: $H$ contains $G$ as a subgraph with weight one on all edges. Moreover, $H$ has one additional vertex $s$ which is connected to all vertices of $G$ with weight $W$ 
	except for one vertex $\hat{s} \in G$ where $s$ is connected to it with weight $1$ instead (we give all latter edges to Bob). 
	The players then run $\alg$ on the stream obtained by appending edges of Bob to Alice's edges. This can be done in $(2p - 1)$ rounds and $(2p -1)\cdot S$ communication.  
	
	When $\OMC(G)$ is a Yes-case, a minimum spanning tree of $G$ has weight $n$ by picking the $n-1$ edges of the Hamiltonian cycle plus the edge $(s,\hat{s})$, all of weight one. 
	On the other hand, in the No-case of $\OMC(G)$, we need to pick at least $n/k-1$ edges of weight $W$ to connect $s$ to any cycle of $G$ that does not contain $\hat{s}$, as those cycles are not connected in $G$; the remainder of 
	the edges can then be of weight $1$ again (recall that $H$ has $n+1$ vertices and so its MST needs $n$ edges). As such, in this case, the weight of MST is at least $n+(n/k-1) \cdot (W-1)$. 
	
	By the choice of $k$, we obtain that the answer in the Yes-case is smaller than  the No-case by a factor of $(1+\eps)$. Hence, Alice and Bob can use the output of $\alg$ to solve $\OMC$.
	The rest again follows as before (just notice that $k \approx W/\eps$ and not $1/\eps$ unlike the previous examples). 
\end{proof}

\subsection{Property Testing of Connectivity, Bipartiteness, and Cycle-freeness}\label{sec:pt}

Given a parameter $\eps \in (0,1)$, an $\eps$-property tester algorithm for a property $P$ needs to decide whether $G$ has the property $P$ or is \emph{$\eps$-far} from having $P$, defined as follows: 

\begin{itemize}
	\item \textbf{Connectivity:} A graph is $\eps$-far from being connected if we need to insert $\eps \cdot n$ more edges to make it connected;
	\item \textbf{Bipartiteness:} A graph is $\eps$-far from being bipartite if we need to delete $\eps \cdot n$ of its edges to make it bipartite;
	\item \textbf{Cycle-freeness:} A graph is $\eps$-far from being cycle-free if we need to delete $\eps \cdot n$ of its edges to remove all its cycle.  
\end{itemize}

Traditionally, these problems have been studied extensively in the query complexity model and for sublinear-time algorithms. However, starting from the pioneering work of~\cite{HuangP16}, 
these problems have been receiving increasing attention in the streaming literature as well~\cite{HuangP16,MonemizadehMPS17,PengS18,CzumajFPS19}. 
In particular,~\cite{HuangP16} gave single-pass streaming algorithms for these problems with $\Ot(n^{1-\mathrm{poly}(\eps)})$ space (for bipartiteness testing, the input graph is assumed to be planar). 
Moreover, a reduction from \BHH implies that all these problems require $n^{1-O(\eps)}$ space in a single-pass~\cite{HuangP16}. 

We prove the following theorem for streaming property testing, formalizing part $(iv)$ of~\Cref{res:implications}.  

\begin{theorem}\label{thm:pt}
Let $p \geq 1$ be an integer and $\eps \in (0,1)$ be a parameter. Any $p$-pass streaming algorithm for $\eps$-property testing of connectivity, bipartiteness, or cycle-freeness in \emph{planar} graphs, with probability at least $2/3$ 
requires $\eps \cdot 2^{-O(p)}\cdot n^{1-g(\eps,p)}$ space where $g$ is the function in~\Cref{eq:g}.
\end{theorem}
\begin{proof}
	The proofs of these parts are all simple corollaries of the above reductions from $\OMC$. 
	
	For connectivity, we already discussed the reduction in~\Cref{sec:results}. Basically, for $k \leq \frac{1}{2\eps}$, the Yes-cases of $\OMC_{n,k}$ are Hamiltonian cycles and hence connected, while the No-cases are disjoint-union of $> \eps \cdot n$ 
	cycles and hence need at least $\eps \cdot n$ edge-insertions to become connected. Thus, property testing algorithms for connectivity can also solve $\OMC_{n,k}$. 
	
	For bipartiteness, we rely on our proof of~\Cref{thm:max-cut}. The lower bound for MAX-CUT involved either bipartite graphs, or graphs where any cut misses $\eps \cdot n$ edges of the graph -- this makes the latter graphs $\eps$-far from being 
	bipartite and the proof follows similar to~\Cref{thm:max-cut}. 
	
	For cycle-freeness, we perform another reduction from $\OMC_{n,k}$ for $k \leq \frac{1}{2\eps}$. Given an instance $G$ of $\OMC_{n,k}$, one of the players, say, Alice for concreteness, remove any one arbitrary edge from her input
	to obtain the graph $H$. Now, in the Yes-case of $\OMC$, $H$ is a cycle-free graph as we removed one edge from a Hamiltonian cycle, while in the No-case, $H$ still consists of $n/k-1 > \eps \cdot n$ vertex-disjoint cycle is thus $\eps$-far
	from being cycle-free. The lower bound follows. 
\end{proof}

\subsection{Matrix Rank and Schatten Norms}\label{sec:rank}

For any $q \geq 0$, the Schatten $q$-norm of an $n$-by-$n$ matrix $A$ is the $\ell_q$-norm of the vector of the singular values, i.e., 
$\paren{\sum_{i=1}^{n} \sigma_i(A)^q}^{\nicefrac{1}{q}}$ where $\sigma_1(A),\ldots,\sigma_n(A)$ are the singular values of $A$ (we used the notation $q$ instead of $p$ in expressing the norms
to avoid ambiguity with number of passes of the streaming algorithm). In particular, the case of $q=0$ corresponds to the rank of matrix $A$, $q=1$ is the nuclear (or trace) norm, and $q=2$ corresponds to the Frobenius norm. 

Streaming algorithms  for computing Schatten $q$-norms of matrices (particularly \emph{sparse} matrices) whose entries are updated in a stream have gain considerable attention lately~\cite{LiW16,ClarksonW09,BuryS15,AssadiKL17,BravermanCKLWY18,BravermmanKKS19}. As this topic is somewhat beyond the scope of our paper on graph streams, we refer the interested 
reader to~\cite{LiW16,BravermanCKLWY18,BravermmanKKS19} for an overview of single-pass and multi-pass streaming algorithms for this problem. But we note that reductions from $\BHH$ have been 
used to establish $n^{1-O(\eps)}$ space lower bounds for single-pass algorithms that $(1+\eps)$-approximate rank~\cite{BuryS15}, or any value of $q \in [0,+\infty)$ which is \emph{not} an even integer~\cite{LiW16,BravermanCKLWY18}. We remark
that~\cite{LiW16,BravermanCKLWY18} also have $\Omega(n^{1-4/q})$ space lower bounds for \emph{even} values of $q$ but this one requires a reduction from set disjointness and not $\BHH$ 
and is weaker than the bounds obtained when $q$ is not an even integer (similar to all the other $\BHH$ lower bounds, extending the lower bounds for values of $q$ which are not even integers
to multi-pass algorithms was left open in~\cite{LiW16,BravermanCKLWY18}). 


\begin{theorem}\label{thm:schatten}
Let $q \in [0,+\infty) \setminus 2\mathbb{Z}$, $p \geq 1$ be an integer and $\eps \in (0,1)$ be a parameter. Any $p$-pass streaming algorithm for $(1+\eps)$-approximation of Schatten $q$-norms in \emph{sparse} matrices, 
with probability at least $2/3$ requires $\eps \cdot 2^{-O(p)}\cdot n^{1-g(\eps,p)}$ space where $g$ is the function in~\Cref{eq:g}. 
\end{theorem}
\begin{proof}
	The proof for matrix rank, i.e., the $q=0$ case, follows from~\Cref{thm:max-matching}, and the standard equivalence between estimating matching size and computing
	the rank of the Tutte matrix~\cite{Tutte47}  with randomly chosen entries established in~\cite{Lovasz79} (see~\cite{BuryS15} for this reduction in the streaming model).
	
	For other values of $q$, we rely on the following lemma established by~\cite[Lemma~5.2]{BravermanCKLWY18}. 
	
	\begin{lemma}[Lemma~5.2 in~\cite{BravermanCKLWY18}]\label{lem:above}
	Suppose that $t \geq 2$ is an integer and $q \in (0,+\infty) \setminus 2\mathbb{Z}$. Let $G$ be an undirected $2$-regular graph 
	consisting of either $(i)$ vertex-disjoint $(2t+1)$-cycles or $(ii)$ vertex-disjoint $(4t+2)$-cycles. Then, the Schatten $q$-norm of the Laplacian matrix $L_G$ associated with
	$G$ differs by a constant factor depending only on $t$ and $q$ between the two cases.
	\end{lemma}
	
	Note that two possible cases of the graph $G$ in~\Cref{lem:above} correspond to the cases of the $k$-vs-$2k$ communication (which is a  ``harder'' variant of $\OMC_{n,k}$) instead of directly mapping to $\OMC_{n,k}$. 
	As such, for this proof, we instead use our more general communication lower bound that holds for distinguishing any two profiles of $k$-cycles (even when $k$ is odd), namely, \Cref{rem:arb_mat}.\footnote{It seems very 
	plausible to show that the Laplacian of the graphs in the two cases of $\OMC_{n,k}$ also differ by a constant factor using a similar argument to~\cite{BravermanCKLWY18} and so one could still use $\OMC$ directly; 
	however, we did not pursue this direction in the paper.} As the players in the communication problem can easily form the Laplacian of their input graph and run the Schatten norm streaming algorithm, the lower bound follows as before. 
\end{proof}

\subsection{Sorting-by-Block-Interchange}\label{sec:string} 

Given a string $x = (x_1,\ldots,x_n)$ representing a permutation on $n$ elements (i.e., $x_1,\ldots,x_n$ are distinct integers in $[n]$), we define the {\em block-interchange} operation on $x$, denoted by $r(i,j,k,l)$ where $1 \leq i \leq j < k \leq l \leq n$, as transforming $x$ to the string 
\[
x' = (x_{[1,i-1]},x_{[k,l]},x_{[j+1,k-1]},x_{[i,j]},x_{[\ell+1,n]}),
\]
 where $x_{[a,b]} = x_a,x_{a+1},\ldots,x_b$ (basically, $x'$ is obtained from $x$ by replacing the position of two blocks of sub-strings of $x$ with each other). 
We define the value of the {\em sorting-by-block-interchange} function on input $x$ to be the minimum number of block-interchanges that are required to sort $x$ in an increasing order (i.e., get $(1,2,\ldots,n)$).

In the streaming model, the entries of $x$ are arriving one by one in the stream and our goal is to output the value of {sorting-by-block-interchange} for $x$. This problem is one of the several 
string processing problems studied by~\cite{VerbinY11} who proved an $n^{1-O(\eps)}$ space lower bound for it using a reduction from $\BHH$. 

We can now prove the following lower bound for this problem. 

\begin{theorem}\label{thm:SBI}
Let $p \geq 1$ be an integer and $\eps \in (0,1)$ be a parameter. Any $p$-pass streaming algorithm that outputs a $(1+\eps)$-approximation of the value the sorting-by-block-interchange for a length-$n$
string with probability at least $2/3$ requires $\eps \cdot 2^{-O(p)}\cdot n^{1-g(\eps,p)}$ space where $g$ is the function in~\Cref{eq:g}.
\end{theorem}

\begin{proof}
	The proof of this theorem follows directly from~\cite[Theorem 2.2.]{VerbinY11} who gave a reduction from the $k$-vs-$2k$ cycle problem on graphs with $4n$ vertices and $k = O(1/\eps)$, 
	and our lower bound for the $k$-vs-$2k$ cycle problem in~\Cref{rem:arb_mat} exactly as in~\Cref{thm:schatten}. 
\end{proof}

\medskip

This concludes the list of all the lower bounds mentioned in~\Cref{res:implications}.


\clearpage

\bibliographystyle{abbrv}
\bibliography{general}

\clearpage
\appendix

\part*{Appendix}

\newcommand{\rY}{\rv{Y}}
\newcommand{\rU}{\rv{U}}
\newcommand{\unif}{\mathcal{U}}

\newcommand{\rM}{\rv{M}}
\newcommand{\rT}{\rv{T}}

\section{Tools from Information Theory}\label{app:info}

We briefly list the definitions and tools from information theory that we use in this paper. 

\begin{itemize}
\item We denote the \emph{Shannon Entropy} of a random variable $\rA$ by
$\en{\rA}$, which is defined as: 
\begin{align}
	\en{\rA} := \sum_{A \in \supp{\rA}} \mu(A) \cdot \log{\paren{1/\mu(A)}}. \label{eq:entropy}
\end{align} 

\item The \emph{conditional entropy} of $\rA$ conditioned on $\rB$ is denoted by $\en{\rA \mid \rB}$ and defined as:
\begin{align}
\en{\rA \mid \rB} := \Ex_{B \sim \rB} \bracket{\en{\rA \mid \rB = B}}, \label{eq:cond-entropy}
\end{align}
where 
$\en{\rA \mid \rB = B}$ is defined in a standard way by using the distribution of $\rA$ conditioned on the event $\rB = B$ in \Cref{eq:entropy}.

\item The \emph{mutual information} of $\rA$ and $\rB$ is denoted by
$\mi{\rA}{\rB}$ and is defined as:
\begin{align}
\mi{\rA}{\rB} := \en{\rA} - \en{\rA \mid  \rB} = \en{\rB} - \en{\rB \mid  \rA}. \label{eq:mi}
\end{align}
\noindent
\item The \emph{conditional mutual information} $\mi{\rA}{\rB \mid \rC}$ is defined similarly as: 
\begin{align}
\mi{\rA}{\rB \mid \rC} := \en{\rA \mid \rC} - \en{\rA \mid  \rB,\rC} = \en{\rB} - \en{\rB \mid  \rA,\rC}. \label{eq:cond-mi}
\end{align}
\end{itemize}


\subsection*{Standard Properties of Entropy and Mutual Information}\label{prelim-sec:prop-en-mi}

We shall use the following basic properties of entropy and mutual information throughout. Proofs of these properties mostly follow from convexity of the entropy function
and Jensen's inequality and can be found in~\cite[Chapter~2]{CoverT06}. 

\begin{fact}\label{fact:it-facts}
  Let $\rA$, $\rB$, $\rC$, and $\rD$ be four (possibly correlated) random variables.
   \begin{enumerate}
  \item \label{part:uniform} $0 \leq \en{\rA} \leq \log{\card{\supp{\rA}}}$. The right equality holds
    iff $\distribution{\rA}$ is uniform.
  \item \label{part:info-zero} $\mi{\rA}{\rB} \geq 0$. The equality holds iff $\rA$ and
    $\rB$ are \emph{independent}.
  \item \label{part:cond-reduce} \emph{Conditioning on a random variable reduces entropy}:
    $\en{\rA \mid \rB,\rC} \leq \en{\rA \mid  \rB}$.  The equality holds iff $\rA \perp \rC \mid \rB$.
    \item \label{part:sub-additivity} \emph{Subadditivity of entropy}: $\en{\rA,\rB \mid \rC}
    \leq \en{\rA \mid \rC} + \en{\rB \mid  \rC}$.
   \item \label{part:ent-chain-rule} \emph{Chain rule for entropy}: $\en{\rA,\rB \mid \rC} = \en{\rA \mid \rC} + \en{\rB \mid \rC,\rA}$.
  \item \label{part:chain-rule} \emph{Chain rule for mutual information}: $\mi{\rA,\rB}{\rC \mid \rD} = \mi{\rA}{\rC \mid \rD} + \mi{\rB}{\rC \mid  \rA,\rD}$.
  \item \label{part:data-processing} \emph{Data processing inequality}: suppose $f(\rA)$ is a deterministic function of $\rA$, then 
  \[\mi{f(\rA)}{\rB \mid \rC} \leq \mi{\rA}{\rB \mid \rC}.\] 
   \end{enumerate}
\end{fact}

\noindent
We also use the following two standard propositions. 

\begin{proposition}\label{prelim-prop:info-increase}
  For random variables $\rA, \rB, \rC, \rD$, if $\rA \perp \rD \mid \rC$, then, 
  \[\mi{\rA}{\rB \mid \rC} \leq \mi{\rA}{\rB \mid  \rC,  \rD}.\]
\end{proposition}
 \begin{proof}
  Since $\rA$ and $\rD$ are independent conditioned on $\rC$, by
  \itfacts{cond-reduce}, $\HH(\rA \mid  \rC) = \HH(\rA \mid \rC, \rD)$ and $\HH(\rA \mid  \rC, \rB) \ge \HH(\rA \mid  \rC, \rB, \rD)$.  We have,
	 \begin{align*}
	  \mi{\rA}{\rB \mid  \rC} &= \HH(\rA \mid \rC) - \HH(\rA \mid \rC, \rB) = \HH(\rA \mid  \rC, \rD) - \HH(\rA \mid \rC, \rB) \\
	  &\leq \HH(\rA \mid \rC, \rD) - \HH(\rA \mid \rC, \rB, \rD) = \mi{\rA}{\rB \mid \rC, \rD}. \qed
	\end{align*} 
	
\end{proof}

\begin{proposition}\label{prelim-prop:info-decrease}
  For random variables $\rA, \rB, \rC,\rD$, if $ \rA \perp \rD \mid \rB,\rC$, then, 
  \[\mi{\rA}{\rB \mid \rC} \geq \mi{\rA}{\rB \mid \rC, \rD}.\]
\end{proposition}
 \begin{proof}
 Since $\rA \perp \rD \mid \rB,\rC$, by \itfacts{cond-reduce}, $\HH(\rA \mid \rB,\rC) = \HH(\rA \mid \rB,\rC,\rD)$. Moreover, since conditioning can only reduce the entropy (again by \itfacts{cond-reduce}), 
  \begin{align*}
 	\mi{\rA}{\rB \mid  \rC} &= \HH(\rA \mid \rC) - \HH(\rA \mid \rB,\rC) \geq \HH(\rA \mid \rD,\rC) - \HH(\rA \mid \rB,\rC) \\
	&= \HH(\rA \mid \rD,\rC) - \HH(\rA \mid \rB,\rC,\rD) = \mi{\rA}{\rB \mid \rC,\rD}. \qed
 \end{align*}
 
\end{proof}

\subsection*{Measures of Distance Between Distributions}\label{prelim-sec:prob-distance}

We use two main measures of distance (or divergence) between distributions, namely the \emph{Kullback-Leibler divergence} (KL-divergence) and the \emph{total variation distance}. 

\paragraph{KL-divergence.} For two distributions $\mu$ and $\nu$ over the same probability space, the \emph{Kullback-Leibler divergence} between $\mu$ and $\nu$ is denoted by $\kl{\mu}{\nu}$ and defined as: 
\begin{align}
\kl{\mu}{\nu}:= \Ex_{a \sim \mu}\Bracket{\log\frac{\mu(a)}{{\nu}(a)}}. \label{eq:kl}
\end{align}
We  have the following relation between mutual information and KL-divergence. 
\begin{fact}\label{prelim-fact:kl-info}
	For random variables $\rA,\rB,\rC \sim \mu$, 
	\[\mi{\rA}{\rB \mid \rC} = \Ex_{(b,c) \sim {(\rB,\rC)}}\Bracket{ \kl{\distributions{\rA \mid \rB=b,\rC=c}{\mu}}{\distributions{\rA \mid \rC=c}{\mu}}}.\] 
\end{fact}
\noindent
One can write the entropy of a random variable as its KL-divergence from the uniform distribution. 
\begin{fact}\label{fact:kl-en}
	Let $\rA$ be any random variable and $\unif$ be the uniform distribution on $\supp{\rA}$. Then, 
	\[
	\en{\rA} = \log{\card{\supp{\rA}}} - \kl{\rA}{\unif}.
	\] 
\end{fact}
\noindent
Finally, we also have the following property. 
\begin{fact}\label{fact:kl-event}
	Let $\rA$ be any random variable and $E$ be any event in the same probability space. Then, 
	\[
	\kl{\rA \mid E}{\rA} \leq \log{\frac{1}{\Pr\paren{E}}}. 
	\] 
\end{fact}

\paragraph{Total variation distance.} We denote the total variation distance between two distributions $\mu$ and $\nu$ on the same 
support $\Omega$ by $\tvd{\mu}{\nu}$, defined as: 
\begin{align}
\tvd{\mu}{\nu}:= \max_{\Omega' \subseteq \Omega} \paren{\mu(\Omega')-\nu(\Omega')} = \frac{1}{2} \cdot \sum_{x \in \Omega} \card{\mu(x) - \nu(x)}.  \label{prelim-eq:tvd}
\end{align}

We use the following basic properties of total variation distance. 
\begin{fact}\label{prelim-fact:tvd-small}
	Suppose $\mu$ and $\nu$ are two distributions for $\event$, then, 
	$
	{\mu}(\event) \leq {\nu}(\event) + \tvd{\mu}{\nu}.
$
\end{fact}

We also have the following bound on the total variation distance of joint variables. 
\begin{fact}\label{fact_tvd_chain_rule}
  For any two distributions $p$ and $q$ of an $n$-tuple $(X_1,X_2,\ldots,X_n)$, we must have
  \begin{align*}
    &\, \tvd{\distributions{X_1,\ldots,X_n}{p}}{\distributions{X_1,\ldots,X_n}{q}} \\
    &\leq \sum_{i=1}^n \E_{(X_1,\ldots,X_{i-1})\sim p} \tvd{\distributions{X_i\mid X_1,\ldots,X_{i-1}}{p}}{\distributions{X_i\mid X_1,\ldots,X_{i-1}}{q}}.
  \end{align*}
\end{fact}

Finally, the following Pinskers' inequality bounds the total variation distance between two distributions based on their KL-divergence, 

\begin{fact}[Pinsker's inequality]\label{prelim-fact:pinskers}
	For any distributions $\mu$ and $\nu$, 
	$
	\tvd{\mu}{\nu} \leq \sqrt{\frac{1}{2} \cdot \kl{\mu}{\nu}}.
	$ 
\end{fact}

\subsection{Auxiliary Information Theory Lemmas}\label{sec:aux-lem}

This section includes some additional information theory lemmas that we prove in this paper. 

\paragraph{Total variation distance of convex combinations of distributions.} The following lemma is a consequence of the triangle inequality for the total variation distance.

\begin{lemma}
\label{lemma:tvdtriangle}
Let $\Omega$ and $I$ be non-empty sets and for all $i \in I$, distributions $\mu_i$ and $\nu_i$ supported on $\Omega$ be given. For every distribution $\mathcal{I}$ supported on $I$, we have
\[
\tvd{  \E_{i \sim \mathcal{I}}\Bracket{  \mu_i  }  }{  \E_{i \sim \mathcal{I}}\Bracket{  \nu_i  }  } \leq \E_{i \sim \mathcal{I}}\Bracket{  \tvd{  \mu_i  }{  \nu_i  }  }. 
\]
\end{lemma}
\begin{proof}
We have:
\begin{align*}
\tvd{  \E_{i \sim \mathcal{I}}\Bracket{  \mu_i  }  }{  \E_{i \sim \mathcal{I}}\Bracket{  \nu_i  }  } &= \frac{1}{2} \cdot \sum_{x \in \Omega} \card{ \E_{i \sim \mathcal{I}}\Bracket{  \mu_i(x)  } - \E_{i \sim \mathcal{I}}\Bracket{  \nu_i(x)  } } \tag{\Cref{prelim-eq:tvd}} \\
&\leq  \frac{1}{2} \cdot \sum_{x \in \Omega} \E_{i \sim \mathcal{I}}\Bracket{ \card{   \mu_i(x)  -  \nu_i(x)  } } \tag{Triangle inequality} \\
&\leq   \E_{i \sim \mathcal{I}}\Bracket{\frac{1}{2} \cdot \sum_{x \in \Omega} \card{   \mu_i(x)  -  \nu_i(x)  } } \tag{Linearity of expectation} \\
&\leq \E_{i \sim \mathcal{I}}\Bracket{  \tvd{  \mu_i  }{  \nu_i  }  } \tag{\Cref{prelim-eq:tvd}} .
\end{align*}
\end{proof}

\paragraph{Entropy and $\ell_2$-norm.} We relate the entropy and $\ell_2$-norm  using the following two lemmas. 
\begin{lemma}\label{lem_l2_entropy-1}
   For any random variable $\rA$ with support size $m > 16$,
    \[
      \norm{{\rA}}_2^2+\en{\rA} \leq \frac{1}{m}+\log m,
    \]
  \end{lemma}
\begin{proof}
	Let $\supp{\rA} = (a_1,\ldots,a_m)$ and for brevity let $\mu := \distribution{\rA}$ and $\unif$ as the uniform distribution on $\supp{\rA}$. 
	Suppose first that there is an element $a_i$ with $\mu(a_i) > 1/2$. Let $\rv{\Theta}$ be the indicator variable for the event $\rA = a_i$. In this case, 
	\begin{align*}
		\norm{{\rA}}_2^2+\en{\rA} &\leq 1+\en{\rA,\rv{\Theta}} \tag{as $\norm{\rA}_2^2 \leq 1$ since $\rA$ is a probability distribution} \\
		&\leq 1 + \en{\rv{\Theta}} + \frac{1}{2} \cdot \en{\rA \mid \rv{\Theta} = 0} \tag{by chain rule of entropy and since conditioned on $\rv{\Theta}=1$, $\rA$ has no entropy} \\
		&\leq 2 + \frac{1}{2} \cdot \log{m} < \log{m}.  \tag{as $\en{\rv{\Theta}} \leq 1$,~\itfacts{uniform}, and since $m > 16$}
	\end{align*}
	As such, in this case, the first inequality holds trivially. 
	
	Let us now consider the more interesting case when $\mu(a_i) \leq 1/2$ for all $a_i \in \supp{\rA}$. We have, 
	\begin{align*}
		\log{m}-\en{\rA} &= \kl{\rA}{\unif} \tag{by~\Cref{fact:kl-en}} \\
		&\geq \frac{1}{2} \cdot \norm{\rA-\unif}^2_1 \tag{by Pinsker's inequality (\Cref{prelim-fact:pinskers}) and since $\tvd{\rA}{\unif} = \norm{\rA-\unif}_1/2$} \\
		&\geq \norm{\rA-\unif}^2_2 \tag{as $\card{\mu(a_i)-\unif(a_i)} \leq 1/2$ for all $i$ in this case and $\norm{x}^2_2 \leq \frac{1}{2} \norm{x}^2_1$ whenever all $x_i \leq 1/2$} \\
		&= \norm{\rA}^2_2 - \frac{1}{m} \tag{as $\norm{x-y}_2^2 = \norm{x}^2 - 2\norm{xy}_2 + \norm{y}_2^2$ for all $x,y$}. 
	\end{align*}
	Rearranging the terms proves the lemma. 
\end{proof}
\begin{lemma}\label{lem_l2_entropy-2}
    For any random variable $\rA$ with support size $m > 16$,
    \[
      \en{\rA} \geq \log{m}-(\log{e})\cdot m \cdot \norm{\rA}_2^2 +\log e.
    \]
\end{lemma}
\begin{proof}
	 We use the same notation as in the previous proof.  We have, 
  \begin{align*}
    \en{\rA}&= \log m- \kl{\rA}{\unif} \tag{by~\Cref{fact:kl-en}} \\
    &= \log{m} -\sum_i \mu(a_i) \cdot \log{(m \cdot \mu(a_i))} \tag{by the definition of the KL-divergence in~\Cref{eq:kl}} \\
    &\geq \log m-\sum_i \mu(a_i) \cdot \log(\exp(m \cdot \mu(a_i)-1)) \tag{as $x+1 \leq e^x$ for all $x$}\\
    &= \log m-\sum_i \mu(a_i) \cdot (m \cdot \mu(a_i)-1) \cdot \log{e} \\ 
    &= \log m-(\log e)\cdot m\sum_i \mu(a_i)^2+\log e \\
    &= \log m-(\log e)\cdot m \cdot \norm{\rA}_2^2+\log e,
  \end{align*}
  concluding the proof. 
\end{proof}

\paragraph{High entropy random permutations.} The following lemma shows that if we have a ``high entropy'' random permutation of $[m]$ with entropy just $\approx  \eps \cdot m$ below the full entropy, 
then the entropy of an \emph{average} coordinate is very close to $\log{m}-\sqrt{\eps}$ (as opposed to $(\log{m!})/m-\eps \leq \log{m}-1$ which follows trivially from the sub-additivity of entropy). 

\begin{lemma}\label{lem_info_perm}
	Let $\rM:[m]\to[m]$ be a random permutation on $[m]$.
	If $\HH(\rv M)\geq \log m!-m/8$, then 
	\[
		m\log m-\sum_{j=1}^m \HH(\rv M(j))\leq 4\sqrt{(\log m!-\HH(\rv M))m}+3.
	\]
\end{lemma}
\begin{proof}
\newcommand{\oT}{\overline{T}}
	Let $1\leq \ell \leq m/2$ be an integer parameter to be fixed later.
	Let $\rv T\subset [m]$ be a \emph{uniformly random} subset of size $\ell$, chosen independent of $\rM$.
	Then, 
	\begin{align*}
		\en{\rM}&= \en{\rM \mid \rT} \tag{as $\rT \perp \rM$ and by~\itfacts{cond-reduce}}\\
		&=\HH(\rv M_{\rv T}\mid \rv T)+\HH(\rv M_{\rv\oT}\mid \rv M_{\rv T}, \rv T) \tag{$\rv M_{\rv T}$ denotes $\rv M(i)$ for all $i\in \rv T$, and $\rv \oT$ is the complement of $\rv T$}\\
		&=\E_T [\HH(\rv M_{\rv T}\mid \rv T=T)]+ \HH(\rv M_{\rv\oT}\mid \rv M_{\rv T}, \rv T) \\
		&\leq \E_T \left[\sum_{i\in T} \HH(\rv M(i)\mid \rv T=T)\right]+ \HH(\rv M_{\rv\oT}\mid \rv M_{\rv T}, \rv T)\tag{by sub-additivity of entropy~\itfacts{sub-additivity}} \\
		&\leq \E_T \left[\sum_{i\in T} \HH(\rv M(i)\mid \rv T=T)\right]+\log (m-\ell)! \tag{as conditioned on $\rT$, $\rM_{\rv\oT}$ is a permutation on $m-\ell$ elements and by~\itfacts{uniform}} \\
		&=\frac{\ell}{m}\sum_{i=1}^m \HH(\rv M(i))+\log (m-\ell)!\,. \tag{by the choice of $\rT$} 
	\end{align*}
	Therefore, 
	\begin{align*}
		\text{LHS in~\Cref{lem_info_perm}} &= m\log m-\sum_{i=1}^m \HH(\rv M(i)) \\
		&\leq m\log m+ \frac{m}{\ell}\cdot (\log (m-\ell)!-\HH(\rv M)) \tag{using the above inequality and by re-arranging the terms} \\
		&=m\log m+\frac{m}{\ell}\cdot (\log m!-\HH(\rv M))-\frac{m}{\ell}\cdot (\log m(m-1)\cdots (m-\ell+1)) \\
		&\leq m\log m+\frac{m}{\ell}\cdot (\log m!-\HH(\rv M))-\frac{m\log (m-\ell)^\ell}{\ell}\\
		&=m\log \left(1+\frac{\ell}{m-\ell}\right)+\frac{m}{\ell}\cdot (\log m!-\HH(\rv M)) \\
		&\leq \frac{(\log e) \cdot m \cdot \ell}{m-\ell}+\frac{m}{\ell}\cdot (\log m!-\HH(\rv M)).
	\end{align*}
	By setting $\ell:=\left\lceil\sqrt{(\log m!-\HH(\rv M))m}\right\rceil\leq m/2$, we have, 
	\begin{align*}
		\text{LHS in~\Cref{lem_info_perm}} &\leq \frac{(\log e) \cdot m \cdot \ell}{m-\ell}+\frac{m}{\ell}\cdot (\log m!-\HH(\rv M))\tag{by the above inequality} \\
		&\leq \frac{(\log e) \cdot m \cdot (\sqrt{(\log m!-\HH(\rv M))m}+1)}{m-m/2}+\sqrt{(\log m!-\HH(\rv M))m} \\
		&\leq 4\sqrt{(\log m!-\HH(\rv M))m}+3. \qed
	\end{align*}

\end{proof}

\paragraph{Weighted sum of entropy terms.} We also have the following technical lemma that bounds the weighted sum of entropy terms (some intuition on the statement is given right after the lemma). 

\begin{lemma}\label{lem_weighted_sum_entropy}
		Suppose $\rv M$ is a random variable with uniform distribution over $[t]$, and $\rA$ is a random variable with support size at most $\ell$ and distribution $\alpha(\rA)$.
		Then for any arbitrary function (not necessarily a distribution) $\beta: \supp{\rA} \rightarrow [\eps,1]$, we have:
		\[
			\frac{\sum_{A \in \supp{\rA}} \beta(A)^{-1} \cdot\alpha(A)^2\cdot \en{\rv M\mid \rv \rA=A}}{\sum_{A\in\supp{\rA}} \beta(A)^{-1} \cdot \alpha(A)^2} \geq \log{t}-\log{(\ell/\epsilon)}.
		\]
\end{lemma}
For a high level intuition of the lemma, consider the case when $\alpha = \beta$; in this case, the lemma 
simply says that conditioning on the random variable $\rA$ cannot reduce the entropy of $\rM$ from $\log{m}$ more than the entropy of $\rA$ which follows immediately from the chain rule (\itfacts{ent-chain-rule}). 
\begin{proof}[Proof of~\Cref{lem_weighted_sum_entropy}]
	As $\rM$ is the uniform distribution, for any choice of $a \in \supp{\rA}$, 
	\begin{align*}
		\en{\rM \mid \rv \rA=A} &= \log{t} - \kl{\rM\mid \rA=A}{\rM} \tag{by~\Cref{fact:kl-en}} \\
		&\geq \log{t} - \log{\frac{1}{\Pr\paren{\rA=A}}} \tag{by~\Cref{fact:kl-event}} \\
		&= \log{t} + \log{\alpha(A)}. 
	\end{align*}
	Therefore, 
	\begin{align*}
		\text{LHS of~\Cref{lem_weighted_sum_entropy}} &\geq \frac{\sum_{A \in \supp{\rA}} \beta(A)^{-1} \cdot\alpha(A)^2\cdot (\log{t}+\log{\alpha(A))}}{\sum_{A \in \supp{\rA}} \beta(A)^{-1} \cdot \alpha(A)^2} \\
		&= \log{t} + \frac{\sum_{A \in \supp{\rA}} \beta(A)^{-1} \cdot\alpha(A)^2\cdot \log{\alpha(A)}}{\sum_{A \in \supp{\rA}} \beta(A)^{-1} \cdot \alpha(A)^2},
	\end{align*}
	and thus we only need to show that 
	\begin{equation}\label{eqn_weighted}
		\sum_A \frac{\alpha(A)^2}{\beta(A)}\cdot\log ((\ell/\epsilon)\cdot\alpha(A))\geq 0.
	\end{equation}
	For each term, the derivative (w.r.t. $\alpha(A)$) is equal to: 
	\[
		\frac{\alpha(A)}{\beta(A)}\cdot(2\log((\ell/\epsilon) \cdot \alpha(A))+1).
	\]
	Note that
	\begin{itemize}
		\item when $\alpha(A)<\epsilon/\ell$, the derivative is at most $\alpha(A)/\beta(A)<1/\ell$;
		\item when $\alpha(A)>1/\ell$, the derivative is at least $\alpha(A)\cdot (2\log(1/\epsilon)+1)>1/\ell$.
	\end{itemize}
	
	If there is any $A \in \supp{\rA}$, where $\alpha(A)<\epsilon/\ell$, there must also be a $A'$ such that $\alpha(A')>1/\ell$ (since the average value of $\alpha(A)$ is at least $1/\ell$).
	Let $\delta=\min\set{\alpha(A')-1/\ell,\epsilon/\ell-\alpha(A)}$.
	
	By increasing $\alpha(A)$ by $\delta$ and decreasing $\alpha(A')$ by $\delta$, the left-hand-side of \eqref{eqn_weighted} could only decrease by the bounds on the derivative; moreover, $\alpha(\cdot)$ still remains a valid distribution.
	
	By doing such adjustments for at most $2\ell$ times, we obtain a new distribution $\hat{\alpha}$ such that $\hat{\alpha}(A)\geq \epsilon/\ell$ for all $A \in \supp{\rA}$, and throughout this process, we only decrease the LHS of \eqref{eqn_weighted}.
	On the other hand, for the distribution $\hat{\alpha}$, since $\hat{\alpha}(A)\geq \epsilon/\ell$, every term in the LHS of \eqref{eqn_weighted} is nonnegative and thus \eqref{eqn_weighted} holds.
		This implies that it must also hold for the $\alpha(\cdot)$ as well, proving the lemma.
	\end{proof}
	

\section{Further Background on Boolean Hidden Hypermatching}\label{app:background-bhh}

\subsection*{Origin of the Problem}

The Boolean Hidden Hypermatching $\BHH$ problem of Verbin and Yu~\cite{VerbinY11} is inspired by the Boolean Hidden Matching problem of~Gavinsky~\etal\cite{GavinskyKKRW07} (that corresponds to $\BHH$ with $t=2$), which itself is a 
Boolean version of the Hidden Matching problem of Bar-Yossef~\etal~\cite{Bar-YossefJK04}. The original motivation behind the latter two works was proving separations between quantum versus one-way randomized communication complexity. 
However, it was observed in the work of Verbin and Yu that these problems are also highly applicable for proving streaming lower bounds. Thus a main conceptual contribution of~\cite{VerbinY11} was the introduction of 
$\BHH$  for proving streaming lower bounds.

 It is worth mentioning that stronger separations between quantum versus \emph{two-way} randomized communication complexity 
has also since been proven, for instance for the Shifted Approximate Equality (Shape) problem of Gavinsky~\cite{Gavinsky16}. While these problems do not suffer from the weakness of $\BHH$ which is only tailored 
to one-way protocols, to our knowledge, they have also not found applications for proving streaming lower bounds. This is primarily due to the fact that such problems are ``much harder'' than most streaming problems of interest and hence not suitable for the purpose of a reduction. Indeed, most
graph streaming problems such as gap cycle counting, MAX-CUT, and property testing of connectivity or bipartiteness, tend to become ``easy'' with large number of passes over the input (only as a function of $\eps$); such tradeoffs provably cannot be 
captured by communication problems that are hard even in the two-way model (see, e.g.~\cite{AssadiCK19b}).

\subsection*{Other Variants} 

There are multiple other variants of $\BHH$ studied in the literature depending on the application. 

Kapralov, Khanna, and Sudan~\cite{KapralovKS15} introduced the Boolean Hidden Partition problem in which the input to Bob, instead of a hypermatching, 
is a random sparse graph with no short cycles. This problem was then used by~\cite{KapralovKS15} to prove an $\Omega(\sqrt{n})$ space lower bound for better-than-$2$ approximation of MAX-CUT. Inspired by this problem, 
Kapralov, Khanna, Sudan, and Velingker~\cite{KapralovKSV17} 
introduced the Implicit Hidden Partition problem, which is a multi-party communication problem and its definition is rather different from the setup of $\BHH$. Communication complexity of this problem was first analyzed in~\cite{KapralovKSV17} for three players 
which gave an $\Omega(n)$ space lower bound for $(1+\Omega(1))$-approximation of MAX-CUT. Subsequently, Kapralov and Krachun~\cite{KapralovK19} proved a communication lower bound for this problem for arbitrary number of players which culminated in 
the {optimal} lower bound of $\Omega(n)$ space for better-than-$2$ approximation of MAX-CUT. 

Yet another variant of $\BHH$ was introduced by Guruswami and Tao~\cite{GuruswamiT19} as the $p$-ary Hidden Matching problem in which Alice, instead of a Boolean vector $x \in \set{0,1}^{n}$, is given a vector over a larger field $\mathbb{F}_p^n$ for some 
arbitrary prime $p$ and we use computation over $\mathbb{F}_p$ (instead of $\mathbb{F}_2$, i.e., XOR) to interpret the matching $Mx$.  This problem was then used in~\cite{GuruswamiT19} to prove lower bounds for streaming unique games and
other CSP-like problems. 

Finally, $\BHH$ and related problems have been studied in the multi-party simultaneous communication model in which the input is partitioned across multiple players who, simultaneously with each other, send a single message each to 
a central coordinator that outputs the final answer~\cite{AssadiKL17,KallaugherKP18}. 

\subsection*{Streaming Lower Bounds} 

We now list several key implications of $\BHH$ for proving streaming lower bounds but note that this is not a comprehensive list of such results. 

\begin{itemize}[leftmargin=10pt]
	\item \textbf{Approximating MAX-CUT:} One can approximate MAX-CUT to within a factor of $2$ in $O(\log{n})$ space (by counting the number of edges) and to $(1+\eps)$ in 
	$\Ot(n/\eps^2)$ space (by maintaining a cut sparsifier); see, e.g.~\cite{KapralovKS15}. A direct reduction from $\BHH$ can prove that any $(1+\eps)$-approximation single-pass algorithm requires $n^{1-O(\eps)}$ space~\cite{KoganK15,KapralovKS15}, 
	and a series of work building on this approach~\cite{KapralovKS15,KapralovKSV17,KapralovK19} culminated in the optimal lower bound of $\Omega(n)$ space for less-than-$2$ approximation in~\cite{KapralovK19}.  
	
	\item \textbf{Streaming unique games and CSPs:} Motivated in parts by lower bounds for MAX-CUT, other constrained satisfaction problems (CSPs) have been studied in~\cite{GuruswamiVV17,GuruswamiT19}. 
	For instance,~\cite{GuruswamiVV17} proves a lower bound of $\Omega(\sqrt{n})$ space for better-than-$(7/8)$ approximation of maximum acyclic digraph using a reduction from $\BHH$ and~\cite{GuruswamiT19} 
	proves an $\Omega(\sqrt{n})$ space lower bound for streaming unique games using a reduction from $p$-ary Hidden Matching problem discussed above.

	\item \textbf{Estimating matching size:} There is an active line of work on estimating  matching size in graphs (in particular planar and low arboricity 
	graphs)~\cite{AssadiKL17,EsfandiariHLMO15,McGregorV18,CormodeJMM17,KapralovKS14}. 
	A reduction from $\BHH$ can prove an $\Omega(\sqrt{n})$ space for better-than-$(3/2)$ approximation in planar graphs~\cite{EsfandiariHLMO15} which was improved to $(1+\eps)$-approximation in 
	$n^{1-O(\eps)}$ space for low-arboricity graphs in~\cite{BuryS15}. These results were further extended to super-linear-in-$n$ lower bounds for dense graphs by~\cite{AssadiKL17}.

	\item \textbf{Estimating matrix rank and Schatten norms:} Using an elegant connection between matching size estimation and estimating rank of $n$-by-$n$ matrices (dating back to a work of Lov{\'a}sz~\cite{Lovasz79}), 
	the aforementioned lower bounds for matching size estimation were also extended to algorithms that can estimate rank of $n$-by-$n$ matrix in the streaming setting~\cite{BuryS15,AssadiKL17}.
	These results were further generalized to all Schatten $p$-norms of $n$-by-$n$ matrices for all values of $p \in [0,+\infty)$ which are \emph{not} even integers in~\cite{LiW16} using another reduction from $\BHH$ (see also~\cite{BravermanCKLWY18}
	for another variant of these results).

	\item \textbf{Streaming Property Testing:} Solving property testing problems in the streaming setting (as opposed to the more familiar query algorithms) has been receiving increasing attention 
	lately~\cite{HuangP16,MonemizadehMPS17,PengS18,CzumajFPS19} starting from the pioneering work of~\cite{HuangP16}. On this front, 
	$\BHH$ has been used to prove $n^{1-O(\eps)}$-space lower bounds for $\eps$-property testing of connectivity, cycle-freeness, and bipartiteness~\cite{HuangP16}. 
\end{itemize}

In addition, the $\BHH$ problem has also been used to prove lower bounds in other settings such as distribution testing~\cite{AndoniMN19}, 
distributed computing~\cite{FischerGO17}, property testing~\cite{BalcanLW019}, and sketching~\cite{KallaugherKP18}. We leave this as an interesting open direction to explore the implications of our multi-round lower bounds for $\OMC$ in these other models.

\end{document}